\documentclass[9pt]{article}
\usepackage{amssymb}
\usepackage{tikz}
\usetikzlibrary{arrows, calc, decorations.pathmorphing, backgrounds, positioning, fit, petri, automata}
\definecolor{grey1}{rgb}{0.5,0.5,0.5}
\usepackage[top=2.54cm, bottom=2.54cm, left=2.2cm, right=2.2cm]{geometry}
\usepackage{algorithmicx}
\usepackage{algcompatible}
\usepackage{algorithm}

\usepackage{hyperref}
\hypersetup{
  colorlinks,
  citecolor=blue,
  linkcolor=red,
  urlcolor=magenta}

\usepackage{subcaption}
\usepackage{booktabs}

\usepackage{mathtools}
\usepackage{graphicx}
\usepackage{amsmath,amsthm,bm,bbm}
\usepackage{rotating}
\usepackage{mathrsfs}
\usepackage{chngcntr}
\usepackage{array}
\usepackage{apptools}
\usepackage[titletoc,title]{appendix}
\usepackage{comment}

\usepackage{xcolor}         
\usepackage{setspace,lscape,longtable}
\usepackage{fullpage}  

\usepackage{lipsum}
\usepackage{enumitem}
\setlist[enumerate]{itemsep=0mm}
\setlist[itemize]{itemsep=0mm}

\usepackage{titlesec}
\usepackage[english]{babel}
\usepackage{xcolor}         
\definecolor{grau}{rgb}{0.8,0.8,0.8}

\newtheoremstyle{mystyle}
{1ex} 
{1ex} 
{\itshape} 
{} 
{\bfseries} 
{} 
{1ex} 
{} 
\theoremstyle{plain}

\newtheorem{theorem}{Theorem}
\newtheorem{lemma}{Lemma}
\newtheorem{corollary}{Corollary}

\theoremstyle{definition}

\newtheorem*{example}{Example}

\newtheorem{remark}{Remark}

\newtheorem{innercustomgeneric}{\customgenericname}
\providecommand{\customgenericname}{}
\newcommand{\newcustomtheorem}[2]{%
  \newenvironment{#1}[1]
  {%
   \renewcommand\customgenericname{#2}%
   \renewcommand\theinnercustomgeneric{##1}%
   \innercustomgeneric
  }
  {\endinnercustomgeneric}
}

\newcustomtheorem{customthm}{Theorem}
\newcustomtheorem{customlemma}{Lemma}

\usepackage[comma,authoryear]{natbib}

\bibpunct{(}{)}{;}{a}{,}{,}

\DeclareMathOperator*{\arginf}{arg\,inf}

\newcommand\undermat[2]{%
  \makebox[0pt][l]{$\smash{\underbrace{\phantom{%
    \begin{matrix}#2\end{matrix}}}_{\text{$#1$}}}$}#2}

\newcommand{\prob}{{\mathbb{P}}}

\newcommand{\expect}{\mathbb{E}}

\newcommand{\transpose}{^{\mathrm{T}}}

\newcommand{\bLambda}{{\bm{\Lambda}}}

\newcommand{\calA}{{\mathcal{A}}}
\newcommand{\calB}{{\mathcal{B}}}

\newcommand{\calF}{{\mathcal{F}}}
\newcommand{\calG}{{\mathcal{G}}}
\newcommand{\calH}{{\mathcal{H}}}

\newcommand{\calK}{{\mathcal{K}}}

\newcommand{\calN}{{\mathcal{N}}}

\newcommand{\calU}{{\mathcal{U}}}
\newcommand{\calV}{{\mathcal{V}}}

\newcommand{\bU}{{\mathbf{U}}}
\newcommand{\bP}{{\mathbf{P}}}

\newcommand{\bY}{{\mathbf{Y}}}
\newcommand{\bv}{{\mathbf{v}}}
\newcommand{\bx}{{\mathbf{x}}}
\newcommand{\bX}{{\mathbf{X}}}
\newcommand{\bA}{{\mathbf{A}}}
\newcommand{\bB}{{\mathbf{B}}}
\newcommand{\bC}{{\mathbf{C}}}
\newcommand{\bD}{{\mathbf{D}}}
\newcommand{\bE}{{\mathbf{E}}}
\newcommand{\bL}{{\mathbf{L}}}

\newcommand{\bQ}{{\mathbf{Q}}}
\newcommand{\bS}{{\mathbf{S}}}
\newcommand{\bW}{{\mathbf{W}}}

\newcommand{\bu}{{\mathbf{u}}}

\newcommand{\by}{{\mathbf{y}}}
\newcommand{\bz}{{\mathbf{z}}}

\newcommand{\bV}{{\mathbf{V}}}

\newcommand{\bbeta}{{\bm{\beta}}}

\newcommand{\bOmega}{{\bm{\Omega}}}

\newcommand{\bSigma}{{\bm{\Sigma}}}
\newcommand{\bxi}{{\boldsymbol{\xi}}}
\newcommand{\bzeta}{{\boldsymbol{\zeta}}}
\newcommand{\beps}{{\bm{\varepsilon}}}
\newcommand{\eye}{{\mathbf{I}}}

\newcommand{\bphi}{{\bm{\phi}}}

\newcommand{\zero}{{\bm{0}}}
\newcommand{\eps}{\epsilon}

\usepackage{lipsum}
\usepackage{enumitem}
\setlist[enumerate]{itemsep=0mm}
\setlist[itemize]{itemsep=0mm}

\author{Fangzheng Xie,
\thanks{Department of Applied Mathematics and Statistics, Johns Hopkins University} 
\and Yanxun Xu,
\footnotemark[1] \thanks{Correspondence should be addressed to Yanxun Xu (yanxun.xu@jhu.edu)}
\and Carey E. Priebe, 
\footnotemark[1]
\and Joshua Cape
\footnotemark[1]}
\title{\bf Bayesian Estimation of Sparse Spiked Covariance Matrices in High Dimensions}
\begin{document}
\allowdisplaybreaks

\maketitle

\begin{abstract}
We propose a Bayesian methodology for estimating spiked covariance matrices with jointly sparse structure in high dimensions. The spiked covariance matrix is reparametrized in terms of the latent factor model, where the loading matrix is equipped with a novel matrix spike-and-slab LASSO prior, which is a continuous shrinkage prior for modeling jointly sparse matrices. We establish the rate-optimal posterior contraction for the covariance matrix with respect to the operator norm as well as that for the principal subspace with respect to the projection operator norm loss. 
We also study the posterior contraction rate of the principal subspace with respect to the two-to-infinity norm loss, a novel loss function measuring the distance between subspaces that
is able to capture element-wise eigenvector perturbations. 
We show that the posterior contraction rate with respect to the two-to-infinity norm loss is tighter than that with respect to the routinely used projection operator norm loss under certain low-rank and bounded coherence conditions. 
In addition, a point estimator for the principal subspace is proposed with the rate-optimal risk bound with respect to the projection operator norm loss. These results are based on a collection of concentration and large deviation inequalities for the matrix spike-and-slab LASSO prior. The numerical performance of the proposed methodology is assessed through synthetic examples and the analysis of a real-world face data example. 
\end{abstract}

\noindent%

\section{Introduction} 
\label{sec:introduction}
In contemporary statistics, datasets are typically collected with high-dimensionality, 
where the dimension  $p$ can be significantly larger than the sample size $n$. For example, in genomics studies, the number of genes is typically much larger than the number of subjects \citep{cancer2012comprehensivelung}. In computer vision, the number of pixels in each image can be comparable to or exceed the number of images when the resolution of these images is relatively high \citep{927464,1407873}. When dealing with such high-dimensional datasets, covariance matrix estimation plays a central role in understanding the complex structure of the data and has 
 received significant attention in various contexts, including latent factor models \citep{bernardo2003bayesian,geweke1996measuring}, Gaussian graphical models \citep{liu2012,MAL-001}, etc. However, in the high-dimensional setting, additional structural assumptions are often necessary in order to address challenges associated with statistical inference \citep{doi:10.1198/jasa.2009.0121}.  
For example, sparsity is introduced for sparse covariance/precision matrix estimation \citep{cai2016,cai2012,doi:10.1093/biostatistics/kxm045}, and low-rank structure is enforced in spiked covariance matrix models \citep{cai2015optimal,johnstone2001distribution}. Readers can refer to \cite{cai2016} for a recent literature review.


In this paper we focus on the sparse spiked covariance matrix models under the Gaussian sampling distribution assumption. 
The spiked covariance matrix models, originally named in \cite{johnstone2001distribution}, is a class of models that can be described as follows: The observations $\by_1,\ldots,\by_n$ are independently collected from the $p$-dimensional mean-zero normal distribution with covariance matrix $\bSigma$ of the form 
\begin{align}\label{eqn:spiked_covariance}
\bSigma = \bU\bLambda\bU\transpose+\sigma^2\eye_p,
\end{align}
where $\bU$ is a $p\times r$ matrix with orthonormal columns, $\bLambda = \mathrm{diag}(\lambda_1,\cdots,\lambda_r)$ is an $r\times r$ diagonal matrix, and $r<p$. Since the spectrum of the covariance matrix is $\{\lambda_1+\sigma^2,\ldots,\lambda_r+\sigma^2,\sigma^2,\cdots,\sigma^2\}$ (in non-increasing order), there exists an eigen-gap $\lambda_r(\bSigma) - \lambda_{r+1}(\bSigma) = \lambda_r$, where $\lambda_r(\bSigma)$ denotes the $r$-th largest eigenvalue of $\bSigma$. Therefore the first $r$ leading eigenvalues of $\bSigma$ can be regarded as ``spike" or signal eigenvalues, and the remaining eigenvalues $\sigma^2$ may be treated as ``bulk" or  noise eigenvalues.  
Here we assume that the eigenvector matrix $\bU$ is jointly sparse,  the formal definition of which is deferred to Section \ref{sub:spiked_covariance_matrix_models}. Roughly speaking, joint sparsity refers to a significant amount of rows in $\bU$ being zero, which allows for  feature selection and brings 
easy interpretation in many applications. 
For example, in the analysis of face images, a classical method to extract common features among different face characteristics, expressions, illumination conditions, etc., is to obtain the eigenvectors of these face data, referred to as  eigenfaces. Each coordinate of these eigenvectors corresponds to a specific pixel in the image. Nonetheless, the number of pixels (features) is typically much larger than the number of images (samples), and it is often desirable to gain insights of the face information via a relatively small number of pixels,  referred to as key pixels. 
By introducing joint sparsity to these eigenvectors, one is able to conveniently model key pixels among multiple face images corresponding to non-zero rows of eigenvectors. A concrete real data example is provided in Section \ref{sub:real_data_example}.

The literature on sparse spiked covariance matrix estimation in high-dimensions from a frequentist perspective is quite rich. In \cite{doi:10.1198/jasa.2009.0121}, it is shown that the classical principal component analysis can fail when $p\gg n$. In \cite{cai2013sparse} and \cite{vu2013minimax}, the minimax estimation of the principal subspace (\emph{i.e.}, the linear subspace spanned by the eigenvector matrix $\bU$) with respect to the projection Frobenius norm loss under various sparsity structure on $\bU$ is considered, and \cite{cai2015optimal} provides minimax estimation procedures of the principal subspace with respect to the projection operator norm loss under the joint sparsity assumption. 

In contrast, there is comparatively limited literature on Bayesian estimation of sparse spiked covariance matrices providing theoretical guarantees. To the best of our knowledge, \cite{gao2015rate} and \cite{pati2014posterior} are the only two works in the literature addressing posterior contraction rates for Bayesian estimation of sparse spiked covariance matrix models. In particular, in \cite{pati2014posterior} the authors discuss the posterior contraction behavior of the  covariance matrix $\bSigma$ with respect to the operator norm loss under the Dirichlet-Laplace shrinkage prior \citep{doi:10.1080/01621459.2014.960967}, but the contraction rates are sub-optimal when the number of spikes $r$ grows with the sample size; In \cite{gao2015rate}, the authors propose a carefully designed prior on $\bU$ that yields rate-optimal posterior contraction of the principal subspace with respect to the projection Frobenius norm loss, but the tractability of computing the full posterior distribution is lost, except for the posterior mean as a point estimator. Neither  \cite{gao2015rate}   nor \cite{pati2014posterior} discusses the posterior contraction behavior for sparse spiked covariance matrix models when the eigenvector matrix $\bU$ exhibits joint sparsity. 

We propose a matrix spike-and-slab LASSO prior to model joint sparsity occurring in the eigenvector matrix $\bU$ of the spiked covariance matrix. 
The matrix spike-and-slab LASSO prior is a novel continuous shrinkage prior that generalizes 
 the classical spike-and-slab LASSO prior for vectors in \cite{rovckova2018bayesian} and \cite{rovckova2016spike} to jointly sparse rectangular matrices. 
One major contribution of this work is that under the matrix spike-and-slab LASSO prior, we establish the rate-optimal posterior contraction for the entire covariance matrix $\bSigma$ with respect to the operator norm loss as well as that for the principal subspace with respect to the projection operator norm loss. 
Furthermore, we also focus on the two-to-infinity norm loss, a novel loss function measuring the closeness between linear subspaces. As will be seen in Section \ref{sub:spiked_covariance_matrix_models}, the two-to-infinity norm loss is able to detect element-wise perturbations of the eigenvector matrix $\bU$ spanning the principal subspace. Under certain low-rank and bounded coherence conditions on $\bU$, we obtain a tighter posterior contraction rate for the principal subspace with respect to the two-to-infinity norm loss than that with respect to the routinely used projection operator norm loss. 
Besides the contraction of the full posterior distribution, the Bayes procedure also leads to a point estimator for the principal subspace with a rate-optimal risk bound. 
In addition to the convergence results \emph{per se}, 
we present a collection of concentration and large deviation inequalities for the matrix spike-and-slab LASSO prior that may be of independent interest. These technical results serve as the main tools for deriving the posterior contraction rates. 
Last but not least, unlike the prior proposed in \cite{gao2015rate}, the matrix spike-and-slab LASSO prior yields a tractable Metropolis-within-Gibbs sampler for posterior inference.

The rest of the paper is organized as follows. In Section \ref{sec:bayesian_sparse_spiked_covariance_matrix_model} we briefly review the background for the sparse spiked covariance matrix models and propose the matrix spike-and-slab LASSO prior. Section \ref{sec:theoretical_properties} elaborates on our theoretical contributions, including the concentration and large deviation inequalities for the matrix spike-and-slab LASSO prior and the posterior contraction results. The numerical performance of the proposed methodology is presented in Section \ref{sec:numerical_examples} through synthetic examples and the analysis of a real-world computer vision dataset. Further discussion is included in Section \ref{sec:discussion}.

\vspace*{1ex}
\noindent\textbf{Notations: }Let $p$ and $r$ be positive integers. 
We adopt the  shorthand notation $[p] = \{1,\ldots,p\}$. For any finite set $S$, we use $|S|$ to denote the cardinality of $S$. The symbols $\lesssim$ and $\gtrsim$ mean the inequality up to a universal constant, \emph{i.e.}, $a \lesssim b$ ($a \gtrsim b$, resp.) if $a\leq Cb$ ($a \geq Cb$) for some absolute constant $C > 0$. We write $a\asymp b$ if $a\lesssim b$ and $a\gtrsim b$.  
The $p\times r$ zero matrix is denoted by $\zero_{p\times r}$, and the $p$-dimensional zero column vector is denoted by $\zero_p$. When the dimension is clear, the zero matrix is simply denoted by $\zero$. The $p\times p$ identity matrix is denoted by $\eye_p$, and when the dimension is clear, is  denoted by $\eye$. 
An orthonormal $r$-frame in $\mathbb{R}^p$ is a $p\times r$ matrix $\bU$ with orthonormal columns, \emph{i.e.}, $\bU\transpose\bU = \eye_{r\times r}$. 
The set of all orthonormal $r$-frames in $\mathbb{R}^p$ is denoted by $\mathbb{O}(p, r)$. When $p = r$, we write $\mathbb{O}(r) = \mathbb{O}(r, r)$. 
For a $p$-dimensional vector $\bx\in\mathbb{R}^p$, we use $x_j$ to denote its $j$th component, $\|\bx\|_1 = \sum_{j=1}^p|x_j|$ to denote its $\ell_1$-norm, $\|\bx\|_2$ to denote its $\ell_2$-norm, and $\|\bx\|_\infty = \max_{j\in[p]}|x_j|$ to denote its $\ell_\infty$-norm. 
For a symmetric square matrix $\bSigma\in\mathbb{R}^{p\times p}$, we use $\lambda_k(\bSigma)$ to denote the $k$th-largest eigenvalue of $\bSigma$. For a matrix $\bA\in\mathbb{R}^{p\times r}$, we use $\bA_{j*}$ to denote the row vector formed by the $j$th row of $\bA$, $\bA_{*k}$ to denote the column vector formed by the $k$th column of $\bA$, 
 the lower case letter $a_{ij}$ to denote the $(i, j)$-th element of $\bA$, $\|\bA\|_{\mathrm{F}} = \sqrt{\sum_{j=1}^p\sum_{k=1}^ra_{jk}^2}$ to denote the Frobenius norm of $\bA$, $\|\bA\|_2 = \sqrt{\lambda_1(\bA\transpose\bA)}$ to denote the operator norm of $\bA$, $\|\bA\|_{2\to\infty} = \max_{\|\bx\|_2 = 1}\|\bA\bx\|_\infty$ to denote the two-to-infinity norm of $\bA$, and $\|\bA\|_\infty = \max_{\|\bx\|_\infty = 1}\|\bA\bx\|_\infty$ to denote the (matrix) infinity norm of $\bA$. The prior and posterior distributions appearing in this paper are denoted by $\Pi$, and the densities of $\Pi$ with respect to the underlying sigma-finite measure are denoted by $\pi$. 


\section{Sparse Bayesian spiked covariance matrix models} 
\label{sec:bayesian_sparse_spiked_covariance_matrix_model}

\subsection{Background} 
\label{sub:spiked_covariance_matrix_models}


In the spiked covariance matrix model \eqref{eqn:spiked_covariance}, the matrix 
$\bSigma$  is of the form $
\bSigma = \bU\bLambda\bU\transpose+\sigma^2\eye_p. $ 
We focus on the case where the leading $r$ eigenvectors of $\bSigma$ (the columns of $\bU$) are jointly sparse \citep{cai2015optimal, vu2013minimax}. Formally, the row support of $\bU$ is defined as
\[
\mathrm{supp}(\bU) = \left\{j\in[p]:\bU_{j*}\transpose\neq\zero_r\right\},
\]
and $\bU$ is said to be jointly $s$-sparse, if $|\mathrm{supp}(\bU)|\leq s$. Heuristically, this assumption asserts that the signal comes from at most $s$ features among all $p$ features. Geometrically, joint sparsity has the interpretation that at most $s$ coordinates of $\by_i$ generate the subspace $\mathrm{Span}\{\bU_{*1},\ldots,\bU_{*r}\}$ \citep{vu2013minimax}. Noted that $s\geq r$ due to the orthonormal constraint on the columns of $\bU$. 

This paper studies a Bayesian framework for estimating the covariance matrix $\bSigma$. We quantify how well the proposed methodology estimates the entire covariance matrix $\bSigma$ and the principal subspace $\mathrm{Span}\{\bU_{*1},\cdots,\bU_{*r}\}$
in the high-dimensional and jointly sparse setup.  
Leaving the Bayesian framework for a moment, we first introduce some necessary background. 
Throughout the paper, we write $\bSigma_0 = \bU_0\bLambda_0\bU_0\transpose + \sigma_0\eye_p$ to be the true covariance matrix that generates the data $\bY = [\by_1,\ldots,\by_n]\transpose$ from the $p$-dimensional multivariate Gaussian distribution $\mathrm{N}_p(\zero_p,\bSigma_0)$, where $\bLambda_0 = \mathrm{diag}(\lambda_{01},\cdots,\lambda_{0r})$. The parameter space of interest for $\bSigma$ is given by
\begin{align*}
\Theta(p, r, s)
&= \left\{\bSigma = \bU\bLambda\bU\transpose + \sigma^2\eye_p:\bU\in\mathbb{O}(p, r),|\mathrm{supp}(\bU)|\leq s,
\lambda_1\geq\ldots\geq\lambda_r>0\right\}.
\end{align*}
The following minimax rate of convergence for $\bSigma$ under the operator norm loss \cite{cai2015optimal} serves as a benchmark for measuring the performance of any estimation procedure for $\bSigma$. 
\begin{theorem}[\citealp{cai2015optimal}]
Let $1\leq r\leq s\leq p$. Suppose that $(s\log p)/n\to 0$ and $\lambda_{01}\geq\lambda_{0r}>0$ are bounded away from $0$ and $\infty$. Then the minimax rate of convergence for estimating $\bSigma\in\Theta(p, r, s)$ is 
\begin{align}\label{eqn:minimax_rate_Sigma}
\inf_{\widehat\bSigma}\sup_{\bSigma_0\in\Theta(p, r, s)}\expect_{\bSigma_0}\|\widehat\bSigma - \bSigma_0\|_2^2\asymp \frac{s\log p}{n}.
\end{align}
\end{theorem}
Estimation of the principal subspace $\mathrm{Span}\{\bU_{*1},\ldots,\bU_{*r}\}$ is less straightforward due to the fact that $\mathrm{Span}\{\bU_{*1},\ldots,\bU_{*r}\}$ may not uniquely determine the eigenvector matrix $\bU$. 
In particular, when there exist replicates among the eigenvalues $\{\lambda_{1}+\sigma^2,\ldots,\lambda_r+\sigma^2\}$ (\emph{i.e.}, $\lambda_k = \lambda_{k+1}$ for some $k\in[r - 1]$), the corresponding eigenvectors $[\bU_{*k}, \bU_{*(k+1)}]$ can only be identified up to orthogonal transformation. 
One solution is to focus on the Frobenius norm loss \citep{cai2013sparse,vu2013minimax} or the operator norm loss \citep{cai2015optimal} of the corresponding projection matrix $\bU\bU\transpose$, which is uniquely determined by $\mathrm{Span}\{\bU_{*1},\ldots,\bU_{*r}\}$ and vice versa. The corresponding minimax rate of convergence for $\bU\bU\transpose$ with respect to the projection operator norm loss $\|\widehat\bU\widehat\bU\transpose - \bU_0\bU_0\transpose\|_2$ is given by \cite{cai2015optimal}:
\begin{align}
\label{eqn:minimax_rate_U}
\inf_{\widehat\bU}\sup_{\bSigma_0\in\Theta(p, r, s)}\expect_{\bSigma_0}\|\widehat\bU\widehat\bU\transpose - \bU_0\bU_0\transpose\|_2^2\asymp \frac{s\log p}{n}.
\end{align}
Though convenient, the direct estimation of the projection matrix $\bU\bU\transpose$ does not provide insight into the element-wise errors of the principal eigenvectors $\{\bU_{*1},\ldots,\bU_{*r}\}$. Motivated by a recent paper \citep{cape2017two}, 
which presents a collection of technical tools for the analysis of element-wise eigenvector perturbation bounds with respect to the two-to-infinity norm, we also focus on the following two-to-infinity norm loss
\begin{align}\label{eqn:two_to_infinity_loss}
\|\widehat\bU - \bU_0\bW_\bU\|_{2\to\infty}
\end{align}
for estimating $\mathrm{Span}\{\bU_{*1},\ldots,\bU_{*r}\}$ in addition to the projection operator norm loss, where $\bW_\bU$ is the orthogonal matrix given by
\[
\bW_\bU = \arginf_{\bW\in\mathbb{O}(r)}\|\widehat\bU - \bU_0\bW\|_{\mathrm{F}}.
\]
Here, $\bW_\bU$ corresponds to the orthogonal alignment of $\bU_0$ so that $\widehat\bU$ and $\bU_0\bW_\bU$ are close in the Frobenius norm sense. 
As pointed out in \cite{cape2017two}, the use of $\bW_\bU$ as the orthogonal alignment matrix is preferred over the two-to-infinity alignment matrix 
\[
\bW_{2\to\infty}^\star = \arginf_{\bW\in\mathbb{O}(r)}\|\widehat\bU - \bU_0\bW\|_{2\to\infty},\]
because $\bW_{2\to\infty}$ is not analytically computable in general, whereas $\bW_\bU$ can be explicitly computed \citep{Stewart90}, facilitating the analysis: Let $\bU_0\transpose\widehat\bU$ admit the singular value decomposition $\bU_0\transpose\widehat\bU = \widetilde\bU\widetilde\bSigma\widetilde\bV\transpose$, then $\bW_\bU = \widetilde\bU\widetilde\bV\transpose$. 

The following lemma formalizes the connection between the projection operator norm loss and the two-to-infinity norm loss.
\begin{lemma}\label{lemma:two_to_infinity_and_projection}
Let $\bU$ and $\bU_0$ be two orthonormal $r$-frames in $\mathbb{R}^p$, where $2r< p$. Then there exists an orthonormal $2r$-frame $\bV_\bU$ in $\mathbb{R}^p$ depending on $\bU$ and $\bU_0$, such that
\[
\|\bU-\bU_0\bW_\bU\|_{2\to\infty}\leq \|\bV_\bU\|_{2\to\infty}
\left(\|\bU\bU\transpose - \bU_0\bU_0\transpose\|_2 + \|\bU\bU\transpose - \bU_0\bU_0\transpose\|_2^2\right),
\]
where $\bW_\bU = \arginf_{\bW\in\mathbb{O}(r)}\|\bU - \bU_0\bW\|_{\mathrm{F}}$ is the Frobenius orthogonal alignment matrix. 
\end{lemma}
When the projection operator norm loss $\|\bU\bU\transpose - \bU_0\bU_0\transpose\|_2$ is much smaller than one, Lemma \ref{lemma:two_to_infinity_and_projection}  states that the two-to-infinity norm loss can be upper bounded by the product of the projection operator norm loss 
and $\|\bV_\bU\|_{2\to\infty}$, where $\bV_\bU\in\mathbb{O}(p, 2r)$ is an orthonormal $2r$-frame in $\mathbb{R}^p$. In particular, under the sparse spiked covariance matrix models in high dimensions, the number of spikes $r$ can be much smaller than the dimension $p$ (\emph{i.e.}, $\bV_\bU$ is a ``tall and thin'' rectangular matrix), and hence the factor $\|\bV_\bU\|_{2\to\infty}$ can be much smaller than $\max_{\bV\in\mathbb{O}(p,2r)}\|\bV\|_2 = 1$. 

We provide the following motivating example for the preference on the two-to-infinity norm loss \eqref{eqn:two_to_infinity_loss} over the projection operator norm loss for $\mathrm{Span}\{\bU_{*1},\ldots,\bU_{*r}\}$.
\begin{example}
Let $s\geq4$ be even and $r = 1$. Suppose the truth $\bU_0$ is given by
\begin{align}
\bU_0 = 
\begin{bmatrix*}
\undermat{s}{\frac{1}{\sqrt{s}}&\ldots&\frac{1}{\sqrt{s}}}&
\undermat{p - s}{0 & \ldots & 0}
\end{bmatrix*}
\transpose,\nonumber
\end{align}
and consider the following two perturbations of $\bU_0$:
\begin{align}
\widehat\bU_1 & = 
\begin{bmatrix*}
\undermat{s/2}{c(\eps)(\frac{1}{\sqrt{s}} + \eps)&\ldots
}&
\undermat{s/2}{c(\eps)(\frac{1}{\sqrt{s}} - \eps)&\ldots
}&
\undermat{p-s}{0&\ldots & 0}
\end{bmatrix*}
\transpose,\nonumber\\\nonumber\\
\widehat\bU_2 & = 
\begin{bmatrix*}
c(\delta)(\frac{1}{\sqrt{s}} + \delta)&
\undermat{s-2}{\frac{1}{\sqrt{s}}&\ldots&\frac{1}{\sqrt{s}}}
& c(\delta)(\frac{1}{\sqrt{s}} - \delta)
&
\undermat{p - s}{0 & \ldots & 0}
\end{bmatrix*}
\transpose,\nonumber\\\nonumber
\end{align}
where $\eps>0$ is some sufficiently small perturbation, $c(\eps)^2 = 1/(1+s\eps^2)$, and $\delta$ is related to $\eps$ by 
\[
c(\delta)^2 = \frac{1}{1 + s\delta^2} = \frac{s}{2}\left\{\frac{1}{\sqrt{1 + s\eps^2}} - 1 + \frac{2}{s}\right\}.
\] 
The perturbed matrices $\widehat\bU_1$ and $\widehat\bU_2$ are designed such that their projection operator norm losses are identical, \emph{i.e.}, $\|\widehat\bU_1\widehat\bU_1\transpose - \bU_0\bU_0\transpose\|_2 = \|\widehat\bU_2\widehat\bU_2\transpose - \bU_0\bU_0\transpose\|_2$.
In contrast, $\widehat\bU_1$ and $\widehat\bU_2$ perturb $\bU_0$ in different fashions: all $s$ nonzero elements in $\bU_0$ are perturbed in $\widehat\bU_1$, whereas only two nonzero elements in $\bU_0$ are perturbed in $\widehat\bU_2$. 
\begin{figure}[htbp] 
\centerline{\includegraphics[width = 1\textwidth]{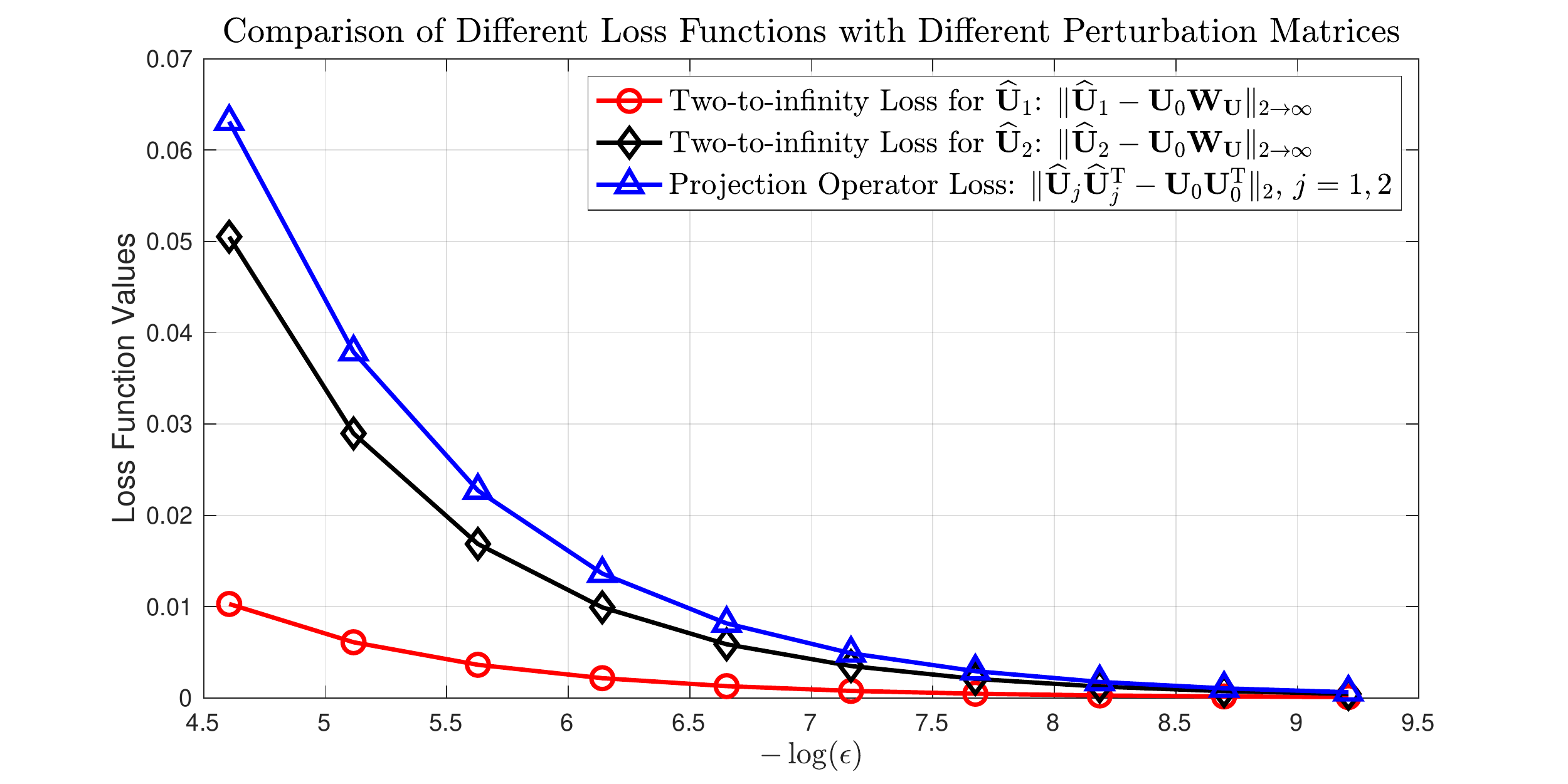}}
\caption{Motivating example: Comparison of different loss function values against different $-\log(\eps)$ values for two perturbed matrices $\widehat\bU_1$ and $\widehat\bU_2$.}
\label{Fig:Motivating_example}
\end{figure}
We examine the two candidate losses $\|\widehat\bU_j - \widehat\bU_0\bW_\bU\|_{2\to\infty}$ and $\|\widehat\bU_j\widehat\bU_j\transpose - \bU_0\bU_0\transpose\|_{2}$
for different values of $\epsilon$ and present them in Figure \ref{Fig:Motivating_example}. It can clearly be seen that the two-to-infinity norm loss is smaller than the projection operator norm loss. Furthermore, the projection operator norm loss is unable to detect the difference between $\widehat\bU_1$ and $\widehat\bU_2$. In contrast, the two-to-infinity norm loss indicates that $\widehat\bU_2$ has larger element-wise deviation from $\bU_0$ than $\widehat\bU_1$ does. 
Thus the two-to-infinity norm loss is capable of detecting element-wise perturbations of the eigenvector compared to the projection operator norm loss for estimating $\mathrm{Span}\{\bU{*1},\ldots,\bU_{*r}\}$. 
\end{example}

\subsection{The matrix spike-and-slab LASSO prior for joint sparsity} 
\label{sub:Matrix_spike_and_slab_lasso_priors_for_joint_sparsity}
We first illustrate the general Bayesian strategies in modeling sparsity occurring in high-dimensional statistics and then elaborate on the proposed prior model. Consider a simple yet canonical sparse normal mean problem. Suppose we observe 
 independent normal data $y_i\sim\mathrm{N}(\beta_i, 1)$, $i = 1,\ldots,n$, with the goal  of estimating the mean vector $\bbeta_n = (\beta_i)_{i=1}^n$, which is assumed to be sparse in the sense that $\sum_{i=1}^n\mathbbm{1}(|\beta_i|\neq 0)\leq s_n$ with the sparsity level $s_n=o(n)$ as $n\to\infty$. To model sparsity on $\bbeta$, classical Bayesian methods impose the spike-and-slab prior of the following form on $\bbeta$: for any measurable set $A\subset\mathbb{R}$,
\begin{align}
\label{eqn:spike_and_slab}
\Pi(\beta_i\in A\mid\lambda, \xi_i)& = (1 - \xi_i)\delta_0(A) + \xi_i\int_A\psi(\beta\mid\lambda)\mathrm{d}\beta,\\
(\xi_i\mid\theta)&\sim\mathrm{Bernoulli}(\theta),\nonumber
\end{align}
where $\xi_i$ is the indicator that $\beta_i = 0$, $\theta\in(0,1)$ represents the prior probability of $\beta_i$ being non-zero, $\delta_0$ is the point-mass at $0$ (called the ``spike'' distribution), and $\psi(\cdot\mid\lambda)$ is the density of an absolutely continuous distribution (called the ``slab'' distribution) with respect to the Lebesgue measure on $\mathbb{R}$ governed by some hyperparameter $\lambda$. Theoretical justifications for the use of spike-and-slab prior \eqref{eqn:spike_and_slab} for sparse normal means and sparse Bayesian factor models have been established in \cite{castillo2012} and \cite{pati2014posterior}, respectively. 
Therein, the spike-and-slab prior \eqref{eqn:spike_and_slab} involves point-mass mixtures, which can be daunting in terms of posterior simulations \citep{pati2014posterior}. To address this issue, the spike-and-slab LASSO prior \citep{rovckova2018bayesian} is designed as a continuous relaxation of \eqref{eqn:spike_and_slab}:
\begin{align}
\label{eqn:SSL}
\pi(\beta_i\mid\lambda_0,\lambda,\xi_i)& = (1-\xi_i)\psi(\beta_i\mid\lambda_0) + \xi_i\psi(\beta_i\mid\lambda),\\
(\xi_i\mid\theta)&\sim\mathrm{Bernoulli}(\theta),\nonumber
\end{align}
where $\psi(\beta\mid\lambda) = (\lambda/2)\exp(-\lambda\beta)$ is the Laplace distribution with mean $0$ and variance $2/\lambda^2$. When $\lambda_0 \gg \lambda$, the spike-and-slab LASSO prior \eqref{eqn:SSL} closely resembles the spike-and-slab prior \eqref{eqn:spike_and_slab}. The continuity feature of the spike-and-slab LASSO prior \eqref{eqn:SSL}, in contrast to the classical spike-and-slab prior \eqref{eqn:spike_and_slab}, is highly desired in high-dimensional settings in terms of computation efficiency.

Motivated by the spike-and-slab LASSO prior, we develop a matrix spike-and-slab LASSO prior to model joint sparsity in sparse spiked covariance matrix models \eqref{eqn:spiked_covariance} with the covariance matrix  $\bSigma = \bU\bLambda\bU\transpose+\sigma^2\eye_p$. The orthonormal constraint on the columns of $\bU$ makes it challenging to incorporate prior distributions. Instead, we consider the following reparametrization of $\bSigma$: 
\begin{align}\label{eqn:reparametriz_Sigma}
\bSigma = \left(\bU\bLambda^{1/2}\bV\transpose\right)\left(\bU\bLambda^{1/2}\bV\transpose\right)\transpose + \sigma^2\eye_p = \bB\bB\transpose + \sigma^2\eye_p,
\end{align}
where $\bB=\bU\bLambda^{1/2}\bV\transpose\in\mathbb{R}^{p\times r}$, and $\bV\in\mathbb{O}(r)$ is an arbitrary orthogonal matrix in $\mathbb{R}^r$. Clearly, in contrast to the orthonormal constraint on $\bU$, there is no constraint on $\bB$ except that $\mathrm{rank}(\bB) = r$. Furthermore, joint sparsity of $\bB$ is inherited from $\bU$: Specifically, for $|\mathrm{supp}(\bU)| = s\geq r$, there exists some permutation matrix $\bP\in\mathbb{R}^{p\times p}$ and $\bU^\star\in\mathbb{O}(s, r)$, such that
\[
\bU = \bP\begin{bmatrix*}
\bU^\star\\\zero_{(p-s)\times r}
\end{bmatrix*}.
\]
It follows directly that 
\[
\bB = \bU\bLambda^{1/2}\bV\transpose = \bP\begin{bmatrix*}
\bU^\star\\\zero_{(p-s)\times r}
\end{bmatrix*}\bLambda^{1/2}\bV\transpose = \bP\begin{bmatrix*}
\bU^\star\bLambda^{1/2}\bV\transpose\\\zero_{(p-s)\times r}
\end{bmatrix*},
\]
implying that $|\mathrm{supp}(\bB)|\leq s$. Therefore, working with $\bB$ allows us to circumvent the orthonormal constraint 
 while maintaining the jointly sparse structure of $\bU$. We propose the following matrix spike-and-slab LASSO prior on $\bB = [b_{jk}]_{p\times r}$: given hyperparameters $\lambda_0>0$ and $\theta\in(0,1)$, for each $j\in[p]$, we independently impose the prior on $\bB_{j*}$ as follows:
\begin{align*}
\pi(\bB_{j*}\mid \lambda_0,\xi_j)& = (1-\xi_j)\prod_{k=1}^r\psi_r(b_{jk}\mid\lambda + \lambda_0) + \xi_j\prod_{k=1}^r\psi_1(b_{jk}\mid\lambda),\\
(\xi_j\mid\theta)&\sim\mathrm{Bernoulli}(\theta),
\end{align*}
where $\bxi = [\xi_1,\ldots,\xi_p]\transpose\in\{0,1\}^{p}$ are binary group assignment indicators, and $\psi_\alpha(x\mid \lambda)$ is the density function of the double Gamma distribution with shape parameter $1/\alpha$ and rate parameter $\lambda$:
\[
\psi_\alpha(x\mid \lambda) = \frac{\lambda^{1/\alpha}}{2\Gamma(1/\alpha)}|x|^{1/\alpha - 1}\exp(-\lambda |x|),\quad -\infty<x<\infty.
\]
We further impose hyperpriors on $\lambda_0$ and $\theta$ as
\begin{align*}
\lambda_0\sim\text{IGamma}(1/p^2, 1)\quad\text{and}\quad
\theta \sim \mathrm{Beta}\left(1, p^{1 + \kappa}\right),
\end{align*}
where $\text{IGamma}(a,b)$ is the inverse Gamma distribution with density $\pi(\lambda_0)\propto \lambda_0^{-a-1}\exp(-b/\lambda_0)$, and $\kappa > 0$ is some fixed constant. 
We refer to the above hierarchical prior on $\bB$ as the matrix spike-and-slab LASSO prior and denote $\bB\sim\mathrm{MSSL}_{p\times r}(\lambda, 1/p^2, p^{1+\kappa})$. The hyperparameter $\lambda$ is fixed throughout. In the single-spike case ($r = 1$), we observe that $\psi_1(b_{jk}\mid\lambda) =(\lambda/2)\exp(-\lambda b_{jk})$ reduces to the density function of the Laplace distribution, and hence the matrix spike-and-slab LASSO prior coincides with the  spike-and-slab LASSO prior  \citep{rovckova2018bayesian}. 
Clearly, it can be seen that {\it a priori}, $\lambda_0$ is much larger than $\lambda$, so that $\xi_j = 0$ corresponds to rows $\bB_{j*}$ that are close to $\zero$, and $\xi_j = 1$ represents that the $j$th row is decently away from $\zero$. 
It should be noted that unlike the spike-and-slab prior \eqref{eqn:spike_and_slab}, the group indicator variable $\xi_j = 0$ or $1$ corresponds to small or large values of $\bB_{j*}$ rather than the exact sparsity of $\bB_{j*}$. In addition, $\theta\sim\mathrm{Beta}(1, p^{1+\kappa})$ indicates that the matrix spike-and-slab LASSO prior favors a large proportion of rows of $\bB$ being close to $\zero$. These features of the matrix spike-and-slab LASSO prior are in accordance with the joint sparsity assumption on $\bU$. We complete the prior specification by imposing $\sigma^2\sim\mathrm{IGamma}(a_\sigma,b_\sigma)$ for some $a_\sigma,b_\sigma>0$ for the sake of conjugacy.

Lastly, we remark that the parametrization \eqref{eqn:reparametriz_Sigma} of the spiked covariance matrix models \eqref{eqn:spiked_covariance} has another interpretation. The sampling model $\by_i\sim\mathrm{N}_p(\zero_p,\bSigma)$ can be equivalently characterized in terms of the latent factor model
\begin{align}
\label{eqn:latent_factor_model}
\by_i = \bB\bz_i + \beps_i,\quad \bz_i\sim\mathrm{N}_r(\zero_r,\eye_r),\quad\beps_i\sim\mathrm{N}_p(\zero_p,\sigma^2\eye_p),\quad i = 1,\ldots,n,
\end{align}
where $\bz_i$, $i = 1,\ldots,n$, are $r$-dimensional latent factors, $\bB$ is a $p\times r$ factor loading matrix, and $\beps_i$, $i = 1,\ldots,n$ are homoscedastic noisy vectors. Since by our earlier discussion $\bB$ is also sparse, this formulation is related to the sparse Bayesian factor models presented in \cite{doi:10.1093/biomet/asr013} and \cite{pati2014posterior}, the differences being the joint sparsity of $\bB$ and prior specifications on $\bB$. 
In addition, the latent factor formulation \eqref{eqn:latent_factor_model} is convenient for posterior simulation through Markov chain Monte Carlo, as discussed in Section 3.1 of \cite{doi:10.1093/biomet/asr013}. 



\section{Theoretical properties} 
\label{sec:theoretical_properties}

\subsection{Properties of the matrix spike-and-slab LASSO prior} 
\label{sub:properties_of_the_matrix_spike_and_slab_lasso_priors}
The theoretical properties of the classical spike-and-slab LASSO prior \eqref{eqn:SSL} have been partially explored by  \cite{rovckova2018bayesian} and \cite{rovckova2016spike}    in the context of sparse linear models and sparse normal means problems, respectively. It is not clear whether the properties of the spike-and-slab LASSO priors adapt to other statistical context, including sparse spiked covariance matrix models, high-dimensional multivariate regression \citep{Bai2018}, etc. 
In this subsection we present a collection of theoretical properties of the matrix spike-and-slab LASSO prior that not only are useful for deriving posterior contraction under the spiked covariance matrix models, but also may be of independent interest for other statistical tasks, \emph{e.g.}, sparse Bayesian linear regression with multivariate response \cite{ning2018bayesian}. 

Let $\bB\in\mathbb{R}^{p\times r}$ be a $p\times r$ matrix, and let $\bB_0\in\mathbb{R}^{p\times r}$ be a jointly $s$-sparse $p\times r$ matrix with $r\leq s\leq p$, corresponding to the underlying truth. In the sparse spiked covariance matrix model, $\bB$ represents the scaled eigenvector matrix $\bU\bLambda^{1/2}$ up to an orthonormal matrix in $\mathbb{O}(r)$, but for generality, we do not impose the statistical context in this subsection. A fundamental measure of goodness for various prior models with high dimensionality is the prior mass assignment on a small neighborhood around the true but unknown value of the parameter. This is referred to as the \emph{prior concentration} in the literature of Bayes theory. Formally, we consider the prior probability of the non-centered ball $\{\|\bB-\bB_0\|_{\mathrm{F}}<\eta\}$ under the prior distribution for small values of $\eta$.
\begin{lemma}
\label{lemma:prior_concentration}
Suppose $\bB\sim\mathrm{MSSL}_{p\times r}(\lambda, 1/p^2, p^{1+\kappa})$ for some fixed positive constants $\lambda$ and $\kappa$, and $\bB_0\in\mathbb{R}^{p\times r}$ is jointly $s$-sparse, where $1\leq r\leq s\leq p/2$. Then for small values of $\eta\in(0,1)$ with $\eta\geq1/p^\gamma$ for some $\gamma>0$, it holds that
\begin{align}
\Pi\left(\|\bB-\bB_0\|_{\mathrm{F}}<\eta\right)\geq\exp\left[-C_1\max\left\{\lambda^2s\|\bB_0\|_{2\to\infty}^2, sr\left|\log\frac{\lambda\eta}{\sqrt{sr}}\right|,s\log p\right\}\right]\nonumber
\end{align}
for some absolute constant $C_1>0$.
\end{lemma}
We next formally characterize how the matrix spike-and-slab LASSO prior imposes joint sparsity on the columns of $\bB$ using a probabilistic argument. Unlike the classical spike-and-slab prior \eqref{eqn:spike_and_slab}, which allows occurrence of exact zeros in the mean vector with positive probability, the spike-and-slab LASSO prior \eqref{eqn:SSL} (the matrix spike-and-slab LASSO prior) is absolutely continuous with respect to the Lebesgue measure on $\mathbb{R}^n$ ($\mathbb{R}^{p\times r}$, respectively), and $|\mathrm{supp}(\bB)| = p$ with probability one. Rather than forcing elements of $\bB$ to be exactly $0$, the matrix spike-and-slab LASSO prior shrinks elements of $\bB$ toward $0$. 
This behavior suggests the following generalization of the row support of a matrix $\bB$: for $\delta>0$ taken to be small, we define $\mathrm{supp}_\delta(\bB) = \{j\in[p]:\|\bB_{j*}\|_2>\delta\}$. Namely, $\mathrm{supp}_\delta(\bB)$ consists of row indices of $\bB$ whose Euclidean norms are greater than $\delta$. Intuitively, one should expect that under the matrix spike-and-slab LASSO prior, $|\mathrm{supp}_\delta(\bB)|$ should be small with large probability. The following lemma formally confirms this 
 intuition. 
\begin{lemma}\label{lemma:prior_sparsity}
Suppose $\bB\sim\mathrm{MSSL}_{p\times r}(\lambda, 1/p^2, p^{1+\kappa})$ for some fixed positive constants $\lambda$ and $\kappa\leq 1$, $1\leq r\leq p$. Let $\delta\in(0,1)$ be a small number with $\delta>1/p^\gamma$ for some $\gamma>0$, and let $s$ be an integer such that $(s\log p)/p$ is sufficiently small. Then for any $\beta > 4\gamma\exp(1)$, it holds that
\[
\Pi\left(|\mathrm{supp}_\delta(\bB)|>\beta s\right)\leq 2\exp\left\{-\min\left(\frac{\beta\kappa}{2},\frac{\beta}{\mathrm{2e}} - 2\gamma\right)s\log p\right\}.
\]
\end{lemma}

We conclude this section by providing a large deviation inequality for the matrix spike-and-slab LASSO prior.
\begin{lemma}\label{lemma:prior_deviation}
Suppose $\bB\sim\mathrm{MSSL}_{p\times r}(\lambda, 1/p^2, p^{1+\kappa})$ for some fixed positive $\lambda$ and $\kappa<1$, and $\bB_0\in\mathbb{R}^{p\times r}$ is jointly $s$-sparse, where $r\log n\lesssim \log p$, and $(s\log p)/p$ is sufficiently small. Let $(\delta_n)_{n=1}^\infty$ and $(t_n)_{n=1}^\infty$ be positive sequences such that $1/p^\gamma\leq\delta_n\leq 1$ and $t_n/(sr) \to \infty$. Then for sufficiently large $n$ and for all $\beta >4\gamma\exp(1)$, it holds that
\begin{align*}
&\Pi\left[\sum_{j=1}^p\|\bB_{j*}\|_1\mathbbm{1}\{j\in\mathrm{supp}_{\delta_n}(\bB)\cup\mathrm{supp}(\bB_0)\}\geq t_n\right]\nonumber\\
&\quad\leq 2\exp\left[-C_2\min\left\{\left(\frac{t_n}{\beta sr}\right)^2,\left(\frac{t_n}{r}\right)^2,\frac{t_n}{r}\right\}\right]
 + 
3\exp\left\{ - \min\left(\frac{\beta\kappa}{2},\frac{\beta}{\mathrm{2e}} - 2\gamma\right)s\log p\right\}\nonumber
\end{align*}
for some absolute constant $C_2>0$. 
\end{lemma}


\subsection{Posterior contraction for the sparse Bayesian spiked covariance matrix model} 
\label{sub:posterior_contraction_for_bayesian_sparse_spiked_covariance_matrices}
We now present the posterior contraction rates for sparse spiked covariance matrix models under the matrix spike-and-slab LASSO prior with respect to various loss functions, which are the main results of this paper. 
We point out that the posterior contraction rates presented in the following theorem are minimax-optimal as they coincide with \eqref{eqn:minimax_rate_Sigma} and \eqref{eqn:minimax_rate_U}. 
\begin{theorem}\label{thm:posterior_contraction}
Assume the data $\by_1,\ldots,\by_n$ are independently sampled from $\mathrm{N}_p(\zero_p, \bSigma_0)$ with $\bSigma_0 = \bU_0\bLambda_0\bU_0\transpose + \sigma_0^2\eye_p$, 
$\bLambda_0 = \mathrm{diag}(\lambda_{01},\ldots,\lambda_{0r})$, 
$|\mathrm{supp}(\bU_0)|\leq s$, and $1\leq r\leq s\leq p$. 
Suppose $(s\log p)/n \to 0$, $p/n\to\infty$, and $r\log n\lesssim \log p$. Let $\bB\sim\mathrm{MSSL}_{p\times r}(\lambda, 1/p^2, p^{1+\kappa})$ for some positive $\lambda>0$ and $\kappa \leq 1$, and $\sigma^2\sim\mathrm{IGamma}(a_\sigma,b_\sigma)$ for some $a_\sigma,b_\sigma\geq1$. Then there exists some constants $M_0>0$, $R_0$, and $C_0$ depending on $\sigma_0$ and $\bLambda_0$, and hyperparameters, such that the following posterior contraction for $\bSigma = \bB\bB\transpose + \sigma^2\eye_p$ holds for all $M\geq M_0$ when $n$ is sufficiently large:
\begin{align}
\label{eqn:posterior_contraction_Sigma}
\expect_0\left\{\Pi\left(\|\bSigma-\bSigma_0\|_2 > M\sqrt{\frac{s\log p}{n}}\mathrel{\bigg|}\bY_n\right)\right\}
&\leq R_0\exp(-C_0s\log p).
\end{align}
For each $\bB$, let $\bU_\bB\in\mathbb{O}(p, r)$ be the left-singular vector matrix of $\bB$. Then the following posterior contraction for $\bU_\bB$ holds for all $M\geq M_0$:
\begin{align}
\label{eqn:posterior_contraction_Projection}
\expect_0\left\{\Pi\left(\|\bU_\bB\bU_\bB\transpose-\bU_0\bU_0\transpose\|_{2} > \frac{2M}{\lambda_{0r}}\sqrt{\frac{s\log p}{n}}\mathrel{\Big|}\bY_n\right)\right\}
&\leq R_0\exp(-C_0s\log p).
\end{align}
\end{theorem}
\begin{remark}
We briefly compare the posterior contraction rates obtained in Theorem \ref{thm:posterior_contraction} with some related results in the literature. 
In \cite{pati2014posterior} the authors consider the posterior contraction with respect to the operator norm loss $\|\bSigma-\bSigma_0\|_2$ of the entire covariance matrix, while
in \cite{gao2015rate}, the authors consider the posterior contraction with respect to the projection Frobenius norm loss $\|\bU\bU\transpose-\bU_0\bU_0\transpose\|_{\mathrm{F}}$ for estimating $\mathrm{Span}\{\bU_{*1},\ldots,\bU_{*r}\}$. In \cite{pati2014posterior}, the notion of sparsity is slightly different than the joint sparsity notion presented here, as they assume that under the latent factor model representation \eqref{eqn:latent_factor_model}, the individual supports of columns of $\bB$ are not necessarily the same. When $r = O(1)$, the assumption in \cite{pati2014posterior} coincides with this paper, and our rate $\eps_n = \sqrt{(s\log p)/n}$ is superior to the rate $\sqrt{(s\log p\log n)/n}$ obtained in \cite{pati2014posterior} by a logarithmic factor. The assumptions in \cite{gao2015rate} are the same as those in \cite{pati2014posterior}, and in \cite{gao2015rate} the authors focus on designing a prior that yields rate-optimal posterior contraction with respect to the Frobenius norm loss of the projection matrices as well as adapting to the prior sparsity $s$ and the rank $r$. Our result in equation \eqref{eqn:posterior_contraction_Projection}, which focuses on the projection operator norm loss, serves as a complement to the rate-optimal posterior contraction for principal subspaces under the joint sparsity assumption in constrast to \cite{gao2015rate}, in which the authors work on the projection Frobenius norm loss.
\end{remark}
To derive the posterior contraction rate for the principal subspace with respect to the two-to-infinity norm loss, we need the posterior contraction result for $\bSigma$ with respect to the stronger matrix infinity norm. These two results are summarized in the following theorem. 
\begin{theorem}\label{thm:posterior_contraction_two_to_infinity}
Assume the conditions in Theorem \ref{thm:posterior_contraction} hold. 
Further assume that the eigenvector matrix $\bU_0$ exhibits bounded coherence: $\|\bU_0\|_{2\to\infty}\leq C_\mu\sqrt{r/s}$ for some constant $C_\mu\geq 1$, and the number of spikes $r$ is sufficiently small in the sense that $r^3/s= O(1)$. 
Then there exists some constants $M_{2\to\infty}>0$ depending on $\sigma_0$ and $\bLambda_0$, and hyperparameters, such that the following posterior contraction for $\bSigma = \bB\bB\transpose + \sigma^2\eye_p$ holds for all $M\geq M_{2\to\infty}$ when $n$ is sufficiently large:
\begin{align}
\label{eqn:posterior_contraction_Sigma}
\expect_0\left\{\Pi\left(\|\bSigma-\bSigma_0\|_\infty > M r\sqrt{\frac{s\log p}{n}}\mathrel{\bigg|}\bY_n\right)\right\}
&\leq R_0\exp(-C_0 s\log p),
\end{align}
For each $\bB$, let $\bU_\bB\in\mathbb{O}(p, r)$ be the left-singular vector matrix of $\bB$. Then the following posterior contraction for $\bU_\bB$ holds for all $M\geq M_0$:
\begin{align}
\label{eqn:posterior_contraction_two_to_infinity}
&\expect_0\left[\Pi\left\{\|\bU_\bB-\bU_0\bW_\bU\|_{2\to\infty} > M
\left(\sqrt{\frac{r^3\log p}{n}}\vee{\frac{s\log p}{n}}\right)
\right\}
\right]
\leq 2R_0\exp(-C_0s\log p),
\end{align}
where $\bW_\bU$ is the Frobenius orthogonal alignment matrix 
\[
\bW_\bU = \arginf_{\bW\in\mathbb{O}(r)}\|\bU_\bB-\bU_0\bW\|_{\mathrm{F}}.
\]
\end{theorem}
\begin{remark}
We also present some remarks concerning the posterior contraction with respect to the two-to-infinity norm loss $\|\bU-\bU_0\bW_\bU\|_{2\to\infty}$. In \cite{cape2017two}, the authors show that 
\[
\|\bU-\bU_0\bW_\bU\|_{2\to\infty}\leq \|\bU-\bU_0\bW_\bU\|_2\lesssim\|\bU\bU\transpose-\bU_0\bU_0\transpose\|_2,
\]
meaning that $\|\bU-\bU_0\bW_\bU\|_{2\to\infty}$ can be coarsely upper bounded by the projection operator norm loss $\|\bU\bU\transpose-\bU_0\bU_0\transpose\|_2$. 
This naive bound immediately yields 
\[
\expect_0\left\{\Pi\left(\|\bU_\bB - \bU_0\bW_\bU\|_{2\to\infty} > M\sqrt{\frac{s\log p}{n}}\mathrel{\Big|}\bY_n\right)\right\}\leq R_0\exp(-C_0s\log p)
\]
for some large $M$, which is the same as \eqref{eqn:posterior_contraction_Projection}.
Our result \eqref{eqn:posterior_contraction_two_to_infinity} improves this rate by a factor of $\{\sqrt{r^3/s}\vee \sqrt{(s\log p)/n}\}$ and, thus yielding a tighter posterior contraction rate with respect to the two-to-infinity norm loss. In particular, when $r\ll s$ (\emph{i.e.}, $\bU_0$ is a ``tall and thin'' rectangular matrix), the factor $\sqrt{r^3/s}$ can be much smaller than $1$. 
\end{remark}
The posterior contraction rate \eqref{eqn:posterior_contraction_Projection} also leads to the following risk bound for a point estimator of the principal subspace $\mathrm{Span}\{\bU_{*1},\ldots,\bU_{*r}\}$:
\begin{theorem}\label{thm:risk_bound_point_estimate}
Assume the conditions in Theorem \ref{thm:posterior_contraction} hold. Let
\[
\widehat{\bOmega} = \int \bU_\bB\bU_\bB\transpose\Pi(\mathrm{d}\bB\mid\bY_n)
\]
be the posterior mean of the projection matrix $\bU_\bB\bU_\bB\transpose$, and set $\widehat\bU\in\mathbb{O}(p, r)$ be the orthonormal $r$-frame in $\mathbb{R}^p$ with columns being the first $r$ eigenvectors corresponding to the first $r$ largest eigenvalues of $\widehat\bOmega$. Then the following risk bound holds for $\widehat\bU$ for sufficiently large $n$:
\begin{align}
\expect_0 \left(\|\widehat\bU\widehat\bU\transpose - \bU_0\bU_0\transpose\|_2\right) &\leq \left(\frac{4M_0}{\lambda_{0r}} + 4\sqrt{R_0}\right)\sqrt{\frac{s\log p}{n}}
.\nonumber
\end{align}
\end{theorem}

The setup so far is concerned with the case where $r$ is known and fixed. When $r$ is unknown, \cite{cai2013sparse} provides a diagonal thresholding method for consistently estimating $r$. 
 In such a setting, the posterior contraction in Theorem \ref{thm:posterior_contraction} reduces to the following weaker version:
\begin{corollary}\label{corr:weak_contraction}
Assume the data $\by_1,\ldots,\by_n$ are independently sampled from $\mathrm{N}_p(\zero_p, \bSigma_0)$ with $\bSigma_0 = \bU_0\bLambda_0\bU_0\transpose + \sigma_0^2\eye_p$, 
$\bLambda_0 = \mathrm{diag}(\lambda_{01},\ldots,\lambda_{0r})$, 
$|\mathrm{supp}(\bU_0)|\leq s$, and $1\leq r\leq s\leq p$. 
Suppose $(s\log p)/n \to 0$, $p/n\to\infty$, and $r\log n\lesssim \log p$, but $r$ is unknown and instead is consistently estimated by $\hat r$ (\emph{i.e.}, $\prob_0(\hat r = r) \to 1$). Let $\bB\sim\mathrm{MSSL}_{p\times \hat r}(\lambda, 1/p^2, p^{1+\kappa})$ for some positive $\lambda>0$ and $\kappa \leq 1$, and $\sigma^2\sim\mathrm{IGamma}(a_\sigma,b_\sigma)$ for some $a_\sigma,b_\sigma\geq1$. Then there exists some large constant $M_0>0$, such that the following posterior contraction for $\bSigma$ holds for all $M\geq M_0$:
\begin{align}
\lim_{n\to\infty}
\expect_0\left\{\Pi\left(\|\bSigma-\bSigma_0\|_2 > M\sqrt{\frac{s\log p}{n}}\mathrel{\bigg|}\bY_n\right)\right\}
&\to 0.\nonumber
\end{align}
For each $\bB$, let $\bU_\bB\in\mathbb{O}(p, \hat r)$ be the left-singular vector matrix of $\bB$. Then the following posterior contraction for $\bU$ holds for all $M\geq M_0$:
\begin{align}
\lim_{n\to\infty}
\expect_0\left\{\Pi\left(\|\bU_\bB\bU_\bB\transpose-\bU_0\bU_0\transpose\|_{2} > \frac{2M}{\lambda_{0r}}\sqrt{\frac{s\log p}{n}}\mathrel{\bigg|}\bY_n\right)\right\}
&\to 0.\nonumber
\end{align}
\end{corollary}


\subsection{Proof Sketch and Auxiliary Results} 
\label{sub:proof_sketch_and_auxiliary_results}
Now we sketch the proof of Theorem \ref{thm:posterior_contraction} along with some important auxiliary results. The proof strategy is based on a modification of the standard testing-and-prior-concentration approach, which was originally developed in \cite{ghosal2000convergence} for proving convergence rates of posterior distributions, and later adopted to a variety of statistical contexts. Specialized to the sparse spiked covariance matrix models, 
let us consider the posterior contraction for $\bSigma$ with respect to the operator norm loss as an example. The posterior contraction for $\bSigma$ with respect to the infinity norm loss can be proved in a similar fashion. 
Denote $\calU_n = \{\bSigma:\|\bSigma-\bSigma_0\|_2\leq M\eps_n\}$, and write the posterior distribution as
\begin{align}
\label{eqn:posterior_distribution}
\Pi(\calU_n^c\mid \bY_n) = \frac{\int_{\calU_n^c}\exp\{\ell_n(\bSigma) - \ell_n(\bSigma_0)\}\Pi(\mathrm{d}\bSigma)}{\int\exp\{\ell_n(\bSigma) - \ell_n(\bSigma_0)\}\Pi(\mathrm{d}\bSigma)} = \frac{N_n(\calU_n)}{D_n},
\end{align}
where $\ell_n(\bSigma)$ is the log-likelihood function of $\bSigma$ given by
\[
\ell_n(\bSigma) = \sum_{i=1}^n\log p(\by_i\mid \bSigma) = \sum_{i=1}^n\left\{-\frac{1}{2}\log\det(2\pi\bSigma) - \frac{1}{2}\by_i\transpose\bSigma^{-1}\by_i\right\}.
\]
To provide a useful upper bound for $\expect_0\{\Pi(\calU_n^c\mid\bY_n)\}$ (\emph{e.g.}, $\exp(-C_0s\log p)$ appearing in Theorem \ref{thm:posterior_contraction}), we modify the original testing-and-prior-concentration approach and require that the following three conditions hold:
\begin{enumerate}
  \item \textbf{Prior concentration condition. }The prior distribution provides sufficient concentration around the true $\bSigma_0$: There exists some constant $C_3>0$ such that
  \[
  \Pi(\|\bSigma-\bSigma_0\|_{\mathrm{F}}^2\leq sr/n)\geq\exp(-C_3s\log p)
  \]
  for sufficient large $n$.
  \item \textbf{Existence of Tests. }There exists a sequence of subsets $(\calF_n)_{n=1}^\infty$ of $\Theta(p, r, s)$, such that $\Pi(\bSigma\in\calF_n^c)\leq \exp(-C_4s\log p)$ for some sufficiently large constant $C_4>0$, and there exists a sequence of test functions $(\phi_n)_{n=1}^\infty$, such that
  \begin{align}
  \expect_0(\phi_n)&\lesssim \exp\left(-C_{41}\sqrt{M}n\eps_n^2\right),
  \nonumber\\
  \sup_{\bSigma\in\calU_n^c\cap\calF_n}\expect_\bSigma(1-\phi_n)&\lesssim \exp(-C_{42}Mn\eps_n^2)\nonumber
  \end{align}
  for some constants $C_{41}, C_{42}>0$.
\end{enumerate}


The prior concentration condition can be verified by invoking Lemma \ref{lemma:prior_concentration}. This condition is useful, as it guarantees that the denominator $D_n$ appearing in the right-hand side of \eqref{eqn:posterior_distribution} can be lower bounded with high probability. The following lemma formalizes this result. 
\begin{lemma}
\label{lemma:evidence_lower_bound}
Let $\calK_n(\eta)=\{\|\bSigma-\bSigma_0\|_{\mathrm{F}}\leq\eta\}$ and $\eta<\sigma_0^2/2$. Then there exists some event $\calA_n$ such that
\[
\calA_n\subset\left\{D_n\geq \Pi_n\{\bSigma\in \calK_n(\eta)\}\exp\left[-\left\{\frac{C_3\log\rho}{2(\lambda_{0r}+\sigma_0^2)}+1\right\}n\eta^2\right]\right\}
\]
for some absolute constant $C_3>0$, and 
\[
\prob_0(\calA_n^c)
\leq2\exp\left\{-\tilde C_3\min\left(\frac{n\eta^2}{\|\bSigma_0^{-1}\|_2^2}, n\eta^2\right)\right\},
\]
where $\rho = 2(\lambda_{01}+\sigma_0^2)/(\lambda_{0r}+\sigma_0^2)$ depends on the spectra of $\bSigma$ only, and $\tilde C_3>0$ is an absolute constant.
\end{lemma}
Verifying the existence of tests is slightly more involved. It relies on Lemma \ref{lemma:prior_sparsity}, Lemma \ref{lemma:prior_deviation}, and the following auxiliary lemma. 
\begin{lemma}\label{lemma:existence_test}
Assume the data $\by_1,\ldots,\by_n$ follow $\mathrm{N}_p(\zero_p, \bSigma)$, $1\leq r\leq p$. Suppose $\bU_0\in\mathbb{O}(p, r)$ satisfies $|\mathrm{supp}(\bU_0)|\leq s$, and $r\leq s\leq p$. 
For any positive $\delta$, $t$, and $\tau$, define 
\begin{align*}
\calF(\delta,\tau,t)& =\Bigg\{\bB\in\mathbb{R}^{p\times r}:|\mathrm{supp}_\delta(\bB)|\leq \tau,
\sum_{j=1}^p\|\bB_{j*}\|_2^2\mathbbm{1}\{j\in\mathrm{supp}_\delta(\bB)\cup \mathrm{supp}(\bU_0)\}\leq t^2
\Bigg\}.
\end{align*}
Let the positive sequences $(\delta_n, \tau_n, t_n,\eps_n)_{n=1}^\infty$ satisfy $(\sqrt{p}\delta_n + 2t_n)\sqrt{p}\delta_n\leq M_1 \eps_n$ for some constant $M_1>0$, and $\eps_n\leq 1$. Consider testing 
\[
H_0:\bSigma = \bSigma_0 = \bU_0\bLambda_0\bU_0\transpose + \sigma_0^2\eye_p
\]
versus
\[
H_1:\bSigma \in\left\{\bSigma = \bB\bB\transpose + \sigma^2\eye_p:\|\bSigma-\bSigma_0\|_2>M\eps_n, \bB\in\calF(\delta_n, \tau_n, t_n)\right\}.
\]
Then for each $M\geq \max\{M_1/2, (128\|\bSigma_0\|_2^4)^{1/3}\}$, there exists a test function $\phi_n:\mathbb{R}^{n\times p}\to [0,1]$, such that
\begin{align}
\expect_0(\phi_n)&\leq 3\exp\left\{(2 + C_4)(\tau_n\log p + 2s_n) - \frac{C_4\sqrt{M}}{\sqrt{2}}n\eps_n^2\right\},\nonumber\\
\sup_{\bSigma\in H_1}\expect_\bSigma(1-\phi_n)&\leq \exp\left\{C_4(\tau_n + 2s_n) - \frac{C_4M}{8}n\eps_n^2\right\}\nonumber
\end{align}
for some absolute constant $C_4>0$.
\end{lemma}



\section{Numerical examples} 
\label{sec:numerical_examples}

\subsection{Synthetic examples} 
\label{sub:synthetic_examples}
We evaluate the numerical performance of the proposed Bayesian method for estimating sparse spiked covariance matrices via simulation studies. We set the sample size $n = 100$ and the number of features $p = 200$. The support size $s$ of the eigenvector matrix $\bU_0$ ranges over $\{8, 12, 20, 40\}$, and the number of spikes $r$ takes values in $\{1, 4\}$. 
The indices of the non-zero rows of $\bU_0$ are uniformly sampled from $\{1,\ldots,p\}$, and we set the diagonal elements of $\bLambda_0$ to be equally spaced over the interval $[10, 20]$, with $\lambda_{01} = 20$ and $\lambda_{0r} = 10$. 
The non-zero rows of $\bU_0$, themselves forming an orthonormal $r$-frame in $\mathbb{R}^s$, denoted by $\bU_0^\star$, are generated as the left singular vector matrix of $\bL$, an 
$s\times r$ matrix consisting of independent $\mathrm{Unif}(1, 2)$ elements. 

Posterior inference is carried out using a standard Metropolis-within-Gibbs sampler, and $1000$ post burn-in samples are collected after $1000$ iterations of burn-in phase. We then take the posterior mean $\widehat\bSigma$ of $\bSigma$ as the point estimator for $\bSigma$, and the   $\widehat\bU$ given by Theorem \ref{thm:risk_bound_point_estimate} as the point estimator for the subspace $\mathrm{Span}\{\bU_{*1},\ldots,\bU_{*r}\}$. 
For comparison, several competitors are considered, 
 including 
{\color{black}the sparse Bayesian factor model with multiplicative Gamma process shrinkage prior (MGPS, \cite{doi:10.1093/biomet/asr013})},
the principal orthogonal complement thresholding method (POET, \cite{fan2013large}), and the sparse principal component analysis method (SPCA, \cite{doi:10.1198/106186006X113430}). 
In each simulation setup (\emph{i.e.}, each $(r, s)$ pair), $50$ replicates of synthetic datasets are generated, and for each synthetic dataset, we compute the point estimators $\widehat\bSigma$, $\widehat\bU$ as well as those offered by the three competing approaches,   the operator norm loss $\|\widehat\bSigma-\bSigma_0\|_2$ for $\bSigma$, the two-to-infinity norm loss and the projection operator norm loss for $\mathrm{Span}\{\bU_{*1},\ldots,\bU_{*r}\}$ ($\|\widehat\bU-\bU_0\bW_\bU\|_{2\to\infty}$ and $\|\widehat\bU\widehat\bU\transpose-\bU_0\bU_0\transpose\|_2$), and compute the medians of these losses.
The results are tabulated in Table \ref{table:Synthetic_Error}.  

\begin{table}[t!]
    \caption{The operator norm loss $\|\widehat\bSigma-\bSigma_0\|_2$ with the posterior mean $\widehat\bSigma$, the squared projection operator norm loss $\|\widehat\bU\widehat\bU\transpose-\bU_0\bU_0\transpose\|_2^2$, and the squared two-to-infinity norm loss $\|\widehat\bU-\bU_0\bW_\bU\|_{2\to\infty}^2$, where $\widehat\bU$ is the point estimator of $\bU$ given by Theorem \ref{thm:risk_bound_point_estimate}. The medians across $50$ replicates of synthetic datasets are tabulated. MSSL stands for the sparse Bayesian spiked covariance matrix model with the matrix spike-and-slab LASSO prior. }
    \begin{subtable}{\linewidth}
        \caption{The operator norm loss $\|\widehat\bSigma-\bSigma_0\|_2$}
        \centering
        \begin{tabular}{c | c  c | c  c | c  c | c  c }
      \hline\hline
      $s$& \multicolumn{2}{c}{$8$} & \multicolumn{2}{c}{$12$} & \multicolumn{2}{c}{$20$} & \multicolumn{2}{c}{$40$}\\
      \hline
      $r$& $1$ & $ 4$ & $1$ & $ 4$ & $1$ & $ 4$ & $1$ & $ 4$\\
      \hline
      MSSL &{\bf 1.85} & \bf{6.68}  & \bf{1.97}  & \bf{6.76} & \bf{2.61}  & \bf{8.11}  & \bf{5.12}  & \bf{10.35} \\
      MGPS & $9.86$ & $16.54$ & $9.88$ & $17.78$ & $9.88$ & $18.52$ & $9.88$ & $19.05$ \\
      POET & $7.54$  & $11.17$  & $7.47$  & $11.10$  & $7.61$  & $11.60$  & $7.60$  & $10.97$ \\
      SPCA &$8.08$ & $18.03$  & $8.09$  &$18.04$ & $8.11$  & $18.07$  & $8.17$  & $18.10$  \\
      \hline\hline
    \end{tabular}%
    \end{subtable}\par
    \vspace*{1ex}
    \begin{subtable}{\linewidth}
        \caption{The squared projection operator norm loss $\|\widehat\bU\widehat\bU\transpose-\bU_0\bU_0\transpose\|_2^2$}
        \centering
        \begin{tabular}{c | c  c | c  c | c  c | c  c }
      \hline\hline
      $s$& \multicolumn{2}{c}{$8$} & \multicolumn{2}{c}{$12$} & \multicolumn{2}{c}{$20$} & \multicolumn{2}{c}{$40$} \\
      \hline
      $r$& $1$ & $ 4$ & $1$ & $ 4$ & $1$ & $ 4$ & $1$ & $ 4$\\
      \hline
      MSSL&\bf{0.0099} & \bf{0.033}  & \bf{0.018}  & \bf{0.036} & \bf{0.026} & \bf{0.046} & \bf{0.10} & \bf{0.061} \\
      MGPS& $0.18$ & $0.27$ & $0.19$ & $0.47$ & $0.20$ & $0.35$ & $0.20$ & $0.27$\\
      POET& $0.18$  & $0.21$ & $0.18$  & $0.20$  & $0.19$ & $0.20$ & $0.18$ & $0.20$\\
      SPCA &$0.05$ & $0.092$ & $0.068$ & $0.11$ & $0.10$ & $0.15$ & $0.18$ & $0.22$\\
      \hline\hline
    \end{tabular}%
    \end{subtable}\par
    \vspace*{1ex}
    \begin{subtable}{\linewidth}
        \caption{The squared two-to-infinity norm loss $\|\widehat\bU-\bU_0\bW_\bU\|_{2\to\infty}^2$}
        \centering
        \begin{tabular}{c | c  c | c  c | c  c | c  c }
    \hline\hline
    $s$& \multicolumn{2}{c}{$8$} & \multicolumn{2}{c}{$12$} & \multicolumn{2}{c}{$20$} & \multicolumn{2}{c}{$40$} \\
    \hline
    $r$& $1$ & $ 4$ & $1$ & $ 4$ & $1$ & $ 4$ & $1$ & $ 4$\\
    \hline
    MSSL&\bf{0.0038} & \bf{0.011} & \bf{0.0058} & \bf{0.012} & $0.014$ & \bf{0.012} & $0.016$ & \bf{0.011}\\
    MGPS & $0.0093$ & $0.085$ & $0.0096$ & $0.14$ & $0.0092$ & $0.14$ & $0.01$ & $0.077$\\
    POET& $0.0082$ & $0.013$ & $0.0082$ & $0.013$ & \bf{0.0086} & $0.012$ & \bf{0.0088} & $0.013$ \\
    SPCA &$0.024$& $0.027$ & $0.022$ &$0.040$ & $0.022$ & $0.039$ & $0.025$ & $0.038$\\
    \hline\hline
  \end{tabular}%
    \end{subtable}
    \label{table:Synthetic_Error}
\end{table}
The numerical results in Tables \ref{table:Synthetic_Error}(a) and \ref{table:Synthetic_Error}(b) indicate that the proposed 
Bayesian approach yields smallest operator norm losses for  $\bSigma$ and smallest projection operator norm losses for the subspace estimation, respectively. 
In terms of the two-to-infinity norm loss for the subspace estimation, Table \ref{table:Synthetic_Error}(c) shows that  the point estimates $\widehat\bU$ using the proposed approach yield smaller losses
compared to the competitors when $s = 8$ and $s = 12$ for both $r = 1$ and $r = 4$, while POET is more accurate for the single-spike cases when $s = 20$ and $s = 40$. The comparison between the two losses for the subspace estimation is also visualized in Figure \ref{fig:Loss_comparison}, {\color{black}suggesting that the two-to-infinity norm loss is less sensitive to the row support size $s$ than the projection operator norm loss as $s$ increases.}
\begin{figure}[t!]
  \centerline{\includegraphics[width=1\textwidth]{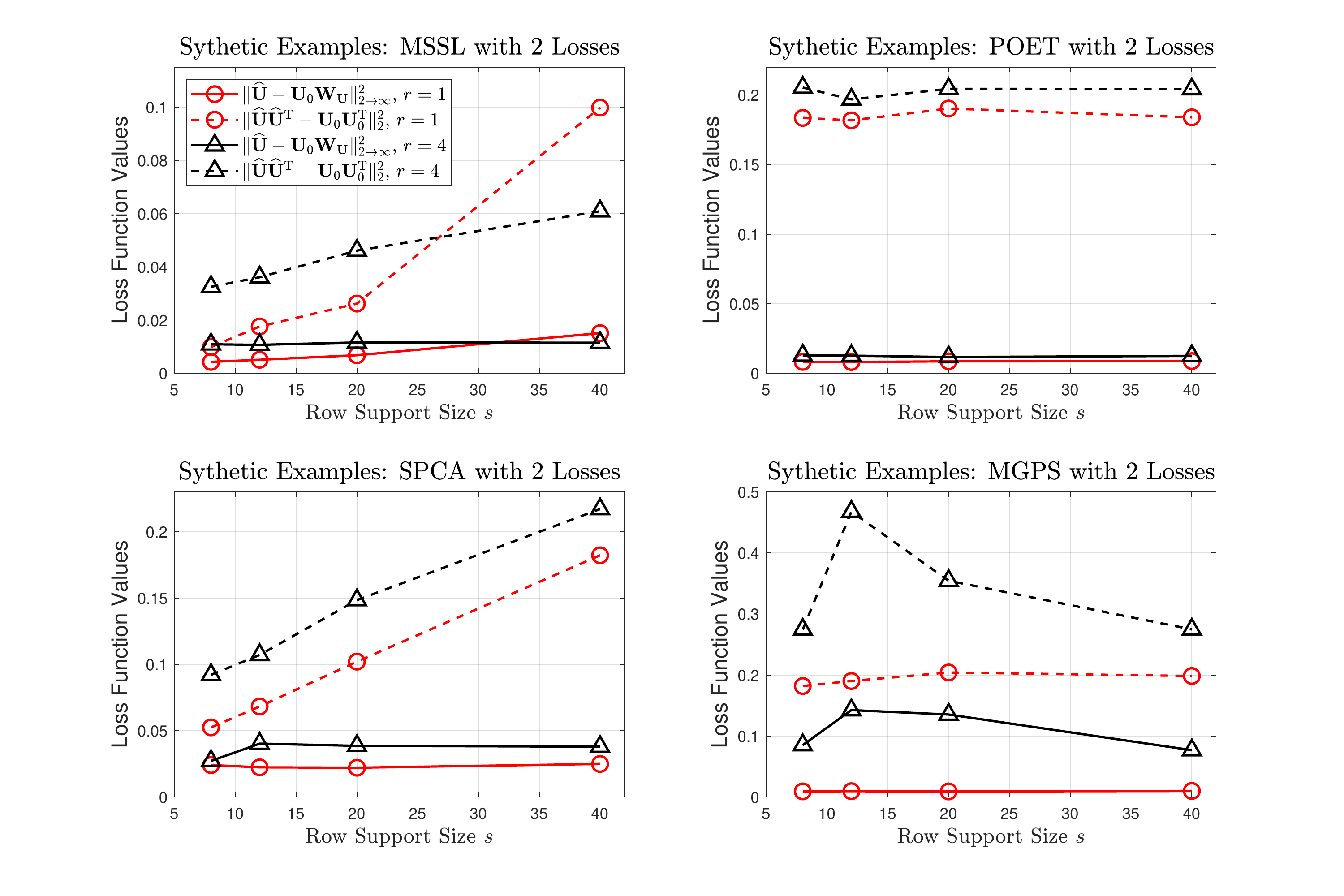}}
  \caption{Comparison of the two-to-infinity norm loss ($\|\widehat\bU-\bU_0\bW_\bU\|_{2\to\infty}$) and the projection operator norm loss ($\|\widehat\bU\widehat\bU\transpose-\bU_0\bU_0\transpose\|_2$) for synthetic examples. MSSL stands for the sparse Bayesian spiked covariance matrix model with the matrix spike-and-slab LASSO prior. }
  \label{fig:Loss_comparison}
\end{figure} 

We further evaluate the performance of estimating the principal subspace $\mathrm{Span}\{\bU_{*1},\ldots,\bU_{*r}\}$ when $s = 20$, $r = 1$ and $s = 40$, $r = 4$ through a single replicate in  Figures \ref{fig:Synthetic_Rank1}, \ref{fig:Synthetic_Rank4_fig1}, and \ref{fig:Synthetic_Rank4_fig2}, respectively. For visualization of recovering $\bU_0$ across different methods, we rotate the estimates according to the Frobenius orthogonal alignment (see section \ref{sub:spiked_covariance_matrix_models} for more details). It can clearly be seen that POET is able to capture the signal but fails to recover the joint sparsity of the principal subspace, whereas SPCA is able to recover the subspace sparsity but is not accurate in estimating the signal. MGPS performs similarly to POET, but its estimated credible intervals are wider than those using the proposed approach. 

Overall, the proposed sparse Bayesian spiked covariance matrix model is able to estimate the signals accurately, recover the row support of $\bU_0$, and provides better uncertainty quantification with narrower credible intervals for simulation setting. 
\begin{figure}[htbp!]
  \centerline{\includegraphics[width=1\textwidth]{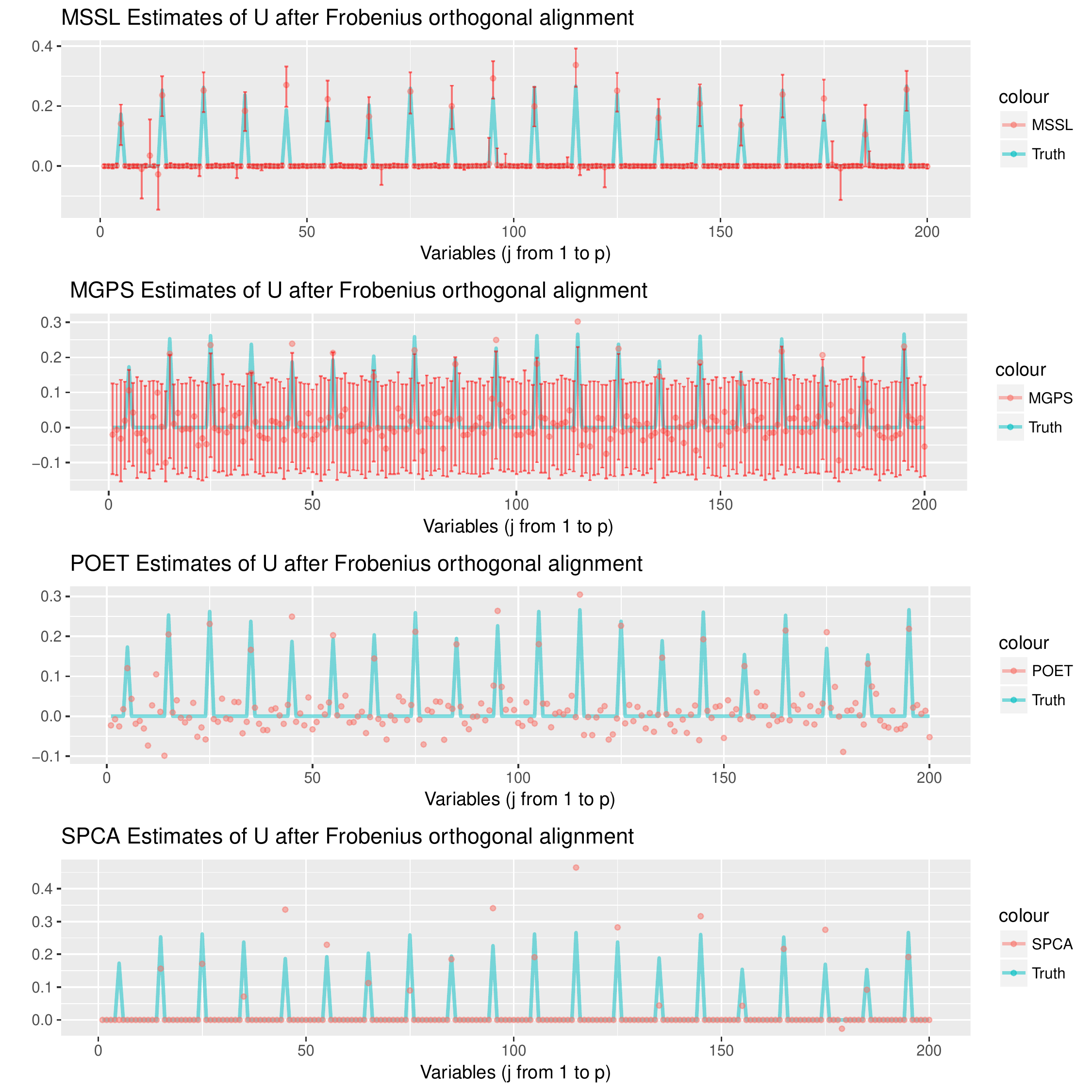}}
  \caption{Simulation performance from a single replicate with $s = 20$ and $r = 1$. The estimates are rotated to the simulation truth $\bU_0$ according to the Frobenius orthogonal alignment. The red bars in the top panels are estimated $95\%$ credible intervals using the proposed approach. MSSL stands for the sparse Bayesian spiked covariance matrix model with the matrix spike-and-slab LASSO prior. }
  \label{fig:Synthetic_Rank1}
\end{figure} 
\begin{figure}[htbp!]
  \centerline{\includegraphics[width=1.05\textwidth]{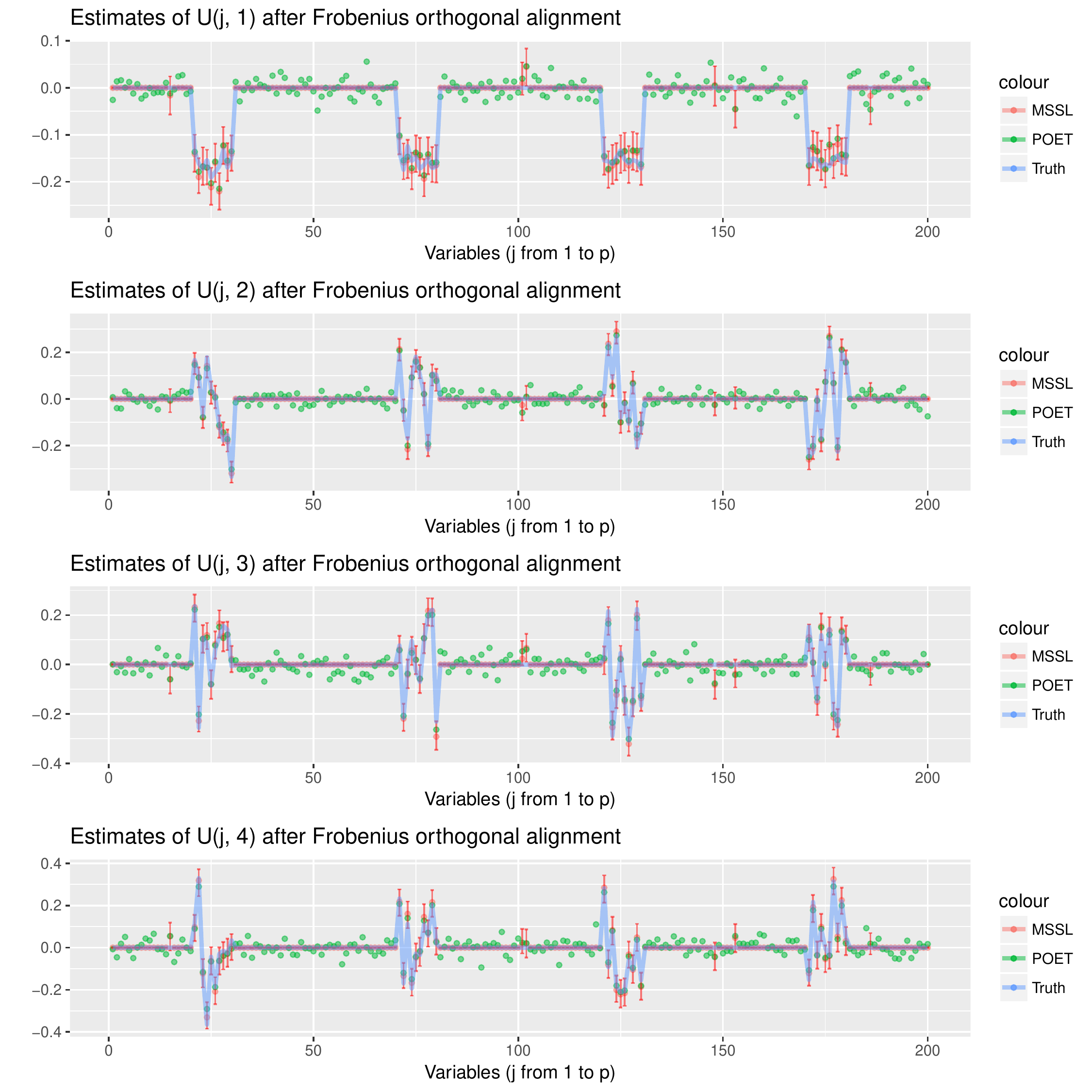}}
  \caption{Simulation performance from a single replicate with $s = 40$ and $r = 4$. The estimates are rotated to the simulation truth $\bU_0$ according to the Frobenius orthogonal alignment. The red bars in the four panels are estimated $95\%$ credible intervals using the proposed approach. }
  \label{fig:Synthetic_Rank4_fig1}
\end{figure}
\begin{figure}[htbp!]
  \centerline{\includegraphics[width=1.05\textwidth]{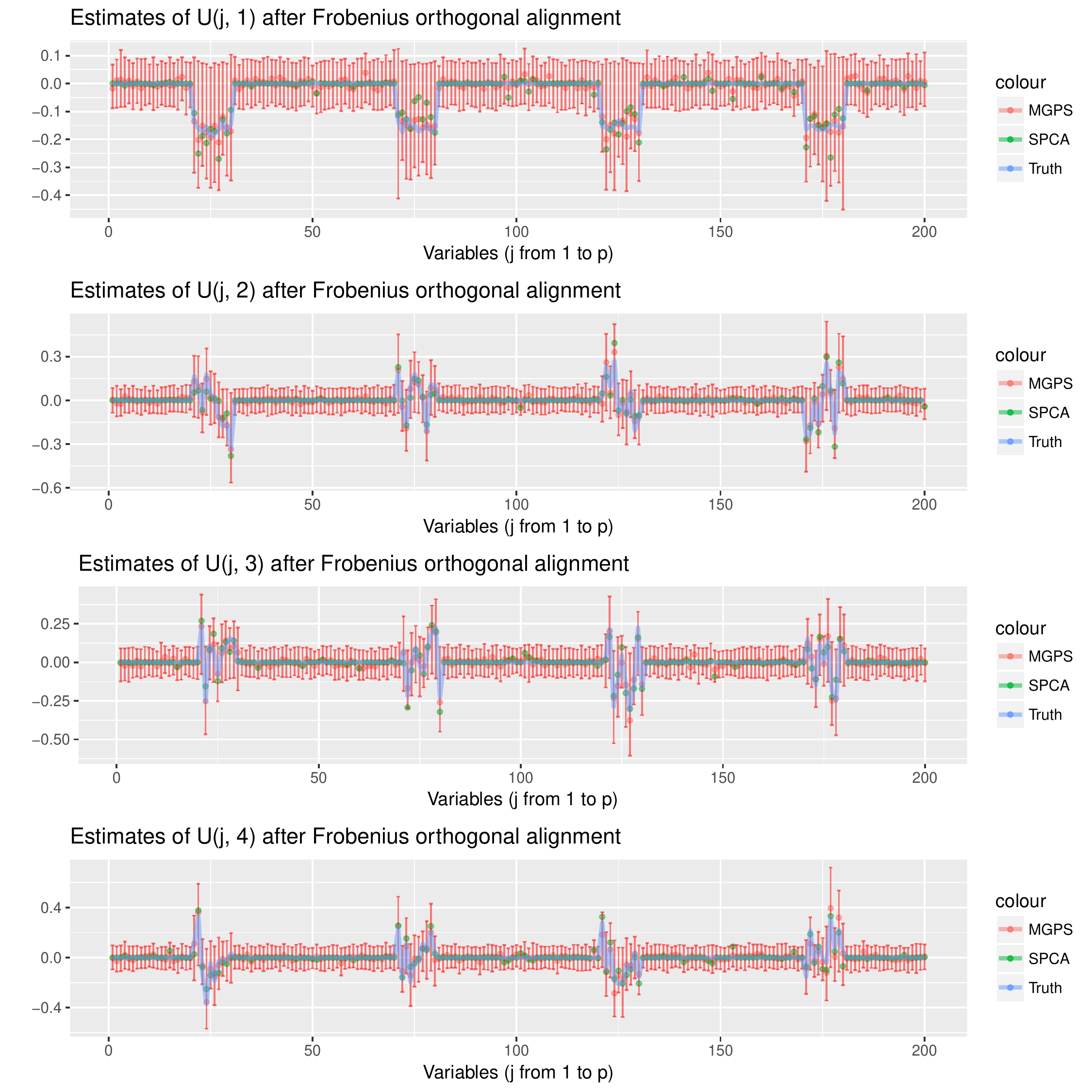}}
  \caption{Simulation performance from a single replicate with $s = 40$ and $r = 4$. The estimates are rotated to the simulation truth $\bU_0$ according to the Frobenius orthogonal alignment. The red bars in the four panels are estimated $95\%$ credible intervals for MGPS.}
  \label{fig:Synthetic_Rank4_fig2}  
\end{figure} 


\subsection{A face data example} 
\label{sub:real_data_example}
The joint sparsity of columns of the eigenvector matrix $\bU$ is highly desired in   feature extraction for high-dimensional data. In this subsection we illustrate how the proposed Bayesian approach is able to extract key features through a real data example in computer vision. 

  We consider a subset of the Extended Yale Face Database B \citep{927464,1407873}. It consists of face images for 38 subjects, and for each subject, 64 aligned images of size $192\times 168$ are taken under different illumination conditions. Here we focus on the 22nd subject and reduce the size of each image to $96\times 84$ (8064 pixels in total), following \cite{doi:10.1093/biomet/asw059}. In doing so we obtain a data matrix $\bY = [\by_1,\ldots,\by_n]\transpose$ of size $64\times 8064$. 

In computer vision, principal component analysis has been widely applied to obtain low-dimensional features, known as eigenfaces, from high-dimensional face image data. Under the proposed Bayesian framework, we perform posterior inference by implementing a Metropolis-within-Gibbs sampler. 
The number of spikes $r$ is estimated using the diagonal thresholding method proposed in \cite{cai2013sparse}.
For comparison, we also implement MGPS \cite{doi:10.1093/biomet/asr013}. 
 Instead of obtaining eigenfaces, we focus on directly extracting the key pixels via thresholding the obtained estimated eigenvector matrix $\widehat\bU$ using the obtained posterior samples. Specifically, for the proposed approach, the estimate $\widehat\bU$ can be computed according to Theorem \ref{thm:risk_bound_point_estimate}, and for MGPS, $\widehat\bU$ can be obtained by computing the left singular vectors of the loading matrix. The key pixels are then obtained by finding $\{j\in[8064]:\|\widehat\bU_{j*}\|_1/r> \tau\}$ for some small tolerance $\tau>0$. 

We present sample images of the 22nd subject in the first row of Figure \ref{fig:FaceData}, and the key pixels of the sample image $\#1$ extracted under the two models with different threshold values of $\tau$ are provided in the second and the third rows of Figure \ref{fig:FaceData}. 
\begin{figure}[t!]
  \centerline{\includegraphics[width=.9\textwidth]{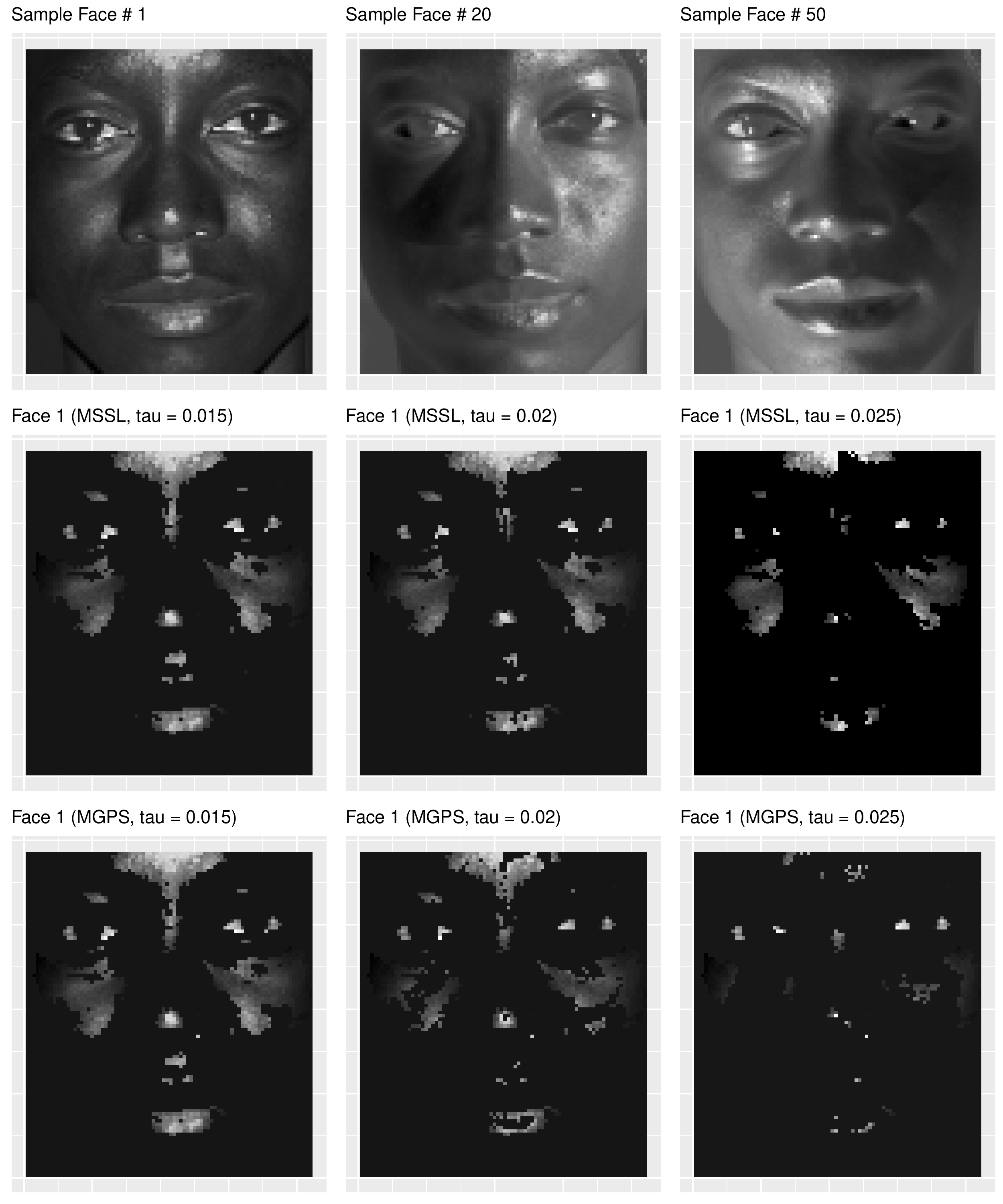}}
  \caption{The face data example: The first row corresponds to sample images of the 22nd subject (image number 1, 20, and 50, respectively).
  The second and the third rows are the key pixels of the $\#1$ image using the proposed Bayesian approach with the matrix spike-and-slab LASSO prior (MSSL) and MGPS with different threshold values of $\tau$. }
  \label{fig:FaceData}
\end{figure}
Under both models, pixels with higher values (corresponding to eyes, cheeks, forehead, and nose tips of the subject) are recovered. This observation is also in accordance with the conclusion from \cite{doi:10.1093/biomet/asw059}. 
{\color{black}
Nevertheless, as the threshold value $\tau$ increases, the number of key pixels captured using MGPS decreases significantly, whereas the proposed approach is more robust to the threshold value $\tau$ and maintains the key pixels that are sensitive to illumination. This phenomenon is expected, since MGPS is not designed to model joint sparsity and feature extraction, but rather column-specific sparsity for each individual factor loading, unlike the matrix spike-and-slab LASSO prior. 
}



\section{Discussion} 
\label{sec:discussion}

We have shown that the two-to-infinity norm loss for principal subspace estimation is superior to the routinely used projection operator norm loss in that the former is able to capture element-wise perturbations of the eigenvector matrix $\bU$ compared to the latter. 
We have derived the contraction rate of the full posterior distribution for the principal subspace with respect to the two-to-infinity norm loss, which is tighter than that with respect to the usual projection operator norm loss, provided that $\bU$ exhibits certain low-rank and bounded coherence features.
In future work, we intend to study whether a point estimator can be found from the posterior distribution with a risk bound that coincides with the posterior contraction rate with respect to the two-to-infinity norm loss. 
In addition, it is also worth exploring the minimax-optimal rates of convergence with respect to the two-to-infinity norm loss. 

Throughout the paper, the number of spikes $r$ is either assumed to be known, or unknown but can be consistently estimated using a frequentist procedure. Alternatively, it is feasible to adaptively estimate $r$ in the literature of Bayesian latent factor models (see, for example, \cite{doi:10.1093/biomet/asr013,gao2015rate,pati2014posterior}). Hence exploring rank-adaptive Bayesian procedure and obtain attractive theoretical properties or computation tractability could also be interesting. 

Markov chain Monte Carlo (MCMC) can be computationally intensive for high-dimensional settings in general. In this paper we 
explored MCMC for Bayesian estimation of the sparse spiked covariance matrix models.  It would be attractive to design efficient computational methods, such as expectation-maximization algorithm for the maximum \emph{a posteriori} estimation instead of computing the full posterior distribution \citep{doi:10.1080/01621459.2015.1100620}, or penalized least-squared estimation \citep{doi:10.1093/biomet/asw059}, and explore the underlying theoretical guarantees in future work. 


\bibliographystyle{apalike}
\bibliography{reference}

\clearpage
\begin{center}
	\begin{Large}
		\textbf{Supplementary Material for ``Bayesian Estimation of Sparse Spiked Covariance Matrix in High Dimensions''}
	\end{Large}
\end{center}
\appendix
\counterwithin{lemma}{section}
\counterwithin{theorem}{section}

\section{Proof of Lemma \ref{lemma:two_to_infinity_and_projection}} 
\label{sec:proof_of_lemma_lemma:two_to_infinity_and_projection}
\begin{customlemma}{2.1}
Let $\bU$ and $\bU_0$ be two orthonormal $r$-frames in $\mathbb{R}^p$, where $2r< p$. Then there exists an orthonormal $2r$-frame $\bV_\bU$ in $\mathbb{R}^p$ depending on $\bU$ and $\bU_0$, such that
\[
\|\bU-\bU_0\bW_\bU\|_{2\to\infty}\leq \|\bV_\bU\|_{2\to\infty}
\left(\|\bU\bU\transpose - \bU_0\bU_0\transpose\|_2 + \|\bU\bU\transpose - \bU_0\bU_0\transpose\|_2^2\right),
\]
where $\bW_\bU = \arginf_{\bW\in\mathbb{O}(r)}\|\bU - \bU_0\bW\|_{\mathrm{F}}$ is the Frobenius orthogonal alignment matrix. 
\end{customlemma}

We will need the following CS matrix decomposition of a partitioned orthonormal matrix to prove Lemma \ref{lemma:two_to_infinity_and_projection}. 
\begin{theorem}[Theorem 5.1 in \citealp{Stewart90}]
Let the orthonormal matrix $\bW \in \mathbb{O}(p, p)$ be partitioned in the form
\[
\bW = \begin{bmatrix*}
\bW_{11}&\bW_{12}\\ \bW_{21}& \bW_{22}
\end{bmatrix*},
\]
where $\bW_{11}\in\mathbb{R}^{r\times r}$, $\bW_{22}\in\mathbb{R}^{(p - r)\times (p - r)}$, and $2r\leq p$. Then there exists orthonormal matrices $\bU = \mathrm{diag}(\bU_{11},\bU_{22})$ and $\bV = \mathrm{diag}(\bV_{11}, \bV_{22})$ with $\bU_{11}, \bV_{11}\in\mathbb{O}(r)$, such that
\[
\bW = \bU\begin{bmatrix*}
\bC & -\bS & \zero \\ \bS & \bC & \zero \\ \zero & \zero & \eye_{(p - 2r)}
\end{bmatrix*}\bV\transpose,
\]
where $\bC = \mathrm{diag}(c_1,\ldots,c_r)$ and $\bS = \mathrm{diag}(s_1,\ldots,s_r)$ are diagonal with non-negative entries, and $\bC^2 + \bS^2 = \eye_r$. 
\end{theorem}

\vspace*{1ex}
\noindent
Let $\bU_\perp$ and $\bU_{0\perp}\in\mathbb{O}(p, p - r)$ be such that $[\bU,\bU_\perp]$ and $[\bU_0,\bU_{0\perp}]\in\mathbb{O}(p)$. By the CS decomposition, there exists $\bU_{11},\bV_{11}\in\mathbb{O}(r)$ and $\bU_{22},\bV_{22}\in\mathbb{O}(p - r)$, such that
\[
\begin{bmatrix*}
\bU_0\transpose\bU & \bU_0\transpose\bU_\perp\\
\bU_{0\perp}\transpose\bU & \bU_{0\perp}\transpose\bU_{\perp}
\end{bmatrix*} = 
\begin{bmatrix*}
\bU_{11} & \zero \\ \zero & \bU_{22}
\end{bmatrix*} 
\begin{bmatrix*}
\bC & -\bS & \zero \\ \bS & \bC & \zero \\ \zero & \zero & \eye_{(p - 2r)}
\end{bmatrix*}
\begin{bmatrix*}
\bV_{11}\transpose & \zero\\ \zero & \bV_{22}\transpose
\end{bmatrix*}
\]
where $\bC = \mathrm{diag}(c_1,\ldots,c_r)$ and $\bS = \mathrm{diag}(s_1,\ldots,s_r)$ are diagonal with non-negative entries, and $\bC^2 + \bS^2 = \eye_r$. Write $\bU_{22}$ into two blocks $\bU_{22} = [\bU_{221},\bU_{222}]$ with $\bU_{221}\in\mathbb{O}(p - r, r)$. Take $\bQ = [\bU_0\bU_{11}, \bU_{0\perp}\bU_{22}]$. Clearly, we have
\begin{align}
\bQ\transpose \bU_0 \bU_{11} = \begin{bmatrix*}
\bU_{11}\transpose & \zero \\ \zero & \bU_{22}\transpose
\end{bmatrix*}
\begin{bmatrix*}
\bU_0\transpose \\ \bU_{0\perp}\transpose
\end{bmatrix*}\bU_0\bU_{11}
= \begin{bmatrix*}
\eye_r\\\zero_r\\ \zero_{p - 2r}
\end{bmatrix*}\nonumber
\end{align}
and
\[
\bQ\transpose\bU \bV_{11} = 
\begin{bmatrix*}
\bU_{11}\transpose & \zero \\ \zero & \bU_{22}\transpose
\end{bmatrix*}
\begin{bmatrix*}
\bU_0\transpose\bU \\ \bU_{0\perp}\transpose\bU
\end{bmatrix*}\bV_{11}
= \begin{bmatrix*}
\bC\\\bS\\ \zero_{p - 2r}
\end{bmatrix*}
\]
Observe that $\|\bU\bU\transpose - \bU_0\bU_0\transpose\|_2 = \|\bS\|_2$, and that $\bU_0\transpose \bU = \bU_{11}\bC\bV_{11}\transpose$ is the singular value decomposition of $\bU_0\transpose\bU$, implying that $\bW_\bU = \bU_{11}\bV_{11}\transpose$. We proceed to compute
\begin{align}
\|\bU - \bU_0\bW_\bU\|_{2\to\infty}
&\leq \|\bU - \bU_0\bU_0\transpose\bU\|_{2\to\infty} + \|\bU_0(\bU_0\transpose \bU - \bW_\bU)\|_{2\to\infty}\nonumber\\
& = \left\|\bQ \begin{bmatrix*}
\bS^2\\-\bS\bC\\\zero
\end{bmatrix*}\right\|_{2\to\infty} + \left\|\bQ \begin{bmatrix*}
\bC(\bC - \eye_r)\\-\bS(\bC - \eye_r)\\\zero
\end{bmatrix*}\right\|_{2\to\infty}\nonumber\\
& = \left\|
\begin{bmatrix*}
\bU_0 & \bU_{0\perp}
\end{bmatrix*}
\begin{bmatrix*}
\bU_{11} & \zero & \zero \\ \zero & \bU_{221} & \bU_{222}
\end{bmatrix*}
\begin{bmatrix*}
\bS^2\\-\bS\bC\\\zero
\end{bmatrix*}\right\|_{2\to\infty} \nonumber\\
&\quad + \left\|\begin{bmatrix*}
\bU_0 & \bU_{0\perp}
\end{bmatrix*}
\begin{bmatrix*}
\bU_{11} & \zero & \zero \\ \zero & \bU_{221} & \bU_{222}
\end{bmatrix*}
\begin{bmatrix*}
\bC(\bC - \eye_r)\\-\bS(\bC - \eye_r)\\\zero
\end{bmatrix*}\right\|_{2\to\infty}\nonumber\\
& = \left\|
\begin{bmatrix*}
\bU_0 & \bU_{0\perp}
\end{bmatrix*}
\begin{bmatrix*}
\bU_{11} & \zero & \\ \zero & \bU_{221}
\end{bmatrix*}
\begin{bmatrix*}
\bS^2\\-\bS\bC
\end{bmatrix*}\right\|_{2\to\infty} 
\nonumber\\&\quad 
+ \left\|\begin{bmatrix*}
\bU_0 & \bU_{0\perp}
\end{bmatrix*}
\begin{bmatrix*}
\bU_{11} & \zero & \\ \zero & \bU_{221} &
\end{bmatrix*}
\begin{bmatrix*}
\bC(\bC - \eye_r)\\-\bS(\bC - \eye_r)
\end{bmatrix*}\right\|_{2\to\infty}\nonumber\\
& = \left\|
\begin{bmatrix*}
\bU_0\bU_{11} & \bU_{0\perp}\bU_{221}
\end{bmatrix*}
\begin{bmatrix*}
\bS^2\\-\bS\bC
\end{bmatrix*}\right\|_{2\to\infty} 
 + \left\|\begin{bmatrix*}
\bU_0\bU_{11} & \bU_{0\perp}\bU_{221}
\end{bmatrix*}
\begin{bmatrix*}
\bC(\bC - \eye_r)\\-\bS(\bC - \eye_r)
\end{bmatrix*}\right\|_{2\to\infty}\nonumber.
\end{align}
Denote $\bV_\bU = [\bU_0\bU_{11}, \bU_{0\perp}\bU_{221}]$. Clearly, $\bV_\bU\in\mathbb{O}(p, 2r)$:
\[
\bV_\bU\transpose \bV_\bU = 
\begin{bmatrix*}
\bU_{11}\transpose & \zero & \\ \zero & \bU_{221}\transpose &
\end{bmatrix*}
\begin{bmatrix*}
\bU_0\transpose\\ \bU_{0\perp}\transpose
\end{bmatrix*}
\begin{bmatrix*}
\bU_0 & \bU_{0\perp}
\end{bmatrix*}
\begin{bmatrix*}
\bU_{11} & \zero & \\ \zero & \bU_{221} &
\end{bmatrix*} = \eye_{2r}.
\]
Furthermore, by the previous derivation and the fact that $\|\bA\bB\|_{2\to\infty}\leq \|\bA\|_{2\to\infty}\|\bB\|_2$, we have
\begin{align}
\|\bU - \bU_0\bW_\bU\|_{2\to\infty}
&\leq \|\bV_\bU\|_{2\to\infty}\left(\left\|\begin{bmatrix*}
\bS^2\\-\bS\bC
\end{bmatrix*}\right\|_{2} + \left\|\begin{bmatrix*}
\bC(\bC - \eye_r)\\-\bS(\bC - \eye_r)
\end{bmatrix*}\right\|_{2}\right)\nonumber\\
& = \|\bV_\bU\|_{2\to\infty}\left(
\left\|\bS^4  + \bS\bC^2\bS\right\|_2^{1/2} + \left\|(\bC - \eye_r)^2\right\|_2^{1/2}
\right)\nonumber\\
& = \|\bV_\bU\|_{2\to\infty}\left(
\left\|\bS\right\|_2 + \left\|\eye_r - \bC\right\|_2
\right)\nonumber\\
& \leq \|\bV_\bU\|_{2\to\infty}\left(
\left\|\bS\right\|_2 + \left\|\eye_r - \bC^2\right\|_2
\right)\nonumber\\
& = \|\bV_\bU\|_{2\to\infty}\left(\left\|\bU\bU\transpose - \bU_0\bU_0\transpose\right\|_2 + \left\|\bU\bU\transpose - \bU_0\bU_0\transpose\right\|_2^2\right)\nonumber,
\end{align}
and the proof is thus completed.

\section{Proofs of Results in Section \ref{sub:posterior_contraction_for_bayesian_sparse_spiked_covariance_matrices}} 
\label{sec:proofs_of_results_in_section_3_2}

\begin{customthm}{3.1}
Assume the data $\by_1,\ldots,\by_n$ are independently sampled from $\mathrm{N}_p(\zero_p, \bSigma_0)$ with $\bSigma_0 = \bU_0\bLambda_0\bU_0\transpose + \sigma_0^2\eye_p$, 
$\bLambda_0 = \mathrm{diag}(\lambda_{01},\ldots,\lambda_{0r})$, 
$|\mathrm{supp}(\bU_0)|\leq s$, and $1\leq r\leq s\leq p$. 
Suppose $\eps_n^2 = (s\log p)/n \to 0$, $p/n\to\infty$, and $r\log n\lesssim \log p$. Let $\bB\sim\mathrm{MSSL}_{p\times r}(\lambda, 1/p^2, p^{1+\kappa})$ for some positive $\lambda>0$ and $\kappa \leq 1$, and $\sigma^2\sim\mathrm{IGamma}(a_\sigma,b_\sigma)$ for some $a_\sigma,b_\sigma\geq1$. Then there exists some constants $M_0>0$, $R_0$, and $C_0$ depending on $\sigma_0$ and $\bLambda_0$, and hyperparameters, such that the following posterior contraction for $\bSigma = \bB\bB\transpose + \sigma^2\eye_p$ holds for all $M\geq M_0$ when $n$ is sufficiently large:
\begin{align}
\expect_0\left\{\Pi\left(\|\bSigma-\bSigma_0\|_2 > M\eps_n\mathrel{\Big|}\bY_n\right)\right\}
&\leq R_0\exp(-C_0s\log p).\nonumber
\end{align}
For each $\bB$, let $\bU_\bB\in\mathbb{O}(p, r)$ be the left-singular vector matrix of $\bB$. Then the following posterior contraction for $\bU_\bB$ holds for all $M\geq M_0$:
\begin{align}
\expect_0\left\{\Pi\left(\|\bU_\bB\bU_\bB\transpose-\bU_0\bU_0\transpose\|_{2} > \frac{2M\eps_n}{\lambda_{0r}}\mathrel{\Big|}\bY_n\right)\right\}
&\leq R_0\exp(-C_0s\log p).\nonumber
\end{align}
\end{customthm}

\begin{proof}[{\bf Proof of Theorem \ref{thm:posterior_contraction}}]
Recall that $\calU_n = \{\|\bSigma - \bSigma_0\|_2\leq M\eps_n\}$ and the posterior probability $\Pi(\calU_n^c\mid\bY_n)$ can be written as $\Pi(\calU_n^c\mid\bY_n) = {N_n(\calU_n^c)}/{D_n}$,  where 
\[
N_n(\calU_n^c) = \int_\calA \exp\{\ell_n(\bSigma) - \ell_n(\bSigma_0)\}\Pi(\mathrm{d}\bSigma),\quad D_n = \int \exp\{\ell_n(\bSigma) - \ell_n(\bSigma_0)\}\Pi(\mathrm{d}\bSigma),
\]
and $\ell_n(\bSigma) = \sum_{i=1}^n\log p(\by_i\mid\bSigma)$ is the log-likelihood function of $\bSigma$. 

\vspace*{1ex}
\noindent\textbf{Step 1: Prior concentration. }Let $\eta_n = \sqrt{(s\log p)/n}$. Then by Lemma 
\ref{lemma:evidence_lower_bound}, there exists a sequence of events $(\calA_n)_{n=1}^\infty$ such that
\[
\calA_n\subset\left\{D_n\geq\Pi(\|\bSigma-\bSigma_0\|_{\mathrm{F}}\leq\eta_n)\exp\left(-C_3's\log p\right)\right\}
\]
for $\eta_n = \sqrt{(s\log p)/n}\leq \sigma_0^2/2$, and 
\begin{align}\label{eqn:prior_concentration_probabilistic_bound}
\prob_0(\calA_n^c)\leq2\exp\left\{-\tilde C_3\min\left(1, \|\bSigma_0^{-1}\|_2^{-2}\right)s\log p\right\},
\end{align}
where $C_3'$ and $\tilde C_3$ are some absolute constants. Denote $\bB_0 = \bU_0\bLambda_0^{1/2}$, where $\bLambda_0^{1/2} = \mathrm{diag}(\lambda_{01}^{1/2},\ldots,\lambda_{0r}^{1/2})$. Then we analyze the prior concentration using a union bound as follows:
\begin{align}
\Pi(\|\bSigma-\bSigma_0\|_{\mathrm{F}}\leq \eta_n)& \geq \Pi\left(
\|\bB\bB\transpose- \bB_0\bB_0\transpose\|_{\mathrm{F}} + \|\sigma^2\eye_p - \sigma_0^2\eye_p\|_{\mathrm{F}} \leq \eta_n
\right)\nonumber\\
&\geq \Pi\left(\|\bB\bB\transpose - \bB_0\bB_0\transpose\|_{\mathrm{F}}\leq\frac{\eta_n}{2}\right)\Pi\left(|\sigma_0^2 - \sigma^2|\leq\frac{\eta_n}{2\sqrt{p}}\right)\nonumber.
\end{align}
On one hand, for $\eta_n = \sqrt{(s\log p)/n}\leq \sigma_0^2/2$, we have
\begin{align}
\Pi\left(|\sigma_0^2 - \sigma^2|\leq\frac{\eta_n}{2\sqrt{p}}\right)&\geq \left\{\min_{\sigma\in[\sigma_0^2/2,3\sigma_0^2/2]}\pi_\sigma(\sigma^2)\right\}\frac{\eta_n}{\sqrt{p}}
\geq C(\sigma_0^2)\mathrm{e}^{-\log p}\nonumber,
\end{align}
where the constant $C(\sigma_0^2) = \min_{\sigma_0^2/2\leq\sigma^2\leq 3\sigma_0^2/2}\pi_\sigma(\sigma^2) > 0$ depends only on $\sigma_0^2$. 
On the other hand, for $\eta_n = \sqrt{(s\log p)/n}\leq\min(\sigma_0^2/2, 16\|\bB_0\|^{1/2}_2)$, we proceed by union bound to derive
\begin{align}
\Pi\left(\|\bB\bB\transpose - \bB_0\bB_0\transpose\|_{\mathrm{F}}\leq\frac{\eta_n}{2}\right)
&\geq \Pi\left(\|\bB - \bB_0\|_{\mathrm{F}}\|\bB-\bB_0+\bB_0\|_2 + \|\bB_0\|_2\|\bB - \bB_0\transpose\|_{\mathrm{F}}\leq\frac{\eta_n}{2}\right)
\nonumber\\&
\geq \Pi\left\{\|\bB - \bB_0\|_{\mathrm{F}}\left(\|\bB-\bB_0\|_{\mathrm{F}} + 2\|\bB_0\|_2\right)\leq\frac{\eta_n}{2}\right\}\nonumber\\&
\geq \Pi\left\{\|\bB-\bB_0\|_{\mathrm{F}}\leq \min\left(\frac{\eta_n}{8\|\bB_0\|_2}, 2\|\bB_0\|_2\right)\right\}\nonumber\\&
= \Pi\left(\|\bB-\bB_0\|_{\mathrm{F}}\leq\frac{\eta_n}{8\|\bB_0\|_2}\right)\nonumber.
\end{align}
Invoking Lemma \ref{lemma:prior_concentration}, we see that there exists some constant $C(\lambda,\bB_0)$ depending on $\lambda$ and $\|\bB_0\|_{2\to\infty}$ only, such that
\begin{align}
\Pi\left(\|\bB\bB\transpose - \bB_0\bB_0\transpose\|_{\mathrm{F}}\leq\frac{\eta_n}{2}\right)
&\geq
\Pi\left(\|\bB-\bB_0\|_{\mathrm{F}}\leq\frac{\eta_n}{8\|\bB_0\|_2}\right)\nonumber\\
&\geq
\exp\left[-C_1\max\left\{\lambda^2s\|\bB_0\|_{2\to\infty}^2, s\log p, sr\left|\log\left(\lambda\frac{\sqrt{\log p}}{\sqrt{rn}}\right)\right|\right\}\right]\nonumber\\
&\geq \exp\left\{-C(\lambda,\bB_0)s\log p\right\}\nonumber.
\end{align}
Therefore, for $\eta_n = \sqrt{(s\log p)/n}\leq \min(\sigma_0^2,16\|\bB_0\|_2^{1/2})$ we obtain
\[
\Pi(\|\bSigma-\bSigma_0\|_{\mathrm{F}}\leq\eta_n)\geq C(\sigma_0^2)\exp\left[-\{1 + C(\lambda,\bB_0)\}s\log p\right],
\]
and over $\calA_n$, we have
\begin{align}\label{eqn:evidence_lower_bound}
D_n\geq C(\sigma_0^2)\exp\left(-C_{0\lambda}s\log p\right)
\end{align}
for some constant $C_{0\lambda}$ depending only on $\lambda$ and $\|\bB_0\|_{2\to\infty}$. 

\vspace*{1ex}\noindent
\textbf{Step 2: Construct subsets $(\calF_n)_{n=1}^\infty$. }Take $\eps_n = \sqrt{(s\log p)/n}$, $\tau_n = \beta s_n$, $t_n = (sr\log p)^2$, and $\delta_n = \eps_n/(t_n\sqrt{p})$, where $\beta>0$ is some constant to be specified later. Clearly, there exists some $\gamma>0$ such that
\[
\delta_n = \frac{\eps_n}{t_n\sqrt{p}} = \frac{\sqrt{s\log p}}{\sqrt{np}(sr\log p)^2} = 
\frac{1}{\sqrt{nps^3r^4(\log p)^3}}\geq \frac{1}{p^\gamma}.
\]
Now let $\beta>4\mathrm{e}\gamma$ and $\calF_n = \calF(\delta_n,\tau_n,t_n)$ be defined in Lemma \ref{lemma:existence_test}. 
Since
\begin{align}
\min\left\{\left(\frac{t_n}{\beta sr}\right)^2, \left(\frac{t_n}{r}\right)^2,\frac{t_n}{r}\right\}& = 
\min\left\{\frac{(sr)^2(\log p)^4}{\beta^2}, s^4r^2(\log p)^4, s^{2}r(\log p)\right\}\nonumber\\
& = \min\left\{\frac{sr^2(\log p)^3}{\beta^2}, s^3r^2(\log p)^3, sr\log p\right\}s\log p\nonumber\\
&\geq \beta s\log p\nonumber
\end{align}
for sufficiently large $n$, 
and $t_n/(sr) = (sr)\log p\to\infty$, we then can invoke Lemmas \ref{lemma:prior_sparsity} and \ref{lemma:prior_deviation} to obtain
\begin{align}
\Pi(\calF_n^c)&\leq \Pi(|\mathrm{supp}_{\delta_n}(\bB)|>\beta s_n) + \Pi\left[\sum_{j=1}^p\|\bB_{j*}\|_2^2\mathbbm{1}\{j\in\mathrm{supp}_{\delta_n}(\bB)\cup\mathrm{supp}(\bU_0)\} > t_n^2\right]\nonumber\\
&\leq 2\exp(-\beta s\log p) + 5\exp\left\{
-\min\left(\frac{\beta\kappa}{2},\frac{\beta}{2\mathrm{e}} - 2\gamma\right)s\log p
\right\}\nonumber\\
\label{eqn:prior_concentration_upper_bound}
&\leq 7\exp\left\{
-\min\left(\frac{\beta\kappa}{2},\frac{\beta}{2\mathrm{e}} - 2\gamma\right)s\log p
\right\}
\end{align}
for sufficiently large $n$ (and hence sufficiently small $s\log p/p$). 

\vspace*{1ex}\noindent
\textbf{Step 3: Decompose the integral $\expect_0\{\Pi(\calU_n^c\mid\bY_n)\}$. }Since by construction we have
\[
(\sqrt{p}\delta_n + 2t_n)\sqrt{p}\delta_n = \left(\sqrt{p}\frac{\eps_n}{t_n\sqrt{p}} + 2t_n\right)\sqrt{p}\delta_n \leq 3t_n\sqrt{p}\delta_n  = 3\eps_n.
\]
Then by Lemma \ref{lemma:existence_test}, for each $M\geq \max\{3/2, (128\|\bSigma_0\|_2^4)^{1/3}\}$, there exists a test function $\phi_n$ such that 
\begin{align}
\label{eqn:testing_typeI_error}
\expect_0(\phi_n)&\leq 3\exp\left[- \left\{\frac{C_4\sqrt{M}}{\sqrt{2}} - (2 + C_4)(\beta + 2)\right\}s\log p\right],\\
\label{eqn:testing_typeII_error}
\sup_{\bSigma\in \calU_n^c\cap\calF_n}\expect_\bSigma(1-\phi_n)&\leq \exp\left[- \left\{\frac{C_4M}{8} - C_4(\beta + 2)\right\}s\log p\right]
\end{align}
for some absolute constant $C_4>0$ for sufficiently large $n$. Now we decompose the target integral $\expect_0\{\Pi(\calU_n^c\mid\bY_n)\}$ using \eqref{eqn:prior_concentration_probabilistic_bound} and \eqref{eqn:testing_typeI_error} as follows:
\begin{align}
\expect_0\{\Pi(\calU_n^c\mid\bY_n)\}&\leq
\expect_0(\phi_n) + \expect_0\left\{(1-\phi_n)\Pi(\calU_n\mid\bY_n)\mathbbm{1}(\calA_n)\right\} + \prob_0(\calA_n^c)\nonumber\\
&\leq 3\exp\left[- \left\{\frac{C_4\sqrt{M}}{\sqrt{2}} - (2 + C_4)(\beta + 2)\right\}s\log p\right]
+ 2\exp\left\{-\tilde C_3\min\left(1,\|\bSigma_0^{-1}\|_2^{-2}\right)s\log p\right\} \nonumber\\
&\quad + \expect_0\left[(1 - \phi_n)\left\{\frac{N_n(\calU_n^c)}{D_n}\right\}\mathbbm{1}(\calA_n)\right]\nonumber.
\end{align}
Now we focus on the third term on the right-hand side of the preceding display. By \eqref{eqn:evidence_lower_bound}, we obtain
\begin{align}
&\expect_0\left[(1 - \phi_n)\left\{\frac{N_n(\calU_n^c)}{D_n}\right\}\mathbbm{1}(\calA_n)\right]
\leq \frac{\exp\left(C_{0\lambda}s\log p\right)}{C(\sigma_0^2)}\expect_0\left\{
(1 - \phi_n)\int_{\calU_n^c}\prod_{i=1}^n\frac{p(\by_i\mid\bSigma)}{p(\by_i\mid\bSigma_0)}\Pi(\mathrm{d}\bSigma)
\right\}\nonumber.
\end{align}
Observe that by Fubini's theorem,
\begin{align}
&\expect_0\left\{
(1 - \phi_n)\int_{\calU_n^c}\prod_{i=1}^n\frac{p(\by_i\mid\bSigma)}{p(\by_i\mid\bSigma_0)}\Pi(\mathrm{d}\bSigma)
\right\}\nonumber\\
&\quad\leq 
\expect_0\left\{
(1 - \phi_n)\int_{\calU_n^c\cap \calF_n}\prod_{i=1}^n\frac{p(\by_i\mid\bSigma)}{p(\by_i\mid\bSigma_0)}\Pi(\mathrm{d}\bSigma)
\right\}
+ \expect_0\left\{\int_{\calF_n^c}\prod_{i=1}^n\frac{p(\by_i\mid\bSigma)}{p(\by_i\mid\bSigma_0)}\Pi(\mathrm{d}\bSigma)
\right\}\nonumber\\
&\quad = \int_{\calU_n^c\cap \calF_n}\expect_0\left\{
(1 - \phi_n)\prod_{i=1}^n\frac{p(\by_i\mid\bSigma)}{p(\by_i\mid\bSigma_0)}
\right\}\Pi(\mathrm{d}\bSigma)
+ \int_{\calF_n^c}\left\{\expect_0\prod_{i=1}^n\frac{p(\by_i\mid\bSigma)}{p(\by_i\mid\bSigma_0)}\right\}\Pi(\mathrm{d}\bSigma)\nonumber\\
&\quad \leq \int_{\calU_n^c\cap\calF_n}\expect_{\bSigma}(1-\phi_n)\Pi(\mathrm{d}\bSigma) + 
\Pi(\calF_n^c)\nonumber\\
&\quad\leq \exp\left[- \left\{\frac{C_4M}{8} - C_4(\beta + 2)\right\}s\log p\right] + 
7\exp\left\{
-\min\left(\frac{\beta\kappa}{2},\frac{\beta}{2\mathrm{e}} - 2\gamma\right)s\log p
\right\}\nonumber,
\end{align}
where the testing type II error probability bound \eqref{eqn:testing_typeII_error} and \eqref{eqn:prior_concentration_upper_bound} are applied to the last inequality. Then by taking
\begin{align}
\beta &= \max\left\{\frac{4}{\kappa}C_{0\lambda}, 2\mathrm{e}\left(2\gamma + 2C_{0\lambda}\right)\right\},\nonumber\\
M &= M_0 = \max\left[\frac{8}{C_4}\left\{C_4(\beta + 2) + 2C_{0\lambda}\right\},\frac{2}{C_4^2}
\left\{C_{0\lambda} + (2 + C_4)(\beta + 2)\right\}^2
\right]
\nonumber,
\end{align}
we obtain the following result:
\begin{align}
\expect_0\left[(1 - \phi_n)\left\{\frac{N_n(\calU_n^c)}{D_n}\right\}\mathbbm{1}(\calA_n)\right]
&\leq \frac{1}{C(\sigma_0^2)}\exp\left[- \left\{\frac{C_4M}{8} - C_4(\beta + 2) - C_{0\lambda}\right\}s\log p\right]\nonumber\\
&\quad + \frac{1}{C(\sigma_0^2)}
7\exp\left[
-\left\{\min\left(\frac{\beta\kappa}{2},\frac{\beta}{2\mathrm{e}} - 2\gamma\right)
 - C_{0\lambda}
\right\}
s\log p
\right]\nonumber\\
&\leq \frac{8}{C(\sigma_0^2)}\exp\left\{-C_{0\lambda}s\log p\right\}.\nonumber
\end{align}
Combining the above results, we finally obtain
\begin{align}
\expect_0\{\Pi(\calU_n^c\mid\bY_n)\}&\leq
\left\{3 + \frac{8}{C(\sigma_0^2)}\right\}\exp\left\{ - C_{0\lambda}s\log p\right\}
+ 2\exp\left\{-\tilde C_3\min\left(1, \left\|\bSigma_0^{-1}\right\|_2^{-2}\right)s\log p\right\}\nonumber\\
&\leq \left\{5 + \frac{11}{C(\sigma_0^2)}\right\}\exp\left[
-\min\left\{ C_{0\lambda}, \tilde C_3, \tilde C_3\|\bSigma_0^{-1}\|_2^{-2} \right\}s\log p
\right]\nonumber\\
& = R_0\exp(-C_0s\log p)\nonumber
\end{align}
by taking $C_0 = \min\left\{ C_{0\lambda}, \tilde C_3, \tilde C_3\|\bSigma_0^{-1}\|_2^{-2} \right\}$ and $R_0 = \left\{5 + {11}/{C(\sigma_0^2)}\right\}$. Therefore, there exists some constant $M_0$, such that for all sufficiently large $n$, we have
\[
\expect_0\left\{\Pi\left(\|\bSigma-\bSigma_0\|_2 > M\eps_n\mathrel{\Big|}\bY_n\right)\right\}
\leq
\expect_0\left\{\Pi\left(\|\bSigma-\bSigma_0\|_2 > M_0\eps_n\mathrel{\Big|}\bY_n\right)\right\}
\leq R_0\exp(-C_0s\log p)
\]
for some absolute constants $C_0$ and $R_0$ depending on $\bSigma_0$ and the hyperparameters only. 

\vspace*{1ex}\noindent
\textbf{Step 4: Bounding the projection operator norm loss using the sine-theta theorem. }To prove the posterior contraction for $\bU$ with respect to the projection operator norm loss \eqref{eqn:posterior_contraction_Projection}, we need the following version of the Davis-Kahan sine-theta theorem, which follows as a recasting of Theorem VII.3.7 in \cite{bhatia2013matrix} in the language of \cite{doi:10.1093/biomet/asv008}:
\begin{theorem}
\label{thm:Davis_Kahan_sin_theta}
Let $\bX$, $\widehat\bX\in\mathbb{R}^{p\times p}$ be symmetric matrices with eigenvalues $\lambda_1\geq\ldots\geq \lambda_p$ and $\hat\lambda_1\geq\ldots\geq\hat\lambda_p$, respectively. Write $\bE = \widehat\bX - \bX$ and fix $1\leq r\leq s\leq p$. Assume that $\delta_{\mathrm{gap}}: = \min(\lambda_{r-1}-\lambda_r, \lambda_s - \lambda_{s + 1})>0$ where $\lambda_0:=\infty$ and $\lambda_{p+1}:=-\infty$. Let $d = s - r + 1$ and let $\bV = [\bv_1,\ldots,\bv_s]\in\mathbb{R}^{p\times d}$ and $\widehat\bV = [\hat\bv_r,\ldots,\hat\bv_s]\in\mathbb{R}^{p\times d}$ have orthonormal columns satisfying $\bX\bv_j = \lambda_j\bv_j$ and $\widehat\bX\hat\bv_j = \hat\lambda_j\hat\bv_j$ for $j = r,r + 1,\ldots,s$. Then 
\[
\|\widehat\bV\widehat\bV\transpose - \bV\bV\transpose\|_2\leq \frac{2}{\delta_{\mathrm{gap}}}\|\bE\|_2.
\]
\end{theorem}
\noindent
To apply the sine-theta theorem, we let $\bX = \bSigma_0 = \bU_0\bLambda_0\bU_0\transpose + \sigma_0^2\eye_p$, $\widehat\bX = \bB\bB\transpose + \sigma^2\eye_p$, and take ``$s$''$ = r$ and ``$r$''$ = 1$, in which case $\delta_{\mathrm{gap}} = \min\{\infty, \lambda_r(\bSigma_0) - \lambda_{r + 1}(\bSigma_0)\} = \lambda_{0r}$, $\bV = \bU_0$, $\widehat\bV = \bU_\bB$, and $\bE = \bSigma - \bSigma_0$. Then by the sine-theta theorem and \eqref{eqn:posterior_contraction_Sigma}, we have
\[
\|\bU_\bB\bU_\bB\transpose-\bU_0\bU_0\transpose\|_{2}\leq \frac{2}{\lambda_{0r}}\|\bE\|_2= \frac{2}{\lambda_{0r}}\|\bSigma - \bSigma_0\|_2
\]
and hence, by the posterior contraction for $\bSigma$, we have
\begin{align}
\expect_0\left\{\Pi\left(\|\bU_\bB\bU_\bB\transpose-\bU_0\bU_0\transpose\|_{2} > \frac{2M\eps_n}{\lambda_{0r}}\mathrel{\Big|}\bY_n\right)\right\}
&\leq \expect_0\left\{\Pi\left(\|\bSigma-\bSigma_0\|_2 > M\eps_n\mathrel{\Big|}\bY_n\right)\right\}
\leq R_0\exp(-C_0s\log p)\nonumber.
\end{align}
\end{proof}

\begin{customthm}{3.2}
Assume the conditions in Theorem \ref{thm:posterior_contraction} hold. 
Further assume that the eigenvector matrix $\bU_0$ exhibits bounded coherence: $\|\bU_0\|_{2\to\infty}\leq C_\mu\sqrt{r/s}$ for some constant $C_\mu\geq1$, and the number of spikes $r$ is sufficiently small in the sense that $r^3/s= O(1)$. 
Then there exists some constants $M_{2\to\infty}>0$ depending on $\sigma_0$ and $\bLambda_0$, and hyperparameters, such that the following posterior contraction for $\bSigma = \bB\bB\transpose + \sigma^2\eye_p$ holds for all $M\geq M_{2\to\infty}$ when $n$ is sufficiently large:
\begin{align*}
\expect_0\left\{\Pi\left(\|\bSigma-\bSigma_0\|_\infty > M r\sqrt{\frac{s\log p}{n}}\mathrel{\bigg|}\bY_n\right)\right\}
&\leq R_0\exp(-C_0 s\log p),\\
\expect_0\left[\Pi\left\{\|\bU_\bB-\bU_0\bW_\bU\|_{2\to\infty} > M
\max\left(\sqrt{\frac{r^3\log p}{n}},{\frac{s\log p}{n}}\right)
\right\}
\right]
&\leq 2R_0\exp(-C_0s\log p),
\end{align*}
where $\bW_\bU$ is the Frobenius orthogonal alignment matrix $\bW_\bU = \arginf_{\bW\in\mathbb{O}(r)}\|\bU_\bB-\bU_0\bW\|_{\mathrm{F}}$.
\end{customthm}

\begin{proof}[\bf Proof of Theorem \ref{thm:posterior_contraction_two_to_infinity}]
The proof is similar to that of Theorem \ref{thm:posterior_contraction}, but we need the following testing lemma dealing with the infinity norm loss $\|\bSigma - \bSigma_0\|_\infty$, which is analogous to Lemma \ref{lemma:existence_test} in the manuscript. The proof is deferred to Section \ref{sec:proof_of_lemma_test_infinity_norm}.

\begin{lemma}\label{lemma:existence_test_infinity_norm}
Assume the data $\by_1,\ldots,\by_n$ follows $\mathrm{N}_p(\zero_p, \bSigma)$, $1\leq r\leq p$. Suppose $\bU_0\in\mathbb{O}(p, r)$ satisfy $|\mathrm{supp}(\bU_0)|\leq s$, and $r\leq s\leq p$. 
For any positive $\delta$, $t$, and $\tau$, define 
\[
\calG(\delta,\tau,t) = \left\{\bB\in\mathbb{R}^{p\times r}:|\mathrm{supp}_\delta(\bB)|\leq \tau, \sum_{j=1}^p\|\bB_{j*}\|_1\mathbbm{1}\{j\in\mathrm{supp}_\delta(\bB)\cup \mathrm{supp}(\bU_0)\}\leq t
\right\}.
\]
Let the positive sequences $(\delta_n, \tau_n, t_n,\eps_n)_{n=1}^\infty$ satisfy $\max(p\delta_nt_n, \delta_nt_n + p\delta_n^2)\leq M_1 \eps_n$ for some constant $M_1>0$, and $\eps_n\leq 1$. Consider testing $H_0:\bSigma = \bSigma_0 = \bU_0\bLambda_0\bU_0\transpose + \sigma_0^2\eye_p$ versus
\[
H_1:\bSigma \in\left\{\bSigma = \bB\bB\transpose + \sigma^2\eye_p:\|\bSigma-\bSigma_0\|_{\infty}>M\eps_n, \bB\in\calG(\delta_n, \tau_n, t_n)\right\}.
\]
Then there exists some absolute constant $C_6 > 0$, such that for each 
\[
M\in\left[\max\left\{\frac{M_1}{2}, 8, \frac{8(\log 2)^2}{C_6^2}\right\}, \frac{2\min(1, 2\|\bSigma_0\|_2)}{\eps_n}\right],
\]
 there exists a test function $\phi_n:\mathbb{R}^{n\times p}\to [0,1]$, such that
\begin{align}
\expect_0(\phi_n)&\leq 
12\exp\left\{6(\tau_n\log p + 2s_n) - 
C_6\min\left(\frac{1}{2},\frac{\|\bSigma_0\|_\infty^2}{\sqrt{2}}\right){\frac{\sqrt{M}n\eps_n^2}{\|\bSigma_0\|_\infty^2}}
\right\}
,\nonumber\\
\sup_{\bSigma\in H_1}\expect_\bSigma(1-\phi_n)&\leq 
4\exp\left\{4(\tau_n + 2s_n) - C_6\min\left(\frac{\|\bSigma_0\|_\infty^2}{8},\frac{1}{32}\right)\frac{Mn\eps_n^2}{\|\bSigma_0\|_\infty^2}\right\}.
\nonumber
\end{align}
\end{lemma}
\vspace*{1ex}\noindent
Before we proceed to the proof, observe that the bounded coherence assumption on $\bU_0$ (\emph{i.e.}, $\|\bU_0\|_{2\to\infty}\leq C_\mu\sqrt{r/s}$ for some $C_\mu\geq1$) implies the following bound for the infinity norm on $\bSigma_0$:
\begin{align}
\|\bSigma_0\|_\infty 
&\leq \|\bU_0\bLambda_0\bU_0\transpose\|_\infty + \sigma_0^2
\leq \lambda_{01}\|\bU_0\|_\infty\|\bU_0\transpose\|_\infty + \sigma_0^2\nonumber\\
&\leq \lambda_{01}\left(\sqrt{r}\|\bU_0\|_{2\to\infty}\right)\left(\sqrt{s}\|\bU_0\transpose\|_{2\to\infty}\right) + \sigma_0^2
\leq C_\mu r\|\bSigma_0\|_2\nonumber.
\end{align}
Hence, 
\[
\frac{n\eps_n^2}{\|\bSigma_0\|_\infty^2} = \frac{r^2s\log p}{C_\mu^2 r^2 \|\bSigma_0\|_2^2} = \frac{s\log p}{C_\mu^2 \|\bSigma_0\|_2^2}.
\]

\vspace*{1ex}\noindent
\textbf{Step 1} remains the same as that in the proof of Theorem \ref{thm:posterior_contraction}. In what follows we will make use of inequalities \eqref{eqn:prior_concentration_probabilistic_bound} and \eqref{eqn:evidence_lower_bound}.

\vspace*{1ex}\noindent
\textbf{Step 2: Construct subsets $(\calG_n)_{n=1}^\infty$. }
This step is also similar to that in the proof of Theorem \ref{thm:posterior_contraction}.
Take $\eps_n = r\sqrt{(s\log p)/n}$, $\tau_n = \beta s_n$, $t_n = (sr\log p)^2$, and $\delta_n = \eps_n/(pt_n)$, where $\beta>0$ is some constant to be specified later. Clearly, there exists some $\gamma>0$ such that
\[
\delta_n = \frac{\eps_n}{pt_n} = \frac{r\sqrt{s\log p}}{p\sqrt{n}(sr\log p)^2} = 
\frac{1}{\sqrt{np^2s^3r^2(\log p)^3}}\geq \frac{1}{p^\gamma}.
\]
Now let $\beta>4\mathrm{e}\gamma$ and $\calG_n = \calG(\delta_n,\tau_n,t_n)$ be defined in Lemma \ref{lemma:existence_test_infinity_norm}. 
Since
\begin{align}
\min\left\{\left(\frac{t_n}{\beta sr}\right)^2, \left(\frac{t_n}{r}\right)^2,\frac{t_n}{r}\right\}
&\geq 
\beta s\log p\nonumber
\end{align}
for sufficiently large $n$, 
and $t_n/(sr) = (sr)\log p\to\infty$, we then can invoke Lemmas \ref{lemma:prior_sparsity} and \ref{lemma:prior_deviation} to obtain
\begin{align}
\Pi(\calG_n^c)&\leq \Pi(|\mathrm{supp}_{\delta_n}(\bB)|>\beta s_n) + \Pi\left[\sum_{j=1}^p\|\bB_{j*}\|_1\mathbbm{1}\{j\in\mathrm{supp}_{\delta_n}(\bB)\cup\mathrm{supp}(\bU_0)\} > t_n\right]\nonumber\\
&\leq 2\exp(-\beta s\log p) + 5\exp\left\{
-\min\left(\frac{\beta\kappa}{2},\frac{\beta}{2\mathrm{e}} - 2\gamma\right)s\log p
\right\}\nonumber\\
\label{eqn:prior_concentration_upper_bound_infinity_norm}
&\leq 7\exp\left\{
-\min\left(\frac{\beta\kappa}{2},\frac{\beta}{2\mathrm{e}} - 2\gamma\right)s\log p
\right\}
\end{align}
for sufficiently large $n$ (and hence sufficiently small $(s\log p)/p$). 

\vspace*{1ex}\noindent
\textbf{Step 3: Decompose the integral. }Since by construction we have
\[
\max(p\delta_nt_n, \delta_nt_n + p\delta_n^2)\leq p\delta_nt_n + p\delta_n^2\leq 2p\delta_nt_n\leq 2\eps_n,
\]
then by Lemma \ref{lemma:existence_test_infinity_norm}, there exists some absolute constant $C_6>0$, such that for sufficiently large $n$, and for each 
\[
M\in\left[\max\left\{8, \frac{8(\log 2)^2}{C_6^2}\right\}, \frac{2\min(1, 2\|\bSigma_0\|_2)}{\eps_n}\right],
\]
there exists a test function $\phi_n$ such that 
\begin{align}
\label{eqn:testing_typeI_error_infinity_norm}
\expect_0(\phi_n)&\leq 12\exp\left[- \left\{\frac{C_6\sqrt{M}}{C_\mu^2\|\bSigma_0\|_2^2}\min\left(\frac{1}{2},\frac{\|\bSigma_0\|_2^2}{\sqrt{2}}\right) - 6(\beta + 2)\right\}s\log p\right],\\
\label{eqn:testing_typeII_error_infinity_norm}
\expect_\bSigma(1-\phi_n)&\leq 4\exp\left[- \left\{\frac{C_6M}{C_\mu^2\|\bSigma_0\|_2^2}\min\left(\frac{\|\bSigma_0\|_2^2}{8}, \frac{1}{32}\right) - 4(\beta + 2)\right\}s\log p\right]
\end{align}
for all $\bSigma\in \{\|\bSigma - \bSigma_0\|_\infty > M\eps_n\}\cap\calG_n$. Denote $\calV_n = \{\|\bSigma - \bSigma_0\|_\infty \leq M\eps_n\}$. Now we decompose the target integral $\expect_0\{\Pi(\calV_n^c\mid\bY_n)\}$ using \eqref{eqn:prior_concentration_probabilistic_bound} and \eqref{eqn:testing_typeI_error_infinity_norm} as follows:
\begin{align}
\expect_0\{\Pi(\calV_n^c\mid\bY_n)\}&\leq
\expect_0(\phi_n) + \expect_0\left\{(1-\phi_n)\Pi(\calV_n\mid\bY_n)\mathbbm{1}(\calA_n)\right\} + \prob_0(\calA_n^c)\nonumber\\
&\leq 12\exp\left[- \left\{\frac{C_6\sqrt{M}}{C_\mu^2\|\bSigma_0\|_2^2}\min\left(\frac{1}{2},\frac{\|\bSigma_0\|_2^2}{\sqrt{2}}\right) - 6(\beta + 2)\right\}s\log p\right]
\nonumber\\&\quad
 + 2\exp\left\{-\tilde C_3\min\left(1,\|\bSigma_0^{-1}\|_2^{-2}\right)s\log p\right\} 
+ \expect_0\left[(1 - \phi_n)\left\{\frac{N_n(\calV_n^c)}{D_n}\right\}\mathbbm{1}(\calA_n)\right]\nonumber.
\end{align}
Now we focus on the third term on the right-hand side of the preceding display. By \eqref{eqn:evidence_lower_bound}, we obtain
\begin{align}
&\expect_0\left[(1 - \phi_n)\left\{\frac{N_n(\calV_n^c)}{D_n}\right\}\mathbbm{1}(\calA_n)\right]
\leq \frac{\exp\left\{C_{0\lambda}\}s\log p\right\}}{C(\sigma_0^2)}\expect_0\left\{
(1 - \phi_n)\int_{\calV_n^c}\prod_{i=1}^n\frac{p(\by_i\mid\bSigma)}{p(\by_i\mid\bSigma_0)}\Pi(\mathrm{d}\bSigma)
\right\}\nonumber.
\end{align}
Observe that by Fubini's theorem,
\begin{align}
&\expect_0\left\{
(1 - \phi_n)\int_{\calV_n^c}\prod_{i=1}^n\frac{p(\by_i\mid\bSigma)}{p(\by_i\mid\bSigma_0)}\Pi(\mathrm{d}\bSigma)
\right\}\nonumber\\
&\quad\leq 
\expect_0\left\{
(1 - \phi_n)\int_{\calV_n^c\cap \calG_n}\prod_{i=1}^n\frac{p(\by_i\mid\bSigma)}{p(\by_i\mid\bSigma_0)}\Pi(\mathrm{d}\bSigma)
\right\}
+ \expect_0\left\{\int_{\calG_n^c}\prod_{i=1}^n\frac{p(\by_i\mid\bSigma)}{p(\by_i\mid\bSigma_0)}\Pi(\mathrm{d}\bSigma)
\right\}\nonumber\\
&\quad = \int_{\calV_n^c\cap \calG_n}\expect_0\left\{
(1 - \phi_n)\prod_{i=1}^n\frac{p(\by_i\mid\bSigma)}{p(\by_i\mid\bSigma_0)}
\right\}\Pi(\mathrm{d}\bSigma)
+ \int_{\calG_n^c}\left\{\expect_0\prod_{i=1}^n\frac{p(\by_i\mid\bSigma)}{p(\by_i\mid\bSigma_0)}\right\}\Pi(\mathrm{d}\bSigma)\nonumber\\
&\quad \leq \int_{\calV_n^c\cap\calG_n}\expect_{\bSigma}(1-\phi_n)\Pi(\mathrm{d}\bSigma) + 
\Pi(\calG_n^c)\nonumber\\
&\quad\leq 
4\exp\left[- \left\{\frac{C_6M}{C_\mu^2\|\bSigma_0\|_2^2}\min\left(\frac{\|\bSigma_0\|_2^2}{8}, \frac{1}{32}\right) - 4(\beta + 2)\right\}s\log p\right]
 + 
7\exp\left\{
-\min\left(\frac{\beta\kappa}{2},\frac{\beta}{2\mathrm{e}} - 2\gamma\right)s\log p
\right\}\nonumber,
\end{align}
where the testing type II error probability bound \eqref{eqn:testing_typeII_error_infinity_norm} and \eqref{eqn:prior_concentration_upper_bound_infinity_norm} are applied to the last inequality. Then by taking $M = M_\infty = \max\left(M_{\infty1},M_{\infty2}\right)$, where
\begin{align}
\beta &= \max\left[\frac{4}{\kappa}C_{0\lambda}, 2\mathrm{e}\left\{2\gamma + 2C_{0\lambda}\right\}\right],\nonumber\\
M_{\infty1}& = \max\left(\frac{32C_\mu^2\|\bSigma_0\|_2^2}{C_6}, \frac{8C_\mu^2}{C_6}\right)\left\{4(\beta + 2) + 2C_{0\lambda}\right\},\nonumber\\
M_{\infty2}& = \max\left(\frac{4C_\mu^4\|\bSigma_0\|_2^4}{C_6^2}, \frac{2C_\mu^4}{C_6^2}\right)
\left\{C_{0\lambda} + 6(\beta + 2)\right\}^2,\nonumber,
\end{align}
we obtain the following result:
\begin{align}
\expect_0\left[(1 - \phi_n)\left\{\frac{N_n(\calV_n^c)}{D_n}\right\}\mathbbm{1}(\calA_n)\right]
&\leq \frac{4}{C(\sigma_0^2)}\exp\left[- \left\{\frac{C_6M}{C_\mu^2\|\bSigma_0\|_2^2}\min\left(
\frac{\|\bSigma_0\|_2^2}{8},\frac{1}{32}
\right) - 4(\beta + 2) - C_{0\lambda}\right\}s\log p\right]\nonumber\\
&\quad + \frac{7}{C(\sigma_0^2)}
\exp\left[
-\left\{\min\left(\frac{\beta\kappa}{2},\frac{\beta}{2\mathrm{e}} - 2\gamma\right)
 - C_{0\lambda}
\right\}
s\log p
\right]\nonumber\\
&\leq \frac{11}{C(\sigma_0^2)}\exp\left(-C_{0\lambda}s\log p\right).\nonumber
\end{align}
Combining the above results, we finally obtain
\begin{align}
\expect_0\{\Pi(\calV_n^c\mid\bY_n)\}&\leq
\left\{3 + \frac{11}{C(\sigma_0^2)}\right\}\exp\left( - C_{0\lambda}s\log p\right)
+ 2\exp\left\{-\tilde C_3\min\left(1, \left\|\bSigma_0^{-1}\right\|_2^{-2}\right)s\log p\right\}\nonumber\\
&\leq \left\{5 + \frac{11}{C(\sigma_0^2)}\right\}\exp\left\{
-\min\left( C_{0\lambda}, \tilde C_3, \tilde C_3\|\bSigma_0^{-1}\|_2^{-2} \right)s\log p
\right\}\nonumber\\
& = R_0\exp(-C_0s\log p)\nonumber
\end{align}
by taking $C_0 = \min\left\{ C_{0\lambda}, \tilde C_3, \tilde C_3\|\bSigma_0^{-1}\|_2^{-2} \right\}$ and $R_0 = \left\{5 + {11}/{C(\sigma_0^2)}\right\}$. Therefore, there exists some constant $M_\infty$, such that for all sufficiently large $n$, we have
\[
\expect_0\left\{\Pi\left(\|\bSigma-\bSigma_0\|_\infty > M\eps_n\mathrel{\Big|}\bY_n\right)\right\}
\leq
\expect_0\left\{\Pi\left(\|\bSigma-\bSigma_0\|_\infty > M_0\eps_n\mathrel{\Big|}\bY_n\right)\right\}
\leq R_0\mathrm{e}^{-C_0s\log p}
\]
for some absolute constants $C_0$ and $R_0$ depending on $\bLambda_0$ and the hyperparameters only whenever $M\geq M_\infty$. Notice that $C_0$ and $R_0$ remain the same with those appearing in Theorem \ref{thm:posterior_contraction}.

\vspace*{1ex}
\noindent
\textbf{Step 4: Bounding the two-to-infinity norm loss using the Neumann trick. }
Let $\bB\bB\transpose = \bU_\bB\bLambda\bU\transpose_\bB$ be the compact spectral decomposition of $\bB\bB\transpose$. Denote $\mathbf{E} = \bB\bB\transpose - \bU_0\bLambda_0\bU_0\transpose$ to be the ``error'' matrix. 
Clearly, $(\bU_0\bLambda_0\bU_0\transpose + \mathbf{E})\bU_\bB = (\bU_\bB\bLambda\bU_\bB\transpose)\bU_\bB = \bU_\bB\bLambda$ by definition, yielding the matrix Sylvester equation 
\[
\bU_\bB\bLambda - \mathbf{E}\bU_\bB = (\bU_0\bLambda_0\bU_0\transpose)\bU_\bB.
\]
Now consider the events 
\begin{align}
\calU_n = \left\{\|\bSigma - \bSigma_0\|_2 \leq M_0\sqrt{\frac{s\log p}{n}}\right\},
\quad\calV_n = \left\{\|\bSigma - \bSigma_0\|_\infty \leq M_\infty r\sqrt{\frac{s\log p}{n}}\right\}.
\nonumber
\end{align}
Suppose $\bSigma\in\calU_n\cap\calV_n$. 
By the Weyl's inequality, for sufficiently large $n$, we have
\begin{align}
|\sigma^2 - \sigma_0^2|& = |\lambda_{r + 1}(\bSigma) - \lambda_{r + 1}(\bSigma_0)|\leq \max_{k\in[p]}|\lambda_k(\bSigma) - \lambda_k(\bSigma_0)|\leq \|\bSigma - \bSigma_0\|_2\leq M_0\sqrt{\frac{s\log p}{n}},\nonumber\\
\lambda_r(\bLambda)&\geq\lambda_{0r} - |\lambda_{0r} - \lambda_r(\bLambda)|\geq \lambda_{0r} - |(\lambda_{0r} + \sigma_0^2) - \{\lambda_r(\bLambda) + \sigma^2\}| - |\sigma_0^2 - \sigma^2|\nonumber\\
&\geq \lambda_{0r} - \max_{k\in [p]}|\lambda_k(\bSigma) - \lambda_k(\bSigma_0)| -  M_0\sqrt{\frac{s\log p}{n}}\nonumber\\
&\geq \lambda_{0r} - 2M_0\sqrt{\frac{s\log p}{n}}>\max\left\{\frac{\lambda_{0r}}{2}, 2M_0\sqrt{\frac{s\log p}{n}}\right\}\nonumber,\\
\|\mathbf{E}\|_2& \leq \|\bSigma - \bSigma_0\|_2 + \|(\sigma^2 - \sigma_0^2)\eye_p\|_2\leq 
2M_0\sqrt{\frac{s\log p}{n}}.\nonumber
\end{align}
Therefore, the spectra of $\bLambda$ and $\mathbf{E}$ are disjoint, and we can apply the Neumann's trick (see Theorem VII.2.2 in \citealp{bhatia2013matrix}) to expand $\bU_\bB$ in terms of a matrix series:
\begin{align}
\label{eqn:Neumann_series}
\bU_\bB = \sum_{m = 0}^\infty \mathbf{E}^m(\bU_0\bLambda_0\bU_0\transpose)\bU_\bB\bLambda^{-(m + 1)}
\end{align}
Now we proceed to bound $\|\bU_\bB - \bU_0\bW_\bU\|_{2\to\infty}$ using the techniques developed in \cite{cape2018signal}. Write
\begin{align}
\bU_\bB - \bU_0\bW_\bU& = (\bU_\bB\bLambda\bU_\bB\transpose - \bU_0\bLambda_0\bU_0\transpose)\bU_\bB\bLambda^{-1} + \bU_0\bLambda_0(\bU_0\transpose \bU_\bB\bLambda^{-1} - \bLambda_0^{-1}\bU_0\transpose\bU_\bB)
 + \bU_0(\bU_0\transpose\bU_\bB - \bW_\bU)\nonumber\\
& = \mathbf{E}\bU_\bB\bLambda^{-1} + \bU_0\bLambda_0(\bU_0\transpose \bU_\bB\bLambda^{-1} - \bLambda_0^{-1}\bU_0\transpose\bU_\bB) + \bU_0(\bU_0\transpose\bU_\bB - \bW_\bU)\nonumber.
\end{align}
By the CS decomposition and the sine-theta theorem, we see that the third term can be bounded:
\[
\|\bU_0(\bU_0\transpose\bU_\bB - \bW_\bU)\|_{2\to\infty}\leq \|\bU_0\|_{2\to\infty}\|\bU_\bB\bU_\bB\transpose - \bU_0\bU_0\transpose\|_2^2
\leq \frac{4M_0^2C_\mu}{\lambda_{0r}^2}\left(\frac{\sqrt{rs}\log p}{n}\right).
\]
Now we consider the second term. Denote $\mathbf{R} = \bU_0\transpose \bU_\bB\bLambda^{-1} - \bLambda_0^{-1}\bU_0\transpose\bU_\bB$. Then the $(i, j)$-th element of $\mathbf{R}$ can be represented as
\[
r_{k\ell} = (\bU_0)_{*k}\transpose(\bU_\bB)_{*\ell}\left\{\frac{1}{\lambda_{\ell}(\bLambda)} - \frac{1}{\lambda_{0k}}\right\} = \frac{1}{\lambda_{\ell}(\bLambda)\lambda_{0k}}\{\lambda_{0k} - \lambda_{\ell}(\bLambda)\}(\bU_0)_{*k}\transpose(\bU_\bB)_{*\ell}.
\]
Therefore, by defining $\mathbf{H}_1\in\mathbb{R}^{r\times r}$ by $(h_1)_{k\ell} = 1/\{\lambda_{\ell}(\bLambda)\lambda_{0k}\}$, we have
\begin{align}
\|\mathbf{R}\|_2 &= \|\mathbf{H}_1\circ (\bU_0\transpose\bU_\bB\bLambda - \bLambda_0 \bU_0\transpose\bU_\bB)\|_2
\leq r\|\mathbf{H}_1\|_{\max}\|\bU_0\transpose\mathbf{E}\bU_\bB\|_2
\leq r\|\mathbf{H}_1\|_{\max} 2M_0\sqrt{\frac{s\log p}{n}}\nonumber,
\end{align}
where $\circ$ represents the Hadamard matrix product (element-wise product), and $\|\cdot\|_{\max}$ is the maximum of the absolute values of the entries of a matrix. 
Furthermore, using the Weyl's inequality, we have
\[
\|\mathbf{H}_1\|_{\max}\leq \frac{1}{\lambda_r(\bLambda)\lambda_{0r}}\leq \frac{2}{\lambda_{0r}^2}
\]
for sufficiently large $n$, since $\|\bLambda^{-1}\|_2 = 1/\lambda_r(\bLambda)\leq 2/\lambda_{0r}$ for sufficiently large $n$. Hence, the second term can be bounded:
\[
\|\bU_0\bLambda_0(\bU_0\transpose \bU_\bB\bLambda^{-1} - \bLambda_0^{-1}\bU_0\transpose\bU_\bB)\|_{2\to\infty}
= \|\bU_0\|_{2\to\infty}\|\bLambda_0\|_2\|\mathbf{R}\|_2\leq \frac{4 M_0C_\mu\lambda_{01}}{\lambda_{0r}^2}\sqrt{\frac{r^3\log p}{n}}.
\]
Now we focus on the first term. By the Neumann matrix series \eqref{eqn:Neumann_series}, we have
\begin{align*}
\|\mathbf{E}\bU_\bB\bLambda^{-1}\|_{2\to\infty}
& = \left\|
\sum_{m = 1}^\infty\mathbf{E}^m(\bU_0\bLambda_0\bU_0\transpose)\bU_\bB\bLambda^{-(m + 1)}
\right\|_{2\to\infty}\nonumber\\
&\leq \|\mathbf{E}\bU_0\|_{2\to\infty}\|\bLambda_0\|_2\|\|\bLambda^{-1}\|_2^2 + 
\sum_{m = 2}^\infty \|\mathbf{E}\|_2^m\|\bLambda_0\|_2\|\bLambda^{-1}\|_2^{(m + 1)}\nonumber\\
&\leq \|\mathbf{E}\bU_0\|_{2\to\infty}\left\{\frac{\lambda_{01}}{\lambda_r(\bLambda)^2}\right\}
 + \left\{\frac{\lambda_{01}}{\lambda_r(\bLambda)}\right\}\frac{\|\mathbf{E}\|_2^2\|\bLambda^{-1}\|_2^2}{1 - \|\mathbf{E}\|_2\|\bLambda^{-1}\|_2}\nonumber\\
&\leq 4\|\mathbf{E}\|_\infty \|\bU_0\|_{2\to\infty}\frac{\lambda_{01}}{\lambda_{0r}^2} + \frac{8\lambda_{01}}{\lambda_{0r}^3}\|\mathbf{E}\|_2^2\nonumber\\
&\leq \frac{4M_\infty C_\mu\lambda_{01}}{\lambda_{0r}^2} \sqrt{\frac{r^3\log p}{n}} + \frac{8M_0^2\lambda_{01}}{\lambda_{0r}^3}\frac{s\log p}{n}\nonumber
\end{align*}
for sufficiently large $n$. In other words, there exists some constant $M_{2\to\infty}$ depending on $M_0$, $M_\infty$, $\bLambda_0$, and hyperparameters, such that
\[
\|\bU_\bB - \bU_0\bW_\bU\|_{2\to\infty} \leq M_{2\to\infty}\max\left(\sqrt{\frac{r^3\log p}{n}},\frac{s\log p}{n}\right)
\]
for sufficiently large $n$ whenever $\bSigma\in\calU_n\cap\calV_n$. Therefore, 
\begin{align}
&\expect_0\left[
\Pi\left\{\|\bU_\bB-\bU_0\bW_\bU\|_{2\to\infty} > M
\max\left(\sqrt{\frac{r^3\log p}{n}},\frac{s\log p}{n}\right)
\right\}
\right]\nonumber\\
&\quad\leq\expect_0\left[
\Pi\left\{\|\bU_\bB-\bU_0\bW_\bU\|_{2\to\infty} > M_{2\to\infty}
\max\left(\sqrt{\frac{r^3\log p}{n}},\frac{s\log p}{n}\right)
\right\}
\right]\nonumber\\
&\quad\leq \expect_0\left\{\Pi(\calU_n^c\mid\bY_n) + \Pi(\calV_n^c\mid\bY_n)\right\}\leq 2R_0\mathrm{e}^{-C_0s\log p}\nonumber,
\end{align}
for sufficiently large $n$ when $M\geq M_{2\to\infty}$, completing the proof.
\end{proof}

\begin{proof}[{\bf Proof of Theorem \ref{thm:risk_bound_point_estimate}}]
For any random matrix $\bX\in\mathbb{R}^{p\times p}$, we have
\[
\|\bE(\bX)\|_2^2 = \max_{\|\bu\|_2 = 1}\{\expect(\bX\bu)\}
\transpose\{\expect(\bX\bu)\} \leq \expect\|\bX\|_2^2
\]
by the Jensen's inequality. 
Now take $\bX = \bU_\bB\bU_\bB\transpose - \bU_0\bU_0\transpose$. Denote the event $\calU_n = \{\|\bU_\bB\bU_\bB\transpose - \bU_0\bU_0\transpose\|_2\leq M_0\eps_n\}$. Invoking the posterior contraction \eqref{eqn:posterior_contraction_Projection}, we have
\begin{align}
\expect_0\left(\left\|\widehat\bOmega - \bU_0\bU_0\transpose\right\|_2^2\right)
& = \expect_0\left\{\left\|\int\left(\bU_\bB\bU_\bB\transpose - \bU_0\bU_0\transpose\right)\Pi(\mathrm{d}\bB\mid\bY_n) \right\|_2^2\right\}\nonumber\\
& \leq 
\expect_0\left\{\int_{\calU_n}\left\|\left(\bU_\bB\bU_\bB\transpose - \bU_0\bU_0\transpose\right) \right\|_2^2\Pi(\mathrm{d}\bB\mid\bY_n)\right\}
\nonumber\\&\quad
 + \expect_0\left\{\int_{\calU_n^c}\left\|\left(\bU_\bB\bU_\bB\transpose - \bU_0\bU_0\transpose\right) \right\|_2^2\Pi(\mathrm{d}\bB\mid\bY_n)\right\}\nonumber\\
&\leq M_0^2\eps_n^2 + \left(\sup_{\bU\in\mathbb{O}(p, r)}\|\bU\bU\transpose - \bU_0\bU_0\transpose\|_2^2\right)\expect_0\left\{\Pi(\calU_n^c\mid\bY_n)\right\}\nonumber\\
&\leq \frac{4M_0^2}{\lambda_{0r}^2}\eps_n^2 + 4R_0\exp(-C_0s\log p).\nonumber
\end{align}
Since for sufficiently large $n$, we have
\[
\eps_n^2 = \frac{s\log p}{n} = \exp\left(\log s + \log\log p - \log n\right)
\geq\exp(-C_0s\log p),
\]
we obtain
\[
\expect_0\left(\left\|\widehat\bOmega - \bU_0\bU_0\transpose\right\|_2\right)
\leq\left\{\expect_0\left(\left\|\widehat\bOmega - \bU_0\bU_0\transpose\right\|_2^2\right)\right\}^{1/2}
\leq \eps_n\left(\frac{2M_0}{\lambda_{0r}} + 2\sqrt{R_0}\right).
\]
Since the columns of $\widehat\bU$ are the leading $r$-eigenvectors of $\widehat\bOmega$ corresponding to $\lambda_1(\widehat\bOmega),\ldots,\lambda_r(\widehat\bOmega)$, \emph{i.e.}, $\widehat\bOmega\widehat\bU_{*k} = \lambda_k(\widehat\bOmega)\widehat\bU_{*k}$, then applying the sine-theta theorem (Theorem \ref{thm:Davis_Kahan_sin_theta}) yields
\[
\expect_0 \left(\|\widehat\bU\widehat\bU\transpose - \bU_0\bU_0\transpose\|_2\right)\leq \left(\frac{4M_0}{\lambda_{0r}} + 4\sqrt{R_0}\right)\eps_n.
\]

\end{proof}


\section{Proofs of Results in Section \ref{sub:properties_of_the_matrix_spike_and_slab_lasso_priors}} 
\label{sec:proofs_of_results_in_section_3_1}
\begin{customlemma}{3.1}
Suppose $\bB\sim\mathrm{MSSL}_{p\times r}(\lambda, 1/p^2, p^{1+\kappa})$ for some fixed positive $\lambda$ and $\kappa$, and $\bB_0\in\mathbb{R}^{p\times r}$ is jointly $s$-sparse, where $1\leq r\leq s\leq p/2$. Then for small values of $\eta\in(0,1)$ with $\eta\geq1/p^\gamma$ for some $\gamma>0$, it holds that
\begin{align}
\Pi\left(\|\bB-\bB_0\|_{\mathrm{F}}<\eta\right)\geq\exp\left[-C_1\max\left\{\lambda^2s\|\bB_0\|_{2\to\infty}^2, sr\left|\log\frac{\lambda\eta}{\sqrt{sr}}\right|,s\log p\right\}\right]\nonumber
\end{align}
for some absolute constant $C_1>0$.
\end{customlemma}
\begin{proof}[{\bf Proof of Lemma \ref{lemma:prior_concentration}}]
Recall that $\pi(b_{jk}\mid\xi_j = 1)=(\lambda/2)\mathrm{e}^{-\lambda|b_{jk}|}$ follows the Laplace distribution with scale parameter $1/\lambda$, and that the Laplace distribution can be alternatively represented as a normal-variance mixture distribution as follows:
\[
(b_{jk}\mid\xi_j = 1, \phi_{jk})\sim\mathrm{N}\left(0,\frac{\phi_{jk}}{\lambda^2}\right),\quad\text{and}\quad \phi_{jk}\sim\mathrm{Exp}(1/2).
\]
On the other hand, by the prior construction $(|b_{jk}|\mid\xi_j = 0,\lambda_0)\sim\mathrm{Gamma}(1/r, \lambda_0 + \lambda)$, it follows that $(\|\bB_{j*}\|_1\mid \xi_j = 0,\lambda_0)\sim\mathrm{Exp}(\lambda_0 + \lambda)$. 
Denote $S_0 = \mathrm{supp}(\bB_0)$. Now we construct the following event
\[
\calB = \bigcap_{j\in S_0}\left\{\xi_j=1, 1\leq\phi_{jk}\leq 2, k \in[r]\right\}\cap \bigcap_{j\in S_0^c}\{\xi_j = 0\}\cap\left\{\lambda_0 + \lambda\geq\frac{\sqrt{2p}}{\eta}\left(\log\frac{p}{s}\right)\right\}
\]
and denote $\bphi = [\phi_{jk}:j\in S_0, k\in[r]]_{s\times r}$. 

\vspace*{1ex}
\noindent
\textbf{Step 1: Conditioning on the event $\calB$. }For any $(\bphi,\bxi,\lambda_0)\in\calB$, we use a union bound to derive
\begin{align}
\Pi\left(\|\bB-\bB_0\|_{\mathrm{F}}<\eta\mathrel{\big|}\bphi,\bxi,\lambda_0\right)
&\geq \Pi\left(\sum_{j\in S_0}\|\bB_{j*} - \bB_{0j*}\|_2^2<\frac{\eta^2}{2}\mathrel{\bigg|}\bphi,\bxi,\lambda_0\right)\prod_{j\in S_0^c}\Pi\left(\|\bB_{j*}\|_1\leq\frac{\eta}{\sqrt{2p}}\mathrel{\bigg|}\bphi,\bxi,\lambda_0\right)\nonumber\\
& \geq \Pi\left(\sum_{j\in S_0}\|\bB_{j*} - \bB_{0j*}\|_2^2<\frac{\eta^2}{2}\mathrel{\bigg|}\bphi,\bxi,\lambda_0\right)\prod_{j\in S_0^c}\left[1-\exp\left\{-\frac{(\lambda_0 + \lambda)\eta}{\sqrt{2p}}\right\}\right]\nonumber\\
& \geq \Pi\left(\sum_{j\in S_0}\|\bB_{j*} - \bB_{0j*}\|_2^2<\frac{\eta^2}{2}\mathrel{\bigg|}\bphi,\bxi,\lambda_0\right)\left\{\left(1-\frac{s}{p}\right)^{p/s}\right\}^s\nonumber\\
& \geq \Pi\left(\sum_{j\in S_0}\|\bB_{j*} - \bB_{0j*}\|_2^2<\frac{\eta^2}{2}\mathrel{\bigg|}\bphi,\bxi,\lambda_0\right)\exp\{-\log(2\mathrm{e})s\}\nonumber,
\end{align}
where the last inequality is due to the fact that $(1-x)^{1/x}\geq\exp\{-\log(2\mathrm{e})\}$ when $x\in[0,1/2]$. It then suffices to provide a lower bound for the first factor. We take advantage of the fact that $(b_{jk}\mid\xi_j = 1,\phi_{jk})\sim\mathrm{N}(0,\phi_{jk}/\lambda^2)$ and apply Anderson's lemma (see, for example, Lemma 1.4 in the supporting document of \citealp{pati2014posterior}) together with the union bound to derive
\begin{align}
&\Pi\left(\sum_{j\in S_0}\|\bB_{j*} - \bB_{0j*}\|_2^2<\frac{\eta^2}{2}\mathrel{\bigg|}\bphi,\bxi,\lambda_0\right)
\nonumber\\
&\quad \geq\exp\left(-\frac{1}{2}\sum_{j\in S_0}\sum_{k=1}^r\frac{\lambda^2b_{0jk}^2}{\phi_{jk}}\right)\Pi\left(\sum_{j\in S_0}\sum_{k=1}^rb_{jk}^2<\frac{\eta^2}{2}\mathrel{\bigg|}\bphi,\bxi,\lambda_0\right)\nonumber\\
&\quad  \geq\exp\left(-\frac{1}{2}\sum_{j\in S_0}\sum_{k=1}^r{\lambda^2b_{0jk}^2}\right)\prod_{j\in S_0}\prod_{k=1}^r\Pi\left(\frac{\lambda^2 b_{jk}^2}{\phi_{jk}}<\frac{\lambda^2\eta^2}{2\phi_{jk}rs}\mathrel{\bigg|}\bphi,\bxi,\lambda_0\right)\nonumber\\
&\quad  \geq\exp\left(-\frac{\lambda^2}{2}\sum_{j\in S_0}\|\bB_0\|_{2\to\infty}^2\right)\prod_{j\in S_0}\prod_{k=1}^r\left\{2\Phi\left(\frac{\lambda\eta}{2\sqrt{rs}}\right) - 1\right\}\nonumber\\
&\quad  \geq\exp\left(-\frac{\lambda^2}{2}s\|\bB_0\|_{2\to\infty}^2 - sr - sr\left|\log\frac{\lambda\eta}{2\sqrt{rs}}\right|\right)\nonumber,
\end{align}
where the fact that $\log\{2\Phi(x) - 1\}\geq -1-\log(x)$ for small $x>0$ is applied in the last inequality.

\noindent
\textbf{Step 2: Control the prior probability of the event $\calB$. }
Recall that
\[
\calB = \bigcap_{j\in S_0}\left\{\xi_j=1, 1\leq\phi_{jk}\leq 2, k \in[r]\right\}\cap \bigcap_{j\in S_0^c}\{\xi_j = 0\}\cap\left\{\lambda_0 + \lambda\geq\frac{\sqrt{2p}}{\eta}\left(\log\frac{p}{s}\right)\right\}.
\]
Then conditioning on $\theta$, we obtain by construction
\begin{align}
\Pi(\calB) 
&= \left\{\prod_{j\in S_0}\prod_{k = 1}^r\Pi(1\leq\phi_{jk}\leq2)\right\}\left\{\int_0^1\theta^s(1-\theta)^{p-s}\Pi(\mathrm{d}\theta)\right\}\Pi\left\{\lambda_0 + \lambda\geq\frac{\sqrt{2p}}{\eta}\left(\log\frac{p}{s}\right)\right\}\nonumber\\
&\geq \exp(-3sr)\left\{\int_0^1\theta^s(1-\theta)^{p-s}\Pi(\mathrm{d}\theta)\right\}\Pi\left\{\lambda_0\geq\frac{\sqrt{2p}}{\eta}\left(\log\frac{p}{s}\right)\right\}\nonumber.
\end{align}
We first focus on the third factor. By assumption $\eta>1/p^\gamma$ for some $\gamma>0$, implying that
\[
\Pi\left\{\lambda_0 > \frac{\sqrt{2p}}{\eta}\left(\log\frac{p}{s}\right)\right\}\geq \Pi(\lambda_0> p^{\gamma}) 
= 1-\frac{1}{\Gamma(1/p^2)}\int_{1/p^\gamma}^\infty x^{1/p^2 - 1}\mathrm{e}^{-x}\mathrm{d}x.
\]
Using an inequality for the incomplete Gamma function \citep{alzer1997some}:
\[
\int_{4\delta}^\infty x^{-1}\mathrm{e}^{-x/2}\mathrm{d}x = \int_{2\delta}^\infty x^{-1}\mathrm{e}^{-x}\mathrm{d}x\leq\log\frac{1}{\delta}
\]
for small values of $\delta>0$, and the fact that $\Gamma(x)\geq 1$ when $0<x\leq 1$, we have:
\[
1 - \frac{1}{\Gamma(1/p^2)}\int_{1/p^\gamma}^\infty x^{1/p^2-1}\mathrm{e}^{-x}\mathrm{d}x\geq 1 - \int_{4/(4p^\gamma)}^\infty x^{-1}\mathrm{e}^{-x/2}\mathrm{d}x\geq 1 - \log\frac{1}{4p^\gamma}\geq\mathrm{e}^{-1}
\]
for sufficiently large $p$ (sufficiently small $\eta$). Next we consider the second factor. Write
\begin{align}
\int_0^1\theta^s(1-\theta)^{p-s}\Pi(\mathrm{d}\theta)
&\geq\int_{s/p^{1+\kappa}}^{2s/p^{1+\kappa}}\exp\left\{-s\log\left(\frac{1-\theta}{\theta}\right) - p\log\left(\frac{1}{1-\theta}\right)\right\}\Pi(\mathrm{d}\theta)\nonumber\\
&\geq\int_{s/p^{1+\kappa}}^{2s/p^{1+\kappa}}\exp\left\{-s\log\left(\frac{p^{1+\kappa}}{s}\right) - p\log\left(1+\frac{s}{p^{1+\kappa}-s}\right)\right\}\Pi(\mathrm{d}\theta)\nonumber\\
&\geq\exp\left\{-(\kappa + 1)s\log p - 2s\right\}\Pi\left(\frac{s}{p^{1+\kappa}}\leq \theta\leq \frac{2s}{p^{1+\kappa}}\right)\nonumber
\end{align}
for sufficiently large $p$. Observe that
\begin{align}
\Pi\left(\frac{s}{p^{1+\kappa}}\leq \theta\leq \frac{2s}{p^{1+\kappa}}\right)
& = \Pi\left(\frac{p^{1+\kappa} - 2s}{p^{1+\kappa}}\leq 1-\theta\leq \frac{p^{1+\kappa} - s}{p^{1+\kappa}}\right)\nonumber\\
& \geq \frac{1}{4p^{1+\kappa}}\left(1 - \frac{2s}{p^{1+\kappa}}\right)^{4p^{1+\kappa}-1}\left(\frac{s}{p^{1+\kappa}}\right)\nonumber\\
& \geq \frac{1}{4p^{1+\kappa}}\left\{\left(1 - \frac{2s}{p^{1+\kappa}}\right)^{p^{1+\kappa}/(2s)}\right\}^{8s}\left(\frac{s}{p^{1+\kappa}}\right)\nonumber\\
& \geq \exp\left\{-(2\kappa + 19)s\log p\right\}\nonumber,
\end{align}
we conclude that
$\Pi(\calB)\geq \exp\left\{
-3sr - 1 - (3\kappa + 22)s\log p
\right\}$.

\vspace*{1ex}
\noindent\textbf{Lower bound prior concentration by restricting over $\calB$: }
We complete the proof by restricting over the event $\calB$ as follows: 
\begin{align}
\Pi\left(\|\bB-\bB_0\|_{\mathrm{F}}<\eta\right)
&\geq \expect_\Pi\left\{
\Pi\left(\|\bB-\bB_0\|_{\mathrm{F}}<\eta\mathrel{\Big|}\bphi,\bxi,\lambda_0\right)\mathbbm{1}(\calB)
\right\}\nonumber\\
& \geq \left\{\inf_{(\bphi,\bxi,\lambda_0)\in \calB}
\Pi\left(\|\bB-\bB_0\|_{\mathrm{F}}<\eta\mathrel{\Big|}\bphi,\bxi,\lambda_0\right)
\right\}\Pi(\calB)\nonumber\\
& \geq\exp\left[-C_1\max\left\{\lambda^2s\|\bB_0\|_{2\to\infty}^2, sr\left|\log\frac{\lambda\eta}{\sqrt{rs}}\right|, s\log p\right\}\right],\nonumber
\end{align}
where $C_1>0$ is some absolute constant. 
\end{proof}

\begin{customlemma}{3.2}
Suppose $\bB\sim\mathrm{MSSL}_{p\times r}(\lambda, 1/p^2, p^{1+\kappa})$ for some fixed positive $\lambda$ and $\kappa\leq 1$, $1\leq r\leq p$. Let $\delta\in(0,1)$ be a small number with $\delta>1/p^\gamma$ for some $\gamma>0$, and let $s$ be an integer such that $(s\log p)/p$ is sufficiently small. Then for any $\beta > 4\mathrm{e}\gamma$, it holds that
\[
\Pi\left(|\mathrm{supp}_\delta(\bB)|>\beta s\right)\leq 2\exp\left\{-\min\left(\frac{\beta\kappa}{2},\frac{\beta}{\mathrm{2e}} - 2\gamma\right)s\log p\right\}.
\]
\end{customlemma}

\begin{proof}[{\bf Proof of Lemma \ref{lemma:prior_sparsity}}]
Recall that by construction, $(\|\bB_{j*}\|_1\mid\xi_j = 1)\sim\mathrm{Gamma}(r, \lambda)$ and $(\|\bB_{j*}\|_1\mid\xi_j = 0, \lambda_0)\sim\mathrm{Exp}(\lambda_0 + \lambda)$, and $(\xi_j\mid\theta)\sim\mathrm{Bernoulli}(\theta)$ independently for each $j\in[p]$. Then with $\bxi$ integrated out, we have, independently for each $j\in[p]$,
\[
\pi(\|\bB_{j*}\|_1\mid\theta,\lambda_0) = (1-\theta)(\lambda_0 + \lambda)\mathrm{e}^{-(\lambda_0 + \lambda)\|\bB_{j*}\|_1} + \theta\frac{\lambda^r}{\Gamma(r)}\|\bB_{j*}\|_1^{r-1}\mathrm{e}^{-\lambda\|\bB_{j*}\|_1}.
\]
Therefore, with $\lambda_0$ integrated out, for any $\delta>1/p^\gamma$, we obtain
\begin{align}
\Pi(\|\bB_{j*}\|_1>\delta\mid\theta)&\leq(1-\theta)\frac{1}{\Gamma(1/p^2)}\int_0^\infty \lambda_0^{-1/p^2 - 1}\mathrm{e}^{-1/\lambda_0}\mathrm{e}^{-(\lambda_0 + \lambda)\delta}\mathrm{d}\lambda_0 + \theta\nonumber\\
&\leq \frac{2}{\mathrm{e}^{p}}\int_0^\infty u^{1/p^2 - 1}\exp\left(-\frac{\delta}{u}-u\right)\mathrm{d}u + \theta\nonumber,
\end{align}
where the last inequality is due to the change of variable $u = 1/\lambda_0$ and the fact that $\Gamma(1/p^2)\geq \mathrm{e}^{p}/2$ for sufficiently large $p$. Now we break down the integral in the preceding display as follows:
\begin{align}
\int_0^\infty u^{1/p^2 - 1}\exp\left(-\frac{\delta}{u}-u\right)\mathrm{d}u \leq
\int_0^{4\delta} u^{1/p^2 - 1}\exp\left(-\frac{\delta}{u}\right)\mathrm{d}u + 
\int_{4\delta}^\infty u^{1/p^2 - 1}\exp\left(-{u}\right)\mathrm{d}u.\nonumber
\end{align}
For the first term, we observe that the function $u\mapsto (1/p^2 - 1)\log u - \delta/u$ achieves the maximum at $u = \delta /(1 - p^{-2})$, and therefore, for sufficiently large $p$ (small $\delta$)
\[
\int_0^{4\delta} u^{1/p^2 - 1}\exp\left(-\frac{\delta}{u}\right)\mathrm{d}u
\leq 4\delta\exp\left\{\left(1-\frac{1}{\mathrm{e}^p}\right)\left(\log\frac{1 - p^{-2}}{\delta}\right)\right\}\leq 4\delta^{1/p^2}\leq \log\frac{1}{\delta}.
\]
For the second term, we apply the technique developed by \cite{doi:10.1080/01621459.2014.960967} to derive
\[
\int_{4\delta}^\infty u^{1/p^2 - 1}\exp\left(-{u}\right)\mathrm{d}u
\leq \int_{4\delta}^\infty u^{-1}\mathrm{e}^{-u/2}\mathrm{d}u\leq\log\frac{1}{\delta},
\]
where the inequality for incomplete Gamma function due to \cite{alzer1997some} is applied. Therefore, for any $\theta$ in the event $\{\theta<A_1s\log p/p^{1+\kappa}\}$ for some constant $A_1$ to be determined later, we obtain
\[
\Pi(\|\bB_{j*}\|_1>\delta\mid\theta)\leq \frac{4}{\mathrm{e}^{p}}\left(\log\frac{1}{\delta}\right) + \theta
\leq\frac{4\gamma\log p + A_1s\log p}{p^{1 + \kappa}}\leq \frac{s\log p}{p}\left(\frac{A_1 + 4\gamma}{p^\kappa}\right).
\]
A version of the Chernoff's inequality for binomial distributions states that \citep{hagerup1990guided}
\[
\prob(X>ap)\leq\left\{\left(\frac{q}{a}\right)^a\exp(a)\right\}^p\quad\text{if }X\sim\mathrm{Binomial}(p, q)\text{ and }q\leq a<1.
\]
Then over the event $\{\theta<A_1s\log p/p^{1+\kappa}\}$, we have
\begin{align}
\Pi(|\mathrm{supp}_\delta(\bB)|>\beta s\mid\theta)
&\leq\exp\left[-\beta s\left\{\log\frac{\beta}{\mathrm{e}(A_1+4\gamma)\log p} + \kappa\log p\right\}\right]
=\exp\left(-\frac{1}{2}\beta\kappa s\log p\right)\nonumber
\end{align}
by taking $A_1 = \beta/\mathrm{e} - 4\gamma$, $q = \Pi(\|\bB_{j*}\|_1>\delta\mid\theta)\leq (A_1 + 4\gamma)s\log p/p^{1+\kappa}$, and $a = \beta s/p$. 
Observe that for sufficiently small $x$, 
$\left(1 - x\right)^{1/x}\leq\mathrm{e}^{-1/2}$. 
Then we integrate with respect to $\Pi(\mathrm{d}\theta)$ and proceed to compute
\begin{align}
\Pi(|\mathrm{supp}_\delta(\bB)|>\beta s) &= \int_0^1 \Pi(|\mathrm{supp}_\delta(\bB)|>\beta s\mid\theta)\Pi(\mathrm{d}\theta)\nonumber\\
&\leq \int_0^{A_1s\log p/p^{1+\kappa}}\Pi(|\mathrm{supp}_\delta(\bB)|>\beta s\mid\theta)\Pi(\mathrm{d}\theta) + \Pi\left(\theta > \frac{A_1s\log p}{p^{1+\kappa}}\right)\nonumber\\
&\leq \exp\left(-\frac{1}{2}\beta\kappa s\log p \right) + \left\{\left(1 - \frac{A_1s\log p}{p^{1+\kappa}}\right)^{p^{1+\kappa}/(A_1s\log p)}\right\}^{A_1s\log p}\nonumber\\
&\leq \exp\left(-\frac{1}{2}\beta \kappa s\log p \right) + \exp\left(- \frac{A_1}{2}s\log p\right)\nonumber\\
&\leq 2\exp\left\{-\min\left(\frac{\beta\kappa}{2},\frac{\beta}{\mathrm{2e}} - 2\gamma\right)s\log p\right\}\nonumber,
\end{align}
and the proof is thus completed. 
\end{proof}

\begin{customlemma}{3.3}
Suppose $\bB\sim\mathrm{MSSL}_{p\times r}(\lambda, 1/p^2, p^{1+\kappa})$ for some fixed positive $\lambda$ and $\kappa<1$, and $\bB_0\in\mathbb{R}^{p\times r}$ is jointly $s$-sparse, where $r\log n\lesssim \log p$, and $(s\log p)/p$ is sufficiently small. Let $(\delta_n)_{n=1}^\infty$ and $(t_n)_{n=1}^\infty$ be positive sequences such that $1/p^\gamma\leq\delta_n\leq 1$ and $t_n/(sr) \to \infty$. Then for sufficiently large $n$ and for all $\beta >4\mathrm{e}\gamma$, it holds that
\begin{align}
&\Pi\left[\sum_{j=1}^p\|\bB_{j*}\|_1\mathbbm{1}\{j\in\mathrm{supp}_{\delta_n}(\bB)\cup\mathrm{supp}(\bB_0)\}\geq t_n\right]\nonumber\\
&\quad\leq 2\exp\left[-C_2\min\left\{\left(\frac{t_n}{\beta sr}\right)^2,\left(\frac{t_n}{r}\right)^2,\frac{t_n}{r}\right\}\right] + 
3\exp\left\{ - \min\left(\frac{\beta\kappa}{2},\frac{\beta}{\mathrm{2e}} - 2\gamma\right)s\log p\right\}\nonumber
\end{align}
for some absolute constant $C_2>0$. 
\end{customlemma}

\begin{proof}[{\bf Proof of Lemma \ref{lemma:prior_deviation}}]
To proof Lemma \ref{lemma:prior_deviation}, we need the following technical results regarding moments of Gamma mixture distributions, the proof of which is deferred to Section \ref{sec:remaining_technical_proofs}.
\begin{lemma}
\label{lemma:Gamma_subexponential_norm}
Suppose that $w$ follows a mixture of an exponential distribution $\mathrm{Exp}(\lambda_0)$ and a Gamma distribution $\mathrm{Gamma}(r,\lambda)$, with mixing weights $1-\theta$ and $\theta$, respectively. Let $\xi = \mathbbm{1}(w>\delta)$, where $\delta$ is some sufficiently small constant such that $\Gamma(r)\leq 2\Gamma(r,\lambda\delta)$, and $\Gamma(r,\delta) = \int_\delta^\infty w^{r-1}\mathrm{e}^{-w}\mathrm{d}w$ is the (upper) incomplete Gamma function. Then the moments of $w$ satisfy
\[
\sup_{m\geq1}\frac{1}{m}\left\{E(w^m\mid\xi = 1)\right\}^{1/m}\leq 2\delta+\frac{2}{\lambda_0}+\frac{2(r+1)}{\lambda}\quad\text{and}\quad
\sup_{m\geq1}\frac{1}{m}\{E(w^m)\}^{1/m}\leq\frac{1}{\lambda_0}+\frac{r+1}{\lambda}.
\]
Furthermore, if $\theta\leq \mathrm{e}^{-r}$, then the moments of $w$ satisfies
\[
\sup_{m\geq1}\frac{1}{m}\{E(w^m)\}^{1/m}\leq\frac{1}{\lambda_0}+\frac{1}{\lambda}.
\]
\end{lemma}
\noindent
Denote $S_0 = \mathrm{supp}(\bB_0)$. We first use the union bound to derive
\begin{align}
&\Pi\left[\sum_{j=1}^p\|\bB_{j*}\|_1\mathbbm{1}\{j\in\mathrm{supp}_{\delta_n}(\bB)\cup\mathrm{supp}(\bB_0)\}\geq t_n\right]\nonumber\\
&\quad\leq \Pi\left\{\sum_{j=1}^p\|\bB_{j*}\|_1\mathbbm{1}(\|\bB_{j*}\|_1>\delta_n)\geq t_n/2\right\} + \Pi\left(\sum_{j\in S_0}\|\bB_{j*}\|_1\geq t_n/2\right),\nonumber
\end{align}
and then analyze the above two terms separately. 

\vspace*{1ex}
\noindent\textbf{Upper bounding the second term. }Recall that
\[
\pi(\|\bB_{j*}\|_1\mid\lambda_0,\theta) = (1-\theta)(\lambda_0 + \lambda)\mathrm{e}^{-(\lambda_0+ \lambda)\|\bB_{j*}\|_1} + \theta\frac{\lambda^r}{\Gamma(r)}\|\bB_{j*}\|^{r-1}\mathrm{e}^{-\lambda\|\bB_{j*}\|_1}.
\]
Denote $\beta' = \beta/\mathrm{e} - 4\gamma> 0$. Over the event $\{\theta\leq (\beta's\log p)/p^{1+\kappa}\}$, it holds that
\[
\expect_\Pi(\|\bB_{j*}\|_1\mid\theta, \lambda_0)
\leq \frac{1}{\lambda_0} + \frac{\beta's\log p}{\lambda p^{1+\kappa}}\leq \frac{2}{\lambda}.
\]
Since $(\beta's\log p)/p^{1+\kappa}\leq 1/\sqrt{p} = \mathrm{e}^{-(\log p)/2}\leq \mathrm{e}^{-r}$ for sufficiently large $n$, we invoke Lemma \ref{lemma:Gamma_subexponential_norm} to derive
\[
\sup_{m\geq1}\left\{\expect_\Pi\left(\|\bB_{j*}\|_1^m\mid \theta,\lambda_0 \right)\right\}^{1/m}\leq \frac{2}{\lambda}
\]
over the event $\{\theta\leq (\beta's\log p)/p^{1+\kappa}\}$, and proceed to apply the large deviation inequality for sub-exponential random variables (Proposition 5.16 in \citealp{vershynin2010introduction}) to obtain
\begin{align}
\Pi\left(\sum_{j\in S_0}\|\bB_{j*}\|_1\geq t_n/2\mathrel{\bigg|}\lambda_0,\theta\right)
&\leq \Pi\left[\sum_{j\in S_0}\left\{\|\bB_{j*}\|_1 - \expect_\Pi(\|\bB_{j*}\|_1\mid\lambda_0,\theta)\right\}\geq t_n/2 - \frac{2s}{\lambda}\mathrel{\bigg|}\lambda_0,\theta\right]\nonumber\\
&\leq \Pi\left[\sum_{j\in S_0}\left\{\|\bB_{j*}\|_1 - \expect_\Pi(\|\bB_{j*}\|_1\mid\lambda_0,\theta)\right\}\geq t_n/4\mathrel{\bigg|}\lambda_0,\theta\right]\nonumber\\
&\leq \exp\left\{-C\min\left(\frac{t_n^2}{s_n^2}, t_n\right)\right\}\nonumber
\end{align}
for sufficiently large $n$, where $C$ is some absolute constant. 
Observe that
\begin{align}
\Pi\left(\theta>\frac{\beta's\log p}{p^{1+\kappa}}\right)
&=\left\{ \left(1 - \frac{\beta's\log p}{p^{1+\kappa}}\right)^{p^{1+\kappa}/(\beta's\log p)}\right\}^{\beta's\log p}
\leq\exp\left\{
-\left(\frac{\beta}{2\mathrm{e}}-4\gamma\right)s\log p
\right\}\nonumber
\end{align}
since $(1 - x)^{1/x}\leq\mathrm{e}^{-1/2}$ for $x\leq 1$, and so we obtain
\begin{align}
\Pi\left(\sum_{j\in S_0}\|\bB_{j*}\|_1\geq t_n/2\right)&\leq 
\expect_\Pi\left[
\Pi\left(\sum_{j\in S_0}\|\bB_{j*}\|_1\geq t_n/2\mathrel{\bigg|}\lambda_0,\theta\right)
\mathbbm{1}\left(\theta\leq\frac{\beta's\log p}{p^{1+\kappa}}\right)
\right]
+ \Pi\left(\theta\leq\frac{\beta's\log p}{p^{1+\kappa}}\right)\nonumber\\
&\leq \exp\left\{-C\min\left(\frac{t_n^2}{s_n^2},t_n\right)\right\} + \exp\left\{
-\left(\frac{\beta}{2\mathrm{e}}-4\gamma\right)s\log p
\right\}\nonumber
\end{align}
for sufficiently large $n$. 

\vspace*{1ex}
\noindent
\textbf{Upper bounding the first term. }Denote $\zeta_j = \mathbbm{1}(\|\bB_{j*}\|_1>\delta_n)$ and $\bzeta = [\zeta_1,\ldots,\zeta_p]\transpose$. By Lemma \ref{lemma:Gamma_subexponential_norm} we obtain the following bound for the conditional expected value and moments of $\|\bB_{j*}\|_1$ given $\zeta_j = 1$ and $\theta$ for sufficiently large $n$:
\[
\expect_\Pi(\|\bB_{j*}\|_1\mid\zeta_j = 1,\theta)\leq\sup_{m\geq1}\left\{\expect_\Pi(\|\bB_{j*}\|_1^m\mid\zeta_j = 1,\theta)\right\}^{1/m}\leq 2\delta_n + \frac{2}{\lambda_0} + \frac{2(r+1)}{\lambda}\leq\frac{8r}{\lambda}.
\]
Since $|\mathrm{supp}_{\delta_n}(\bB)|=\sum_{j=1}^p\zeta_j$, then over the event $\{\bzeta:|\mathrm{supp}_{\delta_n}(\bB)|\leq \beta s\}$, we invoke the large deviation inequality for sub-exponential random variables again to derive
\begin{align}
\Pi\left(\sum_{j=1}^p\|\bB_{j*}\|_1\zeta_j>t_n/2\mathrel{\Big|}\bzeta,\theta\right)
&\leq \Pi\left[\sum_{j\in\mathrm{supp}_{\delta_n}(\bB)}\left\{\|\bB_{j*}\|_1 - \expect_\Pi(\|\bB_{j*}\|_1\mid\zeta_j = 1)\right\}>\frac{t_n}{2} - \frac{8sr}{\lambda}\mathrel{\Big|}\bzeta,\theta\right]\nonumber\\
&\leq \Pi\left[\sum_{j\in\mathrm{supp}_{\delta_n}(\bB)}\left\{\|\bB_{j*}\|_1 - \expect_\Pi(\|\bB_{j*}\|_1\mid\zeta_j = 1)\right\}>\frac{t_n}{4}\mathrel{\Big|}\bzeta,\theta\right]\nonumber\\
&\leq\exp\left[-C\min\left\{\left(\frac{t_n}{\beta sr}\right)^2,\frac{t_n}{r}\right\}\right]\nonumber
\end{align}
for sufficiently large $n$. Invoking Lemma \ref{lemma:prior_sparsity}, we obtain
\begin{align}
&\Pi\left(\sum_{j=1}^p\|\bB_{j*}\|_1\zeta_j>t_n/2\right)\nonumber\\
&\quad\leq \expect_\Pi\left\{\Pi\left(\sum_{j=1}^p\|\bB_{j*}\|_1\zeta_j>t_n/2\mathrel{\Big|}\bzeta,\theta\right)\mathbbm{1}(\bzeta:|\mathrm{supp}_{\delta_n}(\bB)|\leq \beta s)\right\}
+ \Pi\left(|\mathrm{supp}_{\delta_n}(\bB)|\leq \beta s\right)\nonumber\\
&\quad\leq \exp\left[-C\min\left\{\left(\frac{t_n}{\beta sr}\right)^2,\frac{t_n}{r}\right\}\right]
 + 2\exp\left\{-\min\left(\frac{\beta\kappa}{2},\frac{\beta}{2\mathrm{e}} - 2\gamma\right)s\log p\right\}\nonumber.
\end{align}

\vspace*{1ex}
\noindent
\textbf{Combining upper bounds: }
Combining the previous two upper bounds, we obtain
\begin{align}
&\Pi\left[\sum_{j=1}^p\|\bB_{j*}\|_1\mathbbm{1}\{j\in\mathrm{supp}_{\delta_n}(\bB)\cup\mathrm{supp}(\bB_0)\}\geq t_n\right]\nonumber\\
&\quad \leq 2\exp\left[-C\min\left\{\left(\frac{t_n}{\beta sr}\right)^2,\left(\frac{t_n}{s}\right)^2,\frac{t_n}{r}\right\}\right] + 
3 \exp\left\{-\min\left(\frac{\beta\kappa}{2},\frac{\beta}{2\mathrm{e}}-2\gamma\right)s\log p\right\}\nonumber,
\end{align}
and the proof is completed. 
\end{proof}


\section{Proofs of Results in Section \ref{sub:proof_sketch_and_auxiliary_results}} 
\label{sec:proofs_of_results_in_section_3_3}

\begin{customlemma}{3.4}
Let $\calK_n(\eta)=\{\|\bSigma-\bSigma_0\|_{\mathrm{F}}\leq\eta\}$ and $\eta<\sigma_0^2/2$. Then there exists some event $\calA_n$ such that
\[
\calA_n\subset\left\{D_n\geq \Pi_n\{\bSigma\in \calK_n(\eta)\}\exp\left[-\left\{\frac{C_3\log\rho}{2(\lambda_{0r}+\sigma_0^2)}+1\right\}n\eta^2\right]\right\}
\]
for some absolute constants $C_3>0$, and 
\[
\prob_0(\calA_n^c)\leq\exp\left\{-\tilde C_3\min\left(\frac{n\eta^2}{\|\bSigma_0^{-1}\|_2^2}, n\eta^2\right)\right\},
\]
where $\rho = 2(\lambda_{01}+\sigma_0^2)/(\lambda_{0r}+\sigma_0^2)$ and $C_3>0$ are some absolute constants.
\end{customlemma}

\begin{proof}[{\bf Proof of Lemma \ref{lemma:evidence_lower_bound}}]
To prove Lemma \ref{lemma:evidence_lower_bound}, we need the following auxiliary matrix inequality:
\begin{lemma}[\citealp{pati2014posterior}, Supplement Lemma 1.3]
\label{lemma:matrix_eigen_inequality}
Let $\bSigma,\bSigma_0$ be $p\times p$ positive definite matrices and $\eta\in(0,1)$. If $\|\bSigma-\bSigma_0\|_{\mathrm{F}}\leq\eta$ and $\eta<2\lambda_r(\bSigma_0)$, then
\[
\log\det\left(\bSigma_0\bSigma^{-1}\right)-\mathrm{tr}\left(\bSigma_0\bSigma^{-1}-\eye\right)\geq-C_3\frac{\eta^2\log \rho}{\lambda_{r}(\bSigma_0)}
\]
for some absolute constant $C_3>0$, where $\rho = 2\lambda_1(\bSigma_0)/\lambda_r(\bSigma_0)$.
\end{lemma}

\noindent
Denote $\Pi\{\cdot\mid\calK_n(\eta)\} = \Pi\{\cdot\cap \calK_n)/\Pi_n(\calK_n(\eta)\}$ to be the re-normalized restriction of $\Pi$ on $\calK_n(\eta)$. Define random variable
\begin{align}
w_{ni} &= \int\log\frac{p(\by_i\mid\bSigma)}{p(\by_i\mid\bSigma_0)}\Pi\{\mathrm{d}\bSigma\mid\calK_n(\eta)\}\nonumber\\
&= \int \left\{\frac{1}{2}\log\det(\bSigma_0\bSigma^{-1})\right\}\Pi\{\mathrm{d}\bSigma\mid\calK_n(\eta)\} + \frac{1}{2}\by_i\transpose
\left[
\int\left(\bSigma_0^{-1}-\bSigma^{-1}\right)\Pi\{\mathrm{d}\bSigma\mid\calK_n(\eta)\}
\right]
\by_i\nonumber.
\end{align}
Invoking Fubini's theorem and Lemma \ref{lemma:matrix_eigen_inequality}, we derive
\begin{align}
\expect_0(w_{ni}) &= 
\int \left\{\frac{1}{2}\log\det(\bSigma_0\bSigma^{-1})\right\}\Pi\{\mathrm{d}\bSigma\mid\calK_n(\eta)\}
+ 
\frac{1}{2}\int\expect_0\left\{\by_i\transpose
\left(\bSigma_0^{-1}-\bSigma^{-1}\right)\by_i
\right\}\Pi\{\mathrm{d}\bSigma\mid\calK_n(\eta)\}\nonumber\\
& = \frac{1}{2}\int\left\{\log\det\left(\bSigma_0\bSigma^{-1}\right) + \mathrm{tr}\left(\eye - \bSigma_0\bSigma^{-1}\right)\right\}\Pi\{\mathrm{d}\bSigma\mid\calK_n(\eta)\}
\geq -\frac{C_3\log\rho}{2(\lambda_{0r}+\sigma_0^2)}\eta^2\nonumber.
\end{align}
Hence by Jensen's inequality,
\begin{align}
\log D_n - \log\Pi\{\bSigma\in\calK_n(\eta)\}&
\geq \log\left[
\int_{\calK_n(\eta)}\exp\left\{\ell_n(\bSigma) - \ell_n(\bSigma_0)\right\}\frac{\Pi(\mathrm{d}\bSigma)}{\Pi\{\calK_n(\eta)\}}
\right]
\nonumber\\&
 = \log\left[
\int \exp\left\{\ell_n(\bSigma) - \ell_n(\bSigma_0)\right\}\Pi\{\mathrm{d}\bSigma\mid\calK_n(\eta)\}
\right]\nonumber\\
& \geq \int\{\ell_n(\bSigma)-\ell_n(\bSigma_0)\}\Pi\{\mathrm{d}\bSigma\mid\calK_n(\eta)\}\nonumber\\
& = n\expect_0(w_{ni}) + \sum_{i=1}^n\{w_{ni} - \expect_0(w_{ni})\}\nonumber\\
& \geq -\frac{C_3\log\rho}{2(\lambda_{0r} + \sigma_0^2)}n\eta^2 + \sum_{i=1}^n\{w_{ni} - \expect_0(w_{ni})\}.\nonumber
\end{align}
Now let $\calA_n = \{|\sum_{i=1}^n\{w_{ni} - \expect_0(w_{ni})\}|\leq n\eta^2\}$. Clearly, 
\begin{align}
\calA_n&\subset\left\{\log D_n - \log \Pi\{\bSigma\in\calK_n(\eta)\}\geq -\left\{\frac{C_3\log\rho}{2(\lambda_{0r} + \sigma_0^2)} + 1\right\}n\eta^2\right\}\nonumber\\
&=\left\{
D_n \geq\Pi\{\bSigma\in\calK_n(\eta)\}\exp\left[-\left\{\frac{C_3\log\rho}{2(\lambda_{0r} + \sigma_0^2)} + 1\right\}n\eta^2\right]
\right\}\nonumber.
\end{align}
We now analyze the probabilistic bound of $\calA_n^c$. 
Recall $\bSigma_0 = \bU_0\bLambda_0\bU_0\transpose + \sigma_0^2\eye_p$. Let $\bU_{0\perp}$ to be the orthonormal $(p-r)$-frame in $\mathbb{R}^p$ such that $[\bU_0,\bU_{0\perp}]\in\mathbb{O}(p)$, and denote 
\[
\bSigma_0^{1/2} = [\bU_0,\bU_{0\perp}]\mathrm{diag}\{\lambda_1(\bSigma_0)^{1/2},\ldots,\lambda_p(\bSigma_0)^{1/2}\}[\bU_0,\bU_{0\perp}]\transpose. 
\]
Clearly, $\bSigma_0 = (\bSigma_0^{1/2})^2$, and by denoting $\bv_i = \bSigma_0^{-1/2}\by_i$, we have $\bv_i \sim\mathrm{N}_p(\zero_p, \eye_p)$ under $\prob_0$. Re-writing $w_{ni} - \expect_0(w_{ni})$ in terms of $\bv_i$, we have
\[
w_{ni} - \expect_0(w_{ni}) = \bv_i\transpose\bOmega\bv_i - \expect_0(\bv_i\bOmega\bv_i),
\]
where
\[
\bOmega = \frac{1}{2}\int\left(\eye_p - \bSigma^{1/2}_0\bSigma^{-1}\bSigma_0^{1/2}\right)\Pi\{\mathrm{d}\bSigma\mid\calK_n(\eta)\}.
\]
Let $\bOmega = \bU_\bOmega\bD_\bOmega\bU_\bOmega\transpose$ be the spectral decomposition of $\bOmega$, and let $\bx_i = \bU_\bOmega\transpose\bv_i$. Then we proceed to bound
\begin{align}
\prob_0(\calA_n^c)&\leq \prob_0\left(\left|\sum_{i=1}^n\{w_{ni} - \expect_0(w_{ni})\}\right|\geq n\eta^2\right)
\nonumber\\&
= \prob_0\left(\left|\sum_{i=1}^n\left\{\bx_i\transpose\bD_\bOmega\bx_i - \expect_0\left(\bx_i\transpose\bD_\bOmega\bx_i\right)\right\}\right|\geq n\eta^2\right)\nonumber\\
&= \prob_0\left(\left|\sum_{i=1}^n\sum_{j=1}^p\lambda_j(\bOmega)\left\{x_{ij}^2 - \expect_0(x_{ij}^2)\right\}\right|\geq n\eta^2\right)\nonumber\\
&\leq 2\exp\left[-C_3'\min\left\{\frac{n^2\eta^4}{n\sum_{j=1}^p\lambda_j(\bOmega)^2},\frac{n\eta^2}{\max_{j\in[p]}\lambda_j(\bOmega)}\right\}\right]
\nonumber
\end{align}
for some absolute constant $C_3'>0$, 
where the large deviation inequality for sub-exponential random variables is applied again in the last inequality. Observe that over $\calK_n(\eta)$ for $\eta\leq \sigma_0^2/2$, 
\begin{align}
\|\bSigma^{-1}\|_2&\leq \|\bSigma^{-1}-\bSigma_0^{-1}\|_2+\|\bSigma_0^{-1}\|_2
= \|\bSigma^{-1}(\bSigma-\bSigma_0)\bSigma_0^{-1}\|_2 + \|\bSigma_0^{-1}\|_2\nonumber\\
&\leq \|\bSigma^{-1}\|_2\|\bSigma-\bSigma_0\|_{\mathrm{F}}\|\bSigma_0^{-1}\|_2 + \|\bSigma_0^{-1}\|_2 
\leq \frac{\eta}{\sigma_0^2}\|\bSigma^{-1}\|_2 + \|\bSigma_0^{-1}\|_2\nonumber,
\end{align}
implying that $\|\bSigma^{-1}\|_2\leq 2\|\bSigma_0^{-1}\|_2$. 
Also observe that
\begin{align}
\sum_{j=1}^p\lambda_j(\bOmega)^2 &= \|\bOmega\|_{\mathrm{F}}^2 
\leq\frac{1}{4}\int \left\|
\eye_p - \bSigma_0^{1/2}\bSigma^{-1}\bSigma_0^{1/2}
\right\|_{\mathrm{F}}^2\Pi\{\mathrm{d}\bSigma\mid\calK_n(\eta)\}
 = \frac{1}{4}\int \left\|
\eye_p - \bSigma^{-1}\bSigma_0
\right\|_{\mathrm{F}}^2\Pi\{\mathrm{d}\bSigma\mid\calK_n(\eta)\}\nonumber\\
&\leq \frac{1}{4}\int \|\bSigma^{-1}\|_2^2\|\bSigma-\bSigma_0\|_{\mathrm{F}}^2\Pi\{\mathrm{d}\bSigma\mid\calK_n(\eta)\}
\leq \|\bSigma_0^{-1}\|_2^2\int\|\bSigma-\bSigma_0\|_{\mathrm{F}}^2\Pi\{\mathrm{d}\bSigma\mid\calK_n(\eta)\}
\leq \|\bSigma_0^{-1}\|_2^2\eta^2.\nonumber
\end{align}
We finally obtain
\[
\prob_0(\calA_n^c)\leq 2\exp\left\{-\tilde C_3\min\left(\frac{n\eta^2}{\|\bSigma_0^{-1}\|_2^2},n\eta^2\right)\right\}
\]
for some absolute constant $\tilde C_3>0$. 
\end{proof}

\begin{customlemma}{3.5}
Assume the data $\by_1,\ldots,\by_n$ follows $\mathrm{N}_p(\zero_p, \bSigma)$, $1\leq r\leq p$. Suppose $\bU_0\in\mathbb{O}(p, r)$ satisfy $|\mathrm{supp}(\bU_0)|\leq s$, and $r\leq s\leq p$. 
For any positive $\delta$, $t$, and $\tau$, define 
\[
\calF(\delta,\tau,t) = \left\{\bB\in\mathbb{R}^{p\times r}:|\mathrm{supp}_\delta(\bB)|\leq \tau, \sum_{j=1}^p\|\bB_{j*}\|_2^2\mathbbm{1}\{j\in\mathrm{supp}_\delta(\bB)\cup \mathrm{supp}(\bU_0)\}\leq t^2
\right\}.
\]
Let the positive sequences $(\delta_n, \tau_n, t_n,\eps_n)_{n=1}^\infty$ satisfy $(\sqrt{p}\delta_n + 2t_n)\sqrt{p}\delta_n\leq M_1 \eps_n$ for some constant $M_1>0$, and $\eps_n\leq 1$. Consider testing $H_0:\bSigma = \bSigma_0 = \bU_0\bLambda_0\bU_0\transpose + \sigma_0^2\eye_p$ versus
\[
H_1:\bSigma \in\left\{\bSigma = \bB\bB\transpose + \sigma^2\eye_p:\|\bSigma-\bSigma_0\|_2>M\eps_n, \bB\in\calF(\delta_n, \tau_n, t_n)\right\}.
\]
Then for each $M\geq \max\{M_1/2, (128\|\bSigma_0\|_2^4)^{1/3}\}$, there exists a test function $\phi_n:\mathbb{R}^{n\times p}\to [0,1]$, such that
\begin{align}
\expect_0(\phi_n)&\leq 3\exp\left\{(2 + C_4)(\tau_n\log p + 2s_n) - \frac{C_4\sqrt{M}}{\sqrt{2}}n\eps_n^2\right\},\nonumber\\
\sup_{\bSigma\in H_1}\expect_\bSigma(1-\phi_n)&\leq \exp\left\{C_4(\tau_n + 2s_n) - \frac{C_4M}{8}n\eps_n^2\right\}\nonumber
\end{align}
for some absolute constant $C_4>0$.
\end{customlemma}

\begin{proof}[{\bf Proof of Lemma \ref{lemma:existence_test}}]
To proof Lemma \ref{lemma:existence_test}, we need the following oracle testing lemma from \cite{gao2015rate}:
\begin{lemma}[\citealp{gao2015rate}]
\label{lemma:matrix_test_bound}
Let $\by_i\sim\mathrm{N}_d(\zero_d, \bSigma$, where $\bSigma\in\mathbb{R}^{d\times d}$. Then for any $M>0$, there exists a test function $\phi_n$ such that
\begin{align}
\expect_{\bSigma^{(1)}}(\phi_n)&\leq \exp\left(C_4d-\frac{C_4M^2}{4\|\bSigma^{(1)}\|_2^2}n\eps^2\right)+2\exp\left(C_4d-C_4\sqrt{M}n\right)\nonumber,\\
\sup_{\{\bSigma^{(2)}:\|\bSigma^{(2)}-\bSigma^{(1)}\|_2>M\eps\}}\expect_{\bSigma^{(2)}}(1-\phi_n)&\leq \exp\left[C_4d - \frac{C_4Mn\eps^2}{4}\max\left\{1,\frac{M}{(\sqrt{M}+2)^2\|\bSigma^{(1)}\|_2^2}\right\}\right]\nonumber.
\end{align}
with some absolute constant $C_4>0$.
\end{lemma}

\noindent
Let $S_0 = \mathrm{supp}(\bU_0)$ and $S(\delta) = \mathrm{supp}_\delta(\bB)$. Then there exists some permutation matrix $\bP$ such that
\[
\bB = \bP\begin{bmatrix*}
\bB_\delta\\ \bA_\delta
\end{bmatrix*}\quad\text{and}\quad
\bU_0 = \bP\begin{bmatrix*}
\bU_{0\delta}\\ \zero
\end{bmatrix*},
\]
where $\bB_\delta$ and $\bU_{0\delta}$ are $|S(\delta)\cup S_0|\times r$ matrices. Hence for $\bSigma\in\calF(\delta,\tau,t)$, it holds that
\begin{align}
\|\bSigma-\bSigma_0\|_2 &= \left\|
\bP\begin{bmatrix*}
\bB_\delta\bB_\delta\transpose + \sigma^2\eye - \bU_{0\delta}\bLambda_0\bU_{0\delta}\transpose
-\sigma_0^2\eye & \bB_\delta\bA_\delta\transpose\\
\bA_\delta\bB_\delta\transpose & \bA_\delta\bA_\delta\transpose + (\sigma^2 - \sigma_0^2)\eye_d
\end{bmatrix*}\bP\transpose
\right\|_2\nonumber\\
&\leq \left\|
\begin{bmatrix*}
\bB_\delta\bB_\delta\transpose + \sigma^2\eye & \zero\\
\zero & \sigma^2 
\end{bmatrix*} - 
\begin{bmatrix*}
\bU_{0\delta}\bLambda_0\bU_{0\delta}\transpose + \sigma_0^2\eye & \zero\\
\zero & \sigma^2_0
\end{bmatrix*}
\right\|_2 + 
\left\|
\begin{bmatrix*}
\zero & \bB_\delta\bA_\delta\transpose\\
\bA_\delta\bB_\delta\transpose & \bA_\delta\bA_\delta\transpose
\end{bmatrix*}
\right\|_{\mathrm{F}}\nonumber\\
&\leq \left\|\bSigma_{S(\delta)} - \bSigma_{S(\delta)}^{(0)}\right\|_2 + 
(\|\bA_\delta\|_2^2 + 2\|\bB_\delta\|_2^2)^{1/2}\|\bA_\delta\|_{\mathrm{F}}\nonumber\\
&\leq \left\|\bSigma_{S(\delta)} - \bSigma_{S(\delta)}^{(0)}\right\|_2 + (\sqrt{p}\delta + 2t)\sqrt{p}\delta\nonumber,
\end{align}
where
\[\bSigma_{S(\delta)} = 
\begin{bmatrix*}
\bB_\delta\bB_\delta\transpose + \sigma^2\eye & \zero\\
\zero & \sigma^2 
\end{bmatrix*}\quad\text{and}\quad
\bSigma_{S(\delta)}^{(0)} = 
\begin{bmatrix*}
\bU_{0\delta}\bLambda_0\bU_{0\delta}\transpose + \sigma_0^2\eye & \zero\\
\zero & \sigma^2_0
\end{bmatrix*}.
\]
By taking $M \geq 2M_1$, we obtain
\begin{align}
&\left\{
\bSigma = \bB\bB\transpose + \sigma^2\eye:\|\bSigma - \bSigma_0\| > M\eps_n, \bB\in\calF(\delta_n,\tau_n, t_n)
\right\}\nonumber\\
&\quad\subset 
\left\{\bSigma:
\left\|\bSigma_{S(\delta)} - \bSigma_{S(\delta)}^{(0)}\right\|_2>\frac{M}{2}\eps_n:\bB\in\calF(\delta_n,\tau_n,t_n)
\right\}\nonumber\\
&\quad\subset\bigcup_{S(\delta_n)\subset[p]:|S(\delta)|\leq\tau_n}\left\{\bSigma:
\left\|\bSigma_{S(\delta)} - \bSigma_{S(\delta)}^{(0)}\right\|_2>\frac{M}{2}\eps_n\right\}\nonumber.
\end{align}
Since both $\bSigma_{S(\delta_n)}$ and $\bSigma_{S(\delta_n)}^{(0)}$ are $(|S(\delta_n)\cup S_0|+1)\times (|S(\delta_n)\cup S_0|+1)$ square matrices, and
\[
|S(\delta_n)\cup S_0|+1\leq |S(\delta_n)|+S_0 + 1\leq \tau_n + 2s_n,
\]
then for each $S(\delta_n)\subset[p]$ with $|S(\delta_n)|\leq\tau_n$, 
and for each $M\geq \max\{M_1/2, (128\|\bSigma_0\|_2^4)^{1/3}\}$, 
we invoke Lemma \ref{lemma:matrix_test_bound} to construct a test $\phi_{S(\delta_n)}$ depending on the index set $S(\delta_n)$, such that the type I error probability satisfies
\begin{align}
\expect_{\bSigma_{S(\delta_n)}^{(0)}}\left(\phi_{S(\delta_n)}\right)
&\leq \exp\left\{
C_4(\tau_n + 2s_n) - \frac{C_4M^2n\eps^2_n}{16\|\bSigma^{(0)}_{S(\delta_n)}\|_2^2}
\right\} + 2\exp\left\{
C_4(\tau_n + 2s_n) - C_4\sqrt{\frac{M}{2}}n
\right\}\nonumber\\
&\leq 3\exp\left\{
C_4(\tau_n + 2s_n) - C_4\min\left(\frac{M^2}{16\|\bSigma_0\|_2^2},\sqrt{\frac{M}{2}}\right)n\eps_n^2
\right\}\nonumber\\
&\leq 3\exp\left\{
C_4(\tau_n + 2s_n) - C_4\sqrt{\frac{M}{2}}n\eps_n^2
\right\},\nonumber
\end{align}
and for all $\bSigma_{S(\delta_n)}\in\{\|\bSigma_{S(\delta_n)} - \bSigma_{S(\delta_n)}^{(0)}\|_2>M\eps_n/2\}$, the type II error probability satisfies
\begin{align}
\expect_{\bSigma_{S(\delta_n)}^{(1)}}\left(1 - \phi_{S(\delta_n)}\right)
&\leq 
\exp\left[
C_4(\tau_n + 2s_n) - \frac{C_4Mn\eps_n^2}{8}\max\left\{1,\frac{M}{(\sqrt{M} + 2)^2\|\bSigma_{S(\delta_n)}^{(0)}\|_2^2}\right\}
\right]\nonumber\\
&\leq\exp\left\{
C_4(\tau_n + 2s_n) - \frac{C_4Mn\eps_n^2}{8}
\right\}\nonumber.
\end{align}
Notice that for each index set $S(\delta_n)$, the test function $\phi_{S(\delta_n)}$ is only a function of $\bY_n$ through the coordinates $[y_{ij}:i\in[n], j \in S(\delta_n)\cup S_0]$. Hence, $\expect_{\bSigma_{S(\delta_n)}^{(0)}}(\phi_{S(\delta_n)}) = \expect_0(\phi_{S(\delta_n)})$, and for any $p\times p$ covariance matrix $\bSigma$ with $\|\bSigma_{S(\delta_n)} - \bSigma_{S(\delta_n)}^{(0)}\|_2>M\eps_n/2$, it holds that $\expect_{\bSigma_{S(\delta_n)}}(1 - \phi_{S(\delta_n)}) = \expect_\bSigma(1 - \phi_{S(\delta_n)})$. Therefore, by aggregating the test functions 
\[\phi_n = \max_{S(\delta_n)\subset[p]:|S(\delta_n)|\leq\tau_n}\phi_{S(\delta_n)},\]
we obtain
\begin{align}
\expect_0(\phi_n) &\leq\sum_{S(\delta_n)\subset[p]:|S(\delta_n)|\leq\tau_n}\expect_{\bSigma_{S(\delta_n)}^{(0)}}(\phi_{S(\delta_n)})
\leq 3\sum_{q=0}^{\lfloor\tau_n\rfloor}\frac{p!}{q!(p-q)!}\exp\left\{C_4(\tau_n + 2s_n) - C_4\sqrt{\frac{M}{2}}n\eps_n^2\right\}\nonumber\\
&\leq 3(\tau_n + 1)\exp(\tau_n\log p)\exp\left\{C_4(\tau_n + 2s_n) - C_4\sqrt{\frac{M}{2}}n\eps_n^2\right\}\nonumber\\
&\leq 3\exp\left\{\tau_n + \tau_n\log p + C_4(\tau_n + 2s_n) - C_4\sqrt{\frac{M}{2}}n\eps_n^2\right\}\nonumber\\
&\leq 3\exp\left\{(2 + C_4)(\tau_n\log p + 2s_n) - C_4\sqrt{\frac{M}{2}}n\eps_n^2\right\}\nonumber,
\end{align}
and
\begin{align}
\sup_{\bSigma\in H_1}\expect_\bSigma(1-\phi_n)&\leq \sup_{S(\delta_n)\subset[p]:|S(\delta_n)|\leq \tau_n}\sup_{\left\{\bSigma:\|\bSigma_{S(\delta_n)} - \bSigma_{S(\delta_n)}^{(0)}\|_2>M\eps_n/2\right\}}
\expect_{\bSigma_{S(\delta_n)}}\left(1-\phi_{S(\delta_n)}\right)
\nonumber\\
&\leq \exp\left\{C_4(\tau_n + 2s_n) - \frac{C_4M}{8}n\eps_n^2\right\}\nonumber.
\end{align}
The proof is thus completed.
\end{proof}


\section{Proof of Lemma \ref{lemma:existence_test_infinity_norm}} 
\label{sec:proof_of_lemma_test_infinity_norm}
\begin{customlemma}{B.1}
Assume the data $\by_1,\ldots,\by_n$ follows $\mathrm{N}_p(\zero_p, \bSigma)$, $1\leq r\leq p$. Suppose $\bU_0\in\mathbb{O}(p, r)$ satisfies $|\mathrm{supp}(\bU_0)|\leq s$, and $r\leq s\leq p$. 
For any positive $\delta$, $t$, and $\tau$, define 
\[
\calG(\delta,\tau,t) = \left\{\bB\in\mathbb{R}^{p\times r}:|\mathrm{supp}_\delta(\bB)|\leq \tau, \sum_{j=1}^p\|\bB_{j*}\|_1\mathbbm{1}\{j\in\mathrm{supp}_\delta(\bB)\cup \mathrm{supp}(\bU_0)\}\leq t
\right\}.
\]
Let the positive sequences $(\delta_n, \tau_n, t_n,\eps_n)_{n=1}^\infty$ satisfy $\max(p\delta_nt_n, \delta_nt_n + p\delta_n^2)\leq M_1 \eps_n$ for some constant $M_1>0$, and $\eps_n\leq 1$. Consider testing $H_0:\bSigma = \bSigma_0 = \bU_0\bLambda_0\bU_0\transpose + \sigma_0^2\eye_p$ versus
\[
H_1:\bSigma \in\left\{\bSigma = \bB\bB\transpose + \sigma^2\eye_p:\|\bSigma-\bSigma_0\|_{\infty}>M\eps_n, \bB\in\calG(\delta_n, \tau_n, t_n)\right\}.
\]
Then there exists some absolute constant $C_6 > 0$, such that for each 
\[
M\in\left[\max\left\{\frac{M_1}{2}, 8, \frac{8(\log 2)^2}{C_6^2}\right\}, \frac{2\min(1, 2\|\bSigma_0\|_2)}{\eps_n}\right],
\]
 there exists a test function $\phi_n:\mathbb{R}^{n\times p}\to [0,1]$, such that
\begin{align}
\expect_0(\phi_n)&\leq 
12\exp\left\{6(\tau_n\log p + 2s_n) - 
C_6\min\left(\frac{1}{2},\frac{\|\bSigma_0\|_\infty^2}{\sqrt{2}}\right){\frac{\sqrt{M}n\eps_n^2}{\|\bSigma_0\|_\infty^2}}
\right\}
,\nonumber\\
\sup_{\bSigma\in H_1}\expect_\bSigma(1-\phi_n)&\leq 
4\exp\left\{4(\tau_n + 2s_n) - C_6\min\left(\frac{\|\bSigma_0\|_\infty^2}{8},\frac{1}{32}\right)\frac{Mn\eps_n^2}{\|\bSigma_0\|_\infty^2}\right\}.
\nonumber
\end{align}
\end{customlemma}
\begin{proof}[\bf{Proof of Lemma \ref{lemma:existence_test_infinity_norm}}]
The proof of Lemma \ref{lemma:existence_test_infinity_norm} is quite similar to that of Lemma \ref{lemma:existence_test}, except that the following oracle test lemma for the infinity norm is applied instead of Lemma \ref{lemma:matrix_test_bound}.
\begin{lemma}\label{lemma:oracle_test_infinity_norm}
Let $\bx_1,\ldots,\bx_n\sim\mathrm{N}_d(\zero_d,\bSigma)$ independently, where $\bSigma\in\mathbb{R}^{d\times d}$. Let $\eps\in(0,1)$. Then there exists some absolute constant $C_6>0$, such that for each $M$ satisfying $M\geq \max[4, \{(2\log 2)/C_6\}^2]$, and $M\eps\leq\min(1, 2\|\bSigma_0\|_2)$, there exists a test function $\phi_n:\mathbb{R}^{n\times d}\to[0, 1]$, such that
\begin{align}
\expect_0(\phi_n)&\leq 4 \exp\left(4d - \frac{C_6M^2n\eps^2}{4\|\bSigma_0\|_\infty^2}\right) + 8\exp\left(4d - \frac{C_6\sqrt{M}n}{2}\right)\nonumber\\
\sup_{\{\|\bSigma - \bSigma_0\|_\infty>M\eps\}}\expect_\bSigma(1 - \phi_n)&\leq 
4\exp\left\{
 4d - \frac{C_6Mn\eps^2}{4}\min\left(1, \frac{1}{4\|\bSigma_0\|_\infty^2}\right)
 \right\}\nonumber.
\end{align}
\end{lemma}

\vspace*{1ex}
\noindent
Let $S_0 = \mathrm{supp}(\bU_0)$ and $S(\delta) = \mathrm{supp}_\delta(\bB)$. Then there exists some permutation matrix $\bP$ such that
\[
\bB = \bP\begin{bmatrix*}
\bB_\delta\\ \bA_\delta
\end{bmatrix*}\quad\text{and}\quad
\bU_0 = \bP\begin{bmatrix*}
\bU_{0\delta}\\ \zero
\end{bmatrix*},
\]
where $\bB_\delta$ and $\bU_{0\delta}$ are $|S(\delta)\cup S_0|\times r$ matrix. Hence for $\bSigma\in\calG(\delta,\tau,t)$, it holds that
\begin{align}
\|\bSigma-\bSigma_0\|_\infty &= \left\|
\bP\begin{bmatrix*}
\bB_\delta\bB_\delta\transpose + \sigma^2\eye - \bU_{0\delta}\bLambda_0\bU_{0\delta}\transpose
-\sigma_0^2\eye & \bB_\delta\bA_\delta\transpose\\
\bA_\delta\bB_\delta\transpose & \bA_\delta\bA_\delta\transpose + (\sigma^2 - \sigma_0^2)\eye_d
\end{bmatrix*}\bP\transpose
\right\|_\infty\nonumber\\
&\leq \left\|
\begin{bmatrix*}
\bB_\delta\bB_\delta\transpose + \sigma^2\eye & \zero\\
\zero & \sigma^2 
\end{bmatrix*} - 
\begin{bmatrix*}
\bU_{0\delta}\bLambda_0\bU_{0\delta}\transpose + \sigma_0^2\eye & \zero\\
\zero & \sigma^2_0
\end{bmatrix*}
\right\|_\infty + 
\left\|
\begin{bmatrix*}
\zero & \bB_\delta\bA_\delta\transpose\\
\bA_\delta\bB_\delta\transpose & \bA_\delta\bA_\delta\transpose
\end{bmatrix*}
\right\|_\infty\nonumber\\
&\leq \left\|\bSigma_{S(\delta)} - \bSigma_{S(\delta)}^{(0)}\right\|_\infty + 
\max\left(\|\bB_\delta \bA_\delta\transpose\|_\infty, \|\bA_\delta\bB_\delta\transpose\|_\infty + \|\bA_\delta\bA_\delta\transpose\|_\infty\right)
\nonumber,
\end{align}
where
\[\bSigma_{S(\delta)} = 
\begin{bmatrix*}
\bB_\delta\bB_\delta\transpose + \sigma^2\eye & \zero\\
\zero & \sigma^2 
\end{bmatrix*}\quad\text{and}\quad
\bSigma_{S(\delta)}^{(0)} = 
\begin{bmatrix*}
\bU_{0\delta}\bLambda_0\bU_{0\delta}\transpose + \sigma_0^2\eye & \zero\\
\zero & \sigma^2_0
\end{bmatrix*}.
\]
Since
\begin{align}
\max\left(\|\bB_\delta \bA_\delta\transpose\|_\infty, \|\bA_\delta\bB_\delta\transpose\|_\infty + \|\bA_\delta\bA_\delta\transpose\|_\infty\right)
\leq\max\left(\|\bB_\delta\|_\infty\|\bA_\delta\transpose\|_\infty, \|\bA_\delta\|_\infty\|\bB_\delta\transpose\|_\infty + \|\bA_\delta\|_\infty\|\bA_\delta\transpose\|_\infty\right)\nonumber,
\end{align}
and
\begin{align}
\|\bB_\delta\|_\infty &= \max_{j\in S_0\cup S(\delta)}\|\bB_{j*}\|_1\leq \sum_{j = 1}^p\|\bB_{j*}\|_1\mathbbm{1}\{j\in S(\delta)\cup S_0\}\leq t,\nonumber\\
\|\bB_\delta\transpose\|_\infty & \leq \sum_{j = 1}^p\|\bB_{j*}\|_1\mathbbm{1}\{j\in S(\delta)\cup S_0\}\leq t,\nonumber\\
\|\bA_\delta\|_\infty &= \max_{j\in S_0^c\cap S(\delta)^c}\|\bB_{j*}\|_1\leq \max_{j\in S(\delta)^c}\|\bB_{j*}\|_1\leq \delta,\nonumber\\
\|\bA_\delta\transpose\|_\infty &\leq \sum_{j = 1}^p\|\bB_{j*}\|_1\mathbbm{1}\left\{j\in S_0^c\cap S(\delta)^c\right\}\leq \sum_{j = 1}^p\|\bB_{j*}\|_1\mathbbm{1}\left\{j\in S(\delta)^c\right\}\leq p\delta\nonumber,
\end{align}
it follows that
\[
\|\bSigma-\bSigma_0\|_\infty\leq 
\left\|\bSigma_{S(\delta)} - \bSigma_{S(\delta)}^{(0)}\right\|_\infty + 
\max(p\delta_nt_n, \delta_nt_n + p\delta_n^2)\leq\left\|\bSigma_{S(\delta)} - \bSigma_{S(\delta)}^{(0)}\right\|_\infty + M_1\eps_n.
\]
By taking $M \geq 2M_1$, we obtain
\begin{align}
&\left\{
\bSigma = \bB\bB\transpose + \sigma^2\eye:\|\bSigma - \bSigma_0\|_\infty > M\eps_n, \bB\in\calG(\delta_n,\tau_n, t_n)
\right\}\nonumber\\
&\quad\subset 
\left\{\bSigma:
\left\|\bSigma_{S(\delta)} - \bSigma_{S(\delta)}^{(0)}\right\|_\infty>\frac{M}{2}\eps_n:\bB\in\calG(\delta_n,\tau_n,t_n)
\right\}\nonumber\\
&\quad\subset\bigcup_{S(\delta_n)\subset[p]:|S(\delta)|\leq\tau_n}\left\{\bSigma:
\left\|\bSigma_{S(\delta)} - \bSigma_{S(\delta)}^{(0)}\right\|_\infty>\frac{M}{2}\eps_n\right\}\nonumber.
\end{align}
Since both $\bSigma_{S(\delta_n)}$ and $\bSigma_{S(\delta_n)}^{(0)}$ are $(|S(\delta_n)\cup S_0|+1)\times (|S(\delta_n)\cup S_0|+1)$ square matrices, and
\[
|S(\delta_n)\cup S_0|+1\leq |S(\delta_n)|+S_0 + 1\leq \tau_n + 2s_n,
\]
then for each $S(\delta_n)\subset[p]$ with $|S(\delta_n)|\leq\tau_n$, 
and for each 
\[
M\in\left[\max\left\{\frac{M_1}{2}, 8, \frac{8(\log 2)^2}{C_6^2}\right\}, \frac{2\min(1, 2\|\bSigma_0\|_2)}{\eps_n}\right],
\]
(and hence $M/2\geq \max\{4, (2\log 2)^2/C_6^2\}$, $(M/2)\eps_n\leq \min(1, \|\bSigma_{S(\delta_n)}^{(0)}\|_2) = \min(1, \|\bSigma_0\|_2)$),
we invoke Lemma \ref{lemma:oracle_test_infinity_norm} to construct a test $\phi_{S(\delta_n)}$ depending on the index set $S(\delta_n)$, such that the type I error probability satisfies
\begin{align}
\expect_{\bSigma_{S(\delta_n)}^{(0)}}\left(\phi_{S(\delta_n)}\right)
&\leq 
4\exp\left\{
4(\tau_n + 2s_n) - \frac{C_6M^2n\eps_n^2}{16\|\bSigma^{(0)}_{S(\delta_n)}\|_\infty^2}
\right\} + 8\exp\left\{
4(\tau_n + 2s_n) - C_6\sqrt{\frac{M}{2}}n
\right\}\nonumber\\
&\leq 12\exp\left\{
4(\tau_n + 2s_n) - C_6\min\left(\frac{M^2}{16\|\bSigma_0\|_\infty^2},\sqrt{\frac{M}{2}}\right)n\eps_n^2
\right\}\nonumber\\
&\leq 12\exp\left\{
4(\tau_n + 2s_n) - C_6\min\left(\frac{1}{2}, \frac{\|\bSigma_0\|_\infty^2}{\sqrt{2}}\right)\frac{\sqrt{M}n\eps_n^2}{\|\bSigma_0\|_\infty^2}
\right\}\nonumber.
\end{align}
In addition, for all $\bSigma_{S(\delta_n)}\in\{\|\bSigma_{S(\delta_n)} - \bSigma_{S(\delta_n)}^{(0)}\|_2>M\eps_n/2\}$, the type II error probability satisfies
\begin{align}
\expect_{\bSigma_{S(\delta_n)}^{(1)}}\left(1 - \phi_{S(\delta_n)}\right)
&\leq 
4\exp\left\{
4(\tau_n + 2s_n) - \frac{C_6Mn\eps_n^2}{8}\min\left(1,\frac{1}{4\|\bSigma_0\|_\infty^2}\right)
\right\}\nonumber\\
&\leq4\exp\left\{
4(\tau_n + 2s_n) - C_6\min\left(\frac{\|\bSigma_0\|_\infty^2}{8}, \frac{1}{32}\right)\frac{Mn\eps_n^2}{\|\bSigma_0\|_\infty^2}
\right\}\nonumber.
\end{align}
Notice that for each index set $S(\delta_n)$, the test function $\phi_{S(\delta_n)}$ is only a function of $\bY_n$ through the coordinates $[y_{ij}:i\in[n], j \in S(\delta_n)\cup S_0]$. Hence, $\expect_{\bSigma_{S(\delta_n)}^{(0)}}(\phi_{S(\delta_n)}) = \expect_0(\phi_{S(\delta_n)})$, and for any $p\times p$ covariance matrix $\bSigma$ with $\|\bSigma_{S(\delta_n)} - \bSigma_{S(\delta_n)}^{(0)}\|_\infty >M\eps_n/2$, it holds that $\expect_{\bSigma_{S(\delta_n)}}(1 - \phi_{S(\delta_n)}) = \expect_\bSigma(1 - \phi_{S(\delta_n)})$. Therefore, by aggregating the test functions 
\[\phi_n = \max_{S(\delta_n)\subset[p]:|S(\delta_n)|\leq\tau_n}\phi_{S(\delta_n)},\]
we obtain
\begin{align}
\expect_0(\phi_n) &\leq\sum_{S(\delta_n)\subset[p]:|S(\delta_n)|\leq\tau_n}\expect_{\bSigma_{S(\delta_n)}^{(0)}}(\phi_{S(\delta_n)})\nonumber\\
&\leq 12\sum_{q=0}^{\lfloor\tau_n\rfloor}\frac{p!}{q!(p-q)!}\exp\left\{4(\tau_n + 2s_n) - C_6\min\left(\frac{1}{2},\frac{\|\bSigma_0\|_\infty^2}{\sqrt{2}}\right){\frac{\sqrt{M}n\eps_n^2}{\|\bSigma_0\|_\infty^2}}\right\}\nonumber\\
&\leq 12(\tau_n + 1)\exp(\tau_n\log p)\exp\left\{4(\tau_n + 2s_n) - 
C_6\min\left(\frac{1}{2},\frac{\|\bSigma_0\|_\infty^2}{\sqrt{2}}\right){\frac{\sqrt{M}n\eps_n^2}{\|\bSigma_0\|_\infty^2}}
\right\}\nonumber\\
&\leq 12\exp\left\{\tau_n + \tau_n\log p + 4(\tau_n + 2s_n) - 
C_6\min\left(\frac{1}{2},\frac{\|\bSigma_0\|_\infty^2}{\sqrt{2}}\right){\frac{\sqrt{M}n\eps_n^2}{\|\bSigma_0\|_\infty^2}}
\right\}\nonumber\\
&\leq 12\exp\left\{6(\tau_n\log p + 2s_n) - 
C_6\min\left(\frac{1}{2},\frac{\|\bSigma_0\|_\infty^2}{\sqrt{2}}\right){\frac{\sqrt{M}n\eps_n^2}{\|\bSigma_0\|_\infty^2}}
\right\}\nonumber,
\end{align}
and
\begin{align}
\sup_{\bSigma\in H_1}\expect_\bSigma(1-\phi_n)&\leq \sup_{S(\delta_n)\subset[p]:|S(\delta_n)|\leq \tau_n}\sup_{\left\{\bSigma:\|\bSigma_{S(\delta_n)} - \bSigma_{S(\delta_n)}^{(0)}\|_2>M\eps_n/2\right\}}
\expect_{\bSigma_{S(\delta_n)}}\left(1-\phi_{S(\delta_n)}\right)
\nonumber\\
&\leq 4\exp\left\{4(\tau_n + 2s_n) - C_6\min\left(\frac{\|\bSigma_0\|_\infty^2}{8},\frac{1}{32}\right)\frac{Mn\eps_n^2}{\|\bSigma_0\|_\infty^2}\right\}\nonumber.
\end{align}
The proof is thus completed.
\end{proof}

\section{Additional Technical Results and Proofs} 
\label{sec:remaining_technical_proofs}
\begin{proof}
[\bf Proof of Lemma \ref{lemma:Gamma_subexponential_norm}]
Since $p(w) = (1-\theta)\lambda_0\mathrm{e}^{-\lambda_0w}+\theta\{\lambda_0^{r}/\Gamma(r)\}w^{r-1}\mathrm{e}^{-\lambda w}$, then
\[
\prob(\xi = 1)=(1-\theta)\int_\delta^\infty\lambda_0\mathrm{e}^{-\lambda_0w}\mathrm{d}w+\theta\int_\delta^\infty\frac{\lambda^r}{\Gamma(r)}w^{r-1}\mathrm{e}^{-\lambda w}\mathrm{d}w
=(1-\theta)\mathrm{e}^{-\lambda_0\delta}+\theta\frac{\Gamma(r,\lambda\delta)}{\Gamma(r)}.
\]
Then for any measurable $A\subset\mathbb{R}$, we have
\begin{align}
\prob(w\in A\mid\xi = 1)
& = \frac{1}{\prob(\xi = 1)}\left\{
(1-\theta)\int_{A}\mathbbm{1}(w>\delta)\lambda_0\mathrm{e}^{-\lambda_0w}\mathrm{d}w+\theta\int_A\mathbbm{1}(w>\delta)\frac{\lambda^r}{\Gamma(r)}w^{r-1}\mathrm{e}^{-\lambda w}\mathrm{d}w
\right\}\nonumber\\
& = \int_{A}\mathbbm{1}(w>\delta)\left\{(1-\theta')\lambda_0\mathrm{e}^{-\lambda_0(w-\delta)}\mathrm{d}w+\theta'\frac{\lambda^r}{\Gamma(r,\lambda\delta)}w^{r-1}\mathrm{e}^{-\lambda w}\right\}\mathrm{d}w\nonumber,
\end{align}
where
\[
\theta' = \frac{\theta\Gamma(r,\lambda\delta)/\Gamma(r)}{(1-\theta)\mathrm{e}^{-\lambda_0\delta} + \theta\Gamma(r,\lambda\delta)/\Gamma(r)}\in(0,1).
\]
Therefore, 
\[
p(w\mid\xi = 1) = \left\{(1-\theta')\lambda_0\mathrm{e}^{-\lambda_0(w-\delta)}\mathrm{d}w+\theta'\frac{\lambda^r}{\Gamma(r,\lambda\delta)}w^{r-1}\mathrm{e}^{-\lambda w}\right\}\mathbbm{1}(w>\delta).
\]
Hence we proceed and compute
\begin{align}
\left\{E(w^m\mid\xi = 1)\right\}^{1/m}
& = \left\{(1-\theta')\int_\delta^\infty w^m\lambda_0\mathrm{e}^{-\lambda_0(w-\delta)}\mathrm{d}w+\theta'\frac{\Gamma(r)}{\Gamma(r,\lambda\delta)}\int_\delta^\infty w^m\frac{\lambda^r}{\Gamma(r)}w^{r-1}\mathrm{e}^{-\lambda w}\mathrm{d}w\right\}^{1/m}\nonumber\\
& \leq \left\{\int_0^\infty (w+\delta)^m\lambda_0\mathrm{e}^{-\lambda_0w}\mathrm{d}w+\frac{\Gamma(r)}{\Gamma(r,\lambda\delta)}\int_0^\infty w^m\frac{\lambda^r}{\Gamma(r)}w^{r-1}\mathrm{e}^{-\lambda w}\mathrm{d}w\right\}^{1/m}\nonumber\\
& = \left\{\int_0^\delta (w+\delta)^m\lambda_0\mathrm{e}^{-\lambda_0w}\mathrm{d}w+\int_\delta^\infty (w+\delta)^m\lambda_0\mathrm{e}^{-\lambda_0w}\mathrm{d}w+\frac{\Gamma(r)}{\Gamma(r,\lambda\delta)}\frac{(r + m - 1)!}{(r - 1)!\lambda^m}
\right\}^{1/m}\nonumber\\
& \leq \left\{\int_0^\infty (2\delta)^m\lambda_0\mathrm{e}^{-\lambda_0w}\mathrm{d}w+\int_0^\infty (2w)^m\lambda_0\mathrm{e}^{-\lambda_0w}\mathrm{d}w+\frac{\Gamma(r)}{\Gamma(r,\lambda\delta)}\frac{(r + m - 1)!}{(r - 1)!\lambda^m}
\right\}^{1/m}\nonumber\\
& = \left\{(2\delta)^m+2^m\frac{m!}{\lambda_0^m}+\frac{\Gamma(r)}{\Gamma(r,\lambda\delta)}\frac{(r + m - 1)!}{(r - 1)!\lambda^m}
\right\}^{1/m}
 \leq 2\delta + \frac{2m}{\lambda_0}+\frac{2(r + m)}{\lambda}\nonumber.
\end{align}
Hence
\[
\sup_{m\geq1}\left\{E(w^m)\right\}^{1/m}\leq \sup_{m\geq1}\frac{1}{m}\left\{2\delta+\frac{2m}{\lambda_0}+\frac{2(r+m)}{\lambda}\right\} = 2\delta+\frac{2}{\lambda_0}+\frac{2(r+1)}{\lambda}.
\]
Now we compute the sub-exponential norm. Write
\begin{align}
\sup_{m\geq1}\frac{1}{m}\{E(w^m)\}^{1/m} &= \sup_{m\geq1}\frac{1}{m}\left\{(1-\theta)\int_0^\infty w^m\lambda_0\mathrm{e}^{-\lambda_0w}\mathrm{d}w+\theta\int_0^\infty w^m\frac{\lambda^r}{\Gamma(r)}w^{r-1}\mathrm{e}^{-\lambda w}\mathrm{d}w\right\}^{1/m}
\nonumber\\&
= \sup_{m\geq1}\frac{1}{m}\left\{(1-\theta)\frac{m!}{\lambda_0^m}+\theta\frac{(m+r-1)!}{\lambda^m(r-1)!}\right\}^{1/m} 
\nonumber\\&
\leq \frac{1}{\lambda_0}+\frac{1}{\lambda}\sup_{m\geq1}\theta^{1/m}\left(1+\frac{r}{m}\right)
\leq\frac{1}{\lambda_0}+\frac{r+1}{\lambda}
\nonumber.
\end{align}
If $\theta\leq\mathrm{e}^{-r}$, we can further derive the following result using the fact that $\log(1+ru)\leq ru$ for $u\in(0,1]$:
\begin{align}
\sup_{m\geq1}\frac{1}{m}\{E(w^m)\}^{1/m}&\leq \frac{1}{\lambda_0}+\frac{1}{\lambda}\sup_{m\geq1}\theta^{1/m}\left(1+\frac{r}{m}\right)
\leq \frac{1}{\lambda_0}+\frac{1}{\lambda}\exp\left[\sup_{u\in(0,1]} \left\{-ru+\log(1+ru)\right\}\right]
\leq \frac{1}{\lambda_0}+\frac{1}{\lambda}\nonumber.
\end{align}
\end{proof}

\begin{proof}[\bf{Proof of Lemma \ref{lemma:oracle_test_infinity_norm}}]
Denote the alternative set by $\calH_1 = \{\bSigma: \|\bSigma-\bSigma_0\|_\infty > M\eps\}$ and decompose it as follows: $\calH_1 = \bigcup_{j = 0}^\infty \calH_{1j}$, where
\begin{align}
\calH_{10} &= \left\{\|\bSigma - \bSigma_0\|_\infty > M\eps, 
\|\bSigma\|_\infty\leq (\sqrt{M} + 2)\|\bSigma_0\|_\infty
\right\}\nonumber\\
\calH_{1j} &= \left\{(\sqrt{M} + 2)(M\eps^2)^{-(j - 1)/2}\|\bSigma_0\|_\infty< 
\|\bSigma\|_\infty\leq (\sqrt{M} + 2)(M\eps^2)^{-j/2}\|\bSigma_0\|_\infty
\right\}\nonumber.
\end{align}
For each $\calH_{1j}$, we construct test functions $\phi_{nj}$ as follows:
\begin{align}
\phi_{n0} &= \mathbbm{1}\left\{\left\|\frac{1}{n}\sum_{i = 1}^n\bx_i\bx_i\transpose - \bSigma_0\right\|_\infty > M\eps/2\right\},\nonumber\\
\phi_{nj} &= \mathbbm{1}\left\{\left\|\frac{1}{n}\sum_{i = 1}^n\bx_i\bx_i\transpose\right\|_\infty > \frac{\sqrt{M} + 2}{2}\|\bSigma_0\|_\infty (M\eps^2)^{-(j - 1)/2}\right\}.\nonumber
\end{align}
We first control the type I error. By Lemma \ref{lemma:matrix_concentration_infinity_norm}, 
\[
\expect_0(\phi_{n0}) \leq 4\exp\left\{4d - C_6n\min\left(\frac{M\eps}{2\|\bSigma_0\|_\infty}, \frac{M^2\eps^2}{4\|\bSigma_0\|_\infty^2}\right)\right\}
\leq 4\exp\left(4d - \frac{C_6M^2n\eps^2}{4\|\bSigma_0\|_\infty^2}\right)
\]
since $M\eps < 2\|\bSigma_0\|_\infty$ by assumption. In addition, $M\eps^2\leq \sqrt{M}M\eps^2\leq (M\eps)^2\leq 1$, and hence, for any $j\geq1$,
\begin{align}
\expect_0(\phi_{nj})
&\leq \prob_0\left\{
\left\|\frac{1}{n}\sum_{i = 1}^n\bx_i\bx_i\transpose - \bSigma_0\right\|_\infty + \|\bSigma_0\|_\infty > \frac{\sqrt{M} + 2}{2}\|\bSigma_0\|_\infty (M\eps^2)^{-(j - 1)/2}
\right\}\nonumber\\
&\leq \prob_0\left\{
\left\|\frac{1}{n}\sum_{i = 1}^n\bx_i\bx_i\transpose - \bSigma_0\right\|_\infty > \frac{\sqrt{M}}{2}\|\bSigma_0\|_\infty (M\eps^2)^{-(j - 1)/2}
\right\}\nonumber\\
&\leq 4\exp\left[4d  - C_6n\min\left\{\frac{M(M\eps^2)^{2 - j}}{4}, \frac{\sqrt{M}(M\eps^2)^{1/2 - j/2}}{2}\right\}\right]\nonumber\\
&\leq 4\exp\left(4d  - C_6\frac{M^{1 - j/2}n\eps^{-(j - 1)}}{2}\right)\nonumber.
\end{align}
Next we consider the type II error. For any $\bSigma\in\calH_{10}$, the type II error probability can be upper bounded by
\begin{align}
\expect_\bSigma(1- \phi_{n0})&\leq 
\prob_\bSigma\left\{
\|\bSigma - \bSigma_0\|_\infty - \left\|\frac{1}{n}\sum_{i = 1}^n\bx_i\bx_i\transpose - \bSigma\right\|_\infty\leq M\eps/2
\right\}\nonumber\\
&\leq \prob_\bSigma\left\{\left\|\frac{1}{n}\sum_{i = 1}^n\bx_i\bx_i\transpose - \bSigma\right\|_\infty> M\eps/2
\right\}\nonumber\\
&\leq \prob_\bSigma\left\{\left\|\frac{1}{n}\sum_{i = 1}^n\bx_i\bx_i\transpose - \bSigma\right\|_\infty> \|\bSigma\|_\infty \frac{M\eps}{2(\sqrt{M} + 2)\|\bSigma_0\|_\infty}
\right\}\nonumber\\
&\leq 4\exp\left\{4d - \frac{C_6M^2n\eps^2}{4(\sqrt{M} + 2)^2\|\bSigma_0\|_\infty^2}\right\}\nonumber,
\end{align}
where the last inequality is due to Lemma \ref{lemma:matrix_concentration_infinity_norm} and the assumption $M\eps<2\|\bSigma_0\|_\infty$. For any $\bSigma\in\calH_{1j}$ with $j\geq1$, we estimate the type II error as follows:
\begin{align}
\expect_\bSigma(1 - \phi_{nj})
&\leq \prob_\bSigma\left\{\|\bSigma\|_\infty - \left\|\frac{1}{n}\sum_{i = 1}^n\bx_i\bx_i\transpose - \bSigma\right\|_\infty\leq \frac{\sqrt{M} + 2}{2}\|\bSigma_0\|_\infty(M\eps^2)^{-(j - 1)/2}\right\}\nonumber\\
&\leq \prob_\bSigma\left\{\left\|\frac{1}{n}\sum_{i = 1}^n\bx_i\bx_i\transpose - \bSigma\right\|_\infty> \frac{\sqrt{M} + 2}{2}\|\bSigma_0\|_\infty(M\eps^2)^{-(j - 1)/2}\right\}\nonumber\\
&= \prob_\bSigma\left\{\left\|\frac{1}{n}\sum_{i = 1}^n\bx_i\bx_i\transpose - \bSigma\right\|_\infty> \frac{(M\eps^2)^{1/2}}{2}(\sqrt{M} + 2)(M\eps^2)^{-j/2}\|\bSigma_0\|_\infty\right\}\nonumber\\
&\leq\prob_\bSigma\left\{\left\|\frac{1}{n}\sum_{i = 1}^n\bx_i\bx_i\transpose - \bSigma\right\|_\infty> \frac{(M\eps^2)^{1/2}}{2}\|\bSigma\|_\infty\right\}\nonumber\\
&\leq 4\exp\left(4d - \frac{C_6Mn\eps^2}{4}\right)\nonumber
\end{align}
since $M\eps^2\leq M\eps\leq 1$. Now we aggregate the individual tests by taking $\phi_n = \sup_{j\geq 0}\phi_{nj}$. Then the overall type I error probability can be bounded by
\begin{align}
\expect_0(\phi_n)&\leq \sum_{j = 0}^\infty \expect_0(\phi_{nj})\nonumber\\
&\leq 4\exp\left(4d - \frac{C_6M^2n\eps^2}{4\|\bSigma_0\|_\infty^2}\right) + \sum_{j = 1}^\infty 4\exp\left(4d - C_6\frac{M^{1 - j/2}n\eps^{-(j - 1)}}{2}\right)\nonumber\\
& = 4\exp\left(4d - \frac{C_6M^2n\eps^2}{4\|\bSigma_0\|_\infty^2}\right)
+ 4\exp\left(4d\right)\sum_{j = 1}^\infty \exp\left\{ - \frac{C_6Mn\eps}{2}\left(\frac{1}{\sqrt{M}\eps}\right)^j\right\}\nonumber\\
&\leq 4\exp\left(4d - \frac{C_6M^2n\eps^2}{4\|\bSigma_0\|_\infty^2}\right) 
+ 4\exp\left(4d\right)\sum_{j = 1}^\infty \exp\left\{ - j\frac{C_6Mn\eps}{2}\left(\frac{1}{\sqrt{M}\eps}\right)\right\}\nonumber\\
&= 4\exp\left(4d - \frac{C_6M^2n\eps^2}{4\|\bSigma_0\|_\infty^2}\right)
+ 8\exp\left(4d - \frac{C_6\sqrt{M}n}{2}\right)\nonumber,
\end{align}
since $M\geq \{(2\log 2)/C_6\}^2$, where the simple inequality $x^j \geq jx$ for all $x\geq 1$ is applied. Furthermore, the overall type II error probability can also be bounded:
\begin{align}
\sup_{\bSigma\in\calH_1}\expect_\bSigma(1 - \phi_n)
& = \sup_{j\geq 0}\sup_{\bSigma\in\calH_{1j}}\expect_\bSigma(1 - \phi_n)
 = \sup_{j\geq 0}\sup_{\bSigma\in\calH_{1j}}\expect_\bSigma\inf_{j\geq 0}(1 - \phi_{jn})
 \leq \sup_{j\geq 0}\sup_{\bSigma\in\calH_{1j}}\expect_\bSigma(1 - \phi_{jn})\nonumber\\
 &\leq \sup_{j\geq 0}\sup_{\bSigma\in\calH_{1j}}4\exp\left[
 4d - \frac{C_6Mn\eps^2}{4}\min\left\{1, \frac{M}{(\sqrt{M} + 2)^2\|\bSigma_0\|_\infty^2}\right\}
 \right]\nonumber\\
 &\leq 4\exp\left\{
 4d - \frac{C_6Mn\eps^2}{4}\min\left(1, \frac{1}{4\|\bSigma_0\|_\infty^2}\right)
 \right\}\nonumber
\end{align}
since $M\geq 4$. The proof is thus completed.
\end{proof}

\begin{lemma}
\label{lemma:matrix_concentration_infinity_norm}
Let $\bx_1,\ldots,\bx_n\sim\mathrm{N}_d(\zero_d,\bSigma)$ independently, where $\bSigma\in\mathbb{R}^{d\times d}$. Then there exists an absolute constant $C_6 > 0$, such that for any $t>0$,
\[
\prob\left(\left\|\frac{1}{n}\sum_{i=1}^n\bx_i\bx_i\transpose - \bSigma\right\|_\infty>t\|\bSigma\|_\infty\right)\leq 4\exp\{4d - C_6n\min(t, t^2)\}
\]
\end{lemma}
\begin{proof}[\bf{Proof of Lemma \ref{lemma:matrix_concentration_infinity_norm}}]
By definition, 
\[
\|\bA\|_\infty = \sup_{\|\bv\|_\infty = 1}\|\bA\bv\|_\infty = \max_{j\in[p]}\sup_{\|\bv\|_\infty = 1}\mathbf{e}_j\transpose\bA\bv,
\]
where $\mathbf{e}_j$ is the unit vector along the $j$th coordinate direction. Now let $S_\infty^{d-1}(1/2)$ be an $1/2$-net of the $\ell_\infty$-sphere in $\mathbb{R}^d$ ($\{\bv\in\mathbb{R}^d:\|\bv\|_\infty = 1\}$) with minimum cardinality. Then for each $\bv$ with $\|\bv\|_\infty = 1$, there exists some $\bv'\in S_\infty^{d-1}(1/2)$ such that $\|\bv-\bv'\|_\infty<1/2$. Therefore,
\begin{align}
\|\bA\|_\infty &= \max_{j\in[d]}\sup_{\|\bv\|_\infty = 1}\mathbf{e}_j\transpose\bA\bv
\leq \max_{j\in[d]}\sup_{\|\bv\|_\infty = 1}\left\{\mathbf{e}_j\transpose\bA(\bv - \bv') + \mathbf{e}_j\transpose\bA\bv'\right\}\nonumber\\
&\leq \max_{j\in[d]}\sup_{\|\bv\|_\infty = 1}\mathbf{e}_j\transpose\bA(\bv-\bv') + \max_{j\in[d]}\sup_{\bv\in S_\infty^{d-1}(1/2)}\mathbf{e}_j\transpose\bA\bv\nonumber\\
&\leq \frac{1}{2}\|\bA\|_\infty + \max_{j\in[d]}\sup_{\bv\in S_\infty^{d-1}(1/2)}\mathbf{e}_j\transpose\bA\bv\nonumber,
\end{align}
and hence,
\[
\|\bA\|_\infty\leq 2\max_{j\in[d]}\sup_{\bv\in S_\infty^{d-1}(1/2)}\mathbf{e}_j\transpose\bA\bv.
\]
Denote 
\[\mathbf{E} = \frac{1}{n}\sum_{i = 1}^n\bx_i\bx_i\transpose - \bSigma.\] 
Now we apply the union bound to derive
\begin{align}
\prob\left(\left\|\frac{1}{n}\sum_{i = 1}^n\bx_i \bx_i\transpose - \bSigma\right\|_\infty>t\|\bSigma\|_\infty\right)
&= \prob\left[\bigcup_{j\in[d]}\bigcup_{\bv\in S_\infty^{d-1}(1/2)}\left\{\mathbf{e}_j\transpose\mathbf{E}\bv > \frac{t}{2}\|\bSigma\|_\infty\right\}\right]\nonumber\\
&\leq \sum_{j = 1}^d\sum_{\bv\in S_\infty^{d - 1}(1/2)}\prob\left\{\mathbf{e}_j\transpose
\left(\frac{1}{n}\sum_{i = 1}^n\bx_i\bx_i\transpose - \bSigma\right)
\bv > \frac{t}{2}\|\bSigma\|_\infty\right\}\nonumber\\
& = \sum_{j = 1}^d\sum_{\bv\in S_\infty^{d - 1}(1/2)}\prob\left\{\frac{1}{n}\sum_{i = 1}^n(\mathbf{e}_j\transpose\bx_i)(\bv\transpose\bx_i) - \mathbf{e}_j\transpose\bSigma\bv > \frac{t}{2}\|\bSigma\|_\infty\right\}\nonumber.
\end{align}
Observe that
\[
\begin{bmatrix*}
\bv\transpose\bx_i\\\mathbf{e}_j\transpose\bx_i
\end{bmatrix*} \sim\mathrm{N}_2\left(
\begin{bmatrix*}0\\0\end{bmatrix*},
\begin{bmatrix*}
\bv\transpose\bSigma\bv & \bv\transpose\bSigma\mathbf{e}_j\\
\mathbf{e}_j\transpose\bSigma\bv & \mathbf{e}_j\transpose\bSigma\mathbf{e}_j
\end{bmatrix*}
\right),
\]
then we can decompose $(\mathbf{e}_j\transpose\bx_i)(\bv\transpose\bx_i)$ by projecting $\bv\transpose\bx_i$ onto the space spanned by $\mathbf{e}_j\transpose\bx_i$ as follows:
\begin{align}
(\mathbf{e}_j\transpose\bx_i)(\bv\transpose\bx_i)
&=\left(\bv\transpose\bx_i - \frac{\mathbf{e}_j\transpose \bSigma\bv }{\mathbf{e}_j\transpose\bSigma\mathbf{e}_j}\mathbf{e}_j\transpose\bx_i\right)(\mathbf{e}_j\transpose\bx_i) + 
\frac{\mathbf{e}_j\transpose \bSigma\bv }{\mathbf{e}_j\transpose\bSigma\mathbf{e}_j}(\mathbf{e}_j\transpose\bx_i)^2\nonumber\\
&\overset{d}{=} \sqrt{\mathbf{e}_j\transpose\bSigma\mathbf{e}_j \bv\transpose\bSigma\bv - (\mathbf{e}_j\transpose\bSigma\bv)^2}  \zeta_{i1}\zeta_{i2} + \mathbf{e}_j\transpose\bSigma\bv \zeta_{i2}^2,\nonumber
\end{align}
where $\zeta_{i1}$ and $\zeta_{i2}$ are independent $\mathrm{N}(0, 1)$ random variables, $i = 1,\ldots,n$, and $\overset{d}{=}$ indicates the equality in distribution. Hence,
\begin{align}
&\prob\left\{\frac{1}{n}\sum_{i = 1}^n(\mathbf{e}_j\transpose\bx_i)(\bv\transpose\bx_i) - \mathbf{e}_j\transpose\bSigma\bv > \frac{t}{2}\|\bSigma\|_\infty\right\}\nonumber\\
&\quad \leq  \prob\left\{
\sqrt{\mathbf{e}_j\transpose\bSigma\mathbf{e}_j \bv\transpose\bSigma\bv - (\mathbf{e}_j\transpose\bSigma\bv)^2}\left|\frac{1}{n}\sum_{i = 1}^n\zeta_{i1}\zeta_{i2}\right| + 
|\mathbf{e}_j\transpose\bSigma\bv|\left|\frac{1}{n}\sum_{i = 1}^n \left(\zeta_{i2}^2 - 1\right)\right| > \frac{t}{2}\|\bSigma\|_\infty\right\}\nonumber\\
&\quad\leq \prob\left\{
\|\bSigma\|_\infty\left|\frac{1}{n}\sum_{i = 1}^n\zeta_{i1}\zeta_{i2}\right| + 
\|\bSigma\|_\infty\left|\frac{1}{n}\sum_{i = 1}^n \left(\zeta_{i2}^2 - 1\right)\right| > \frac{t}{2}\|\bSigma\|_\infty\right\}\nonumber\\
&\quad\leq \prob\left\{\left|\frac{1}{n}\sum_{i = 1}^n\zeta_{i1}\zeta_{i2}\right| > \frac{t}{4}\right\} + 
\prob\left\{\left|\frac{1}{n}\sum_{i = 1}^n \left(\zeta_{i2}^2 - 1\right)\right| > \frac{t}{4}\right\}\nonumber.
\end{align}
Since $\zeta_{i1}\zeta_{i2}$ and $\zeta_{i2}^2 - 1$ are mean-zero sub-exponential random variables, it follows from the large-deviation inequality for sub-exponential random variables that
\[
\prob\left\{\left|\frac{1}{n}\sum_{i = 1}^n\zeta_{i1}\zeta_{i2}\right| > \frac{t}{4}\right\} + 
\prob\left\{\left|\frac{1}{n}\sum_{i = 1}^n \left(\zeta_{i2}^2 - 1\right)\right| > \frac{t}{4}\right\}\leq 
4\exp\left\{-C_6n\min(t, t^2)\right\}
\]
for some absolute constant $C_6>0$. It suffices to bound $|S_\infty^{d-1}(1/2)|$. Since
\[
|S_\infty^{d-1}(1/2)| 
= \calN(1/2, \{\bv\in\mathbb{R}^d:\|\bv\|_\infty = 1\}, \|\cdot\|_\infty)
\leq \calN(1/2, \{\bv\in\mathbb{R}^d:\|\bv\|_\infty \leq 1\}, \|\cdot\|_\infty)\leq 6^d,
\]
it follows that
\begin{align}
\prob\left(\left\|\frac{1}{n}\sum_{i = 1}^n\bx_i \bx_i\transpose - \bSigma\right\|_\infty>t\|\bSigma\|_\infty\right)
&\leq \sum_{j = 1}^d\sum_{\bv\in S_\infty^{d - 1}(1/2)}\prob\left\{\frac{1}{n}\sum_{i = 1}^n(\mathbf{e}_j\transpose\bx_i)(\bv\transpose\bx_i) - \mathbf{e}_j\transpose\bSigma\bv > \frac{t}{2}\|\bSigma\|_\infty\right\}
\nonumber\\&
\leq 4d\exp\{d\log 6 - C_6n\min(t, t^2)\}\nonumber\\
&\leq 4\exp\{4d - C_6n\min(t, t^2)\}\nonumber,
\end{align}
and the proof is thus completed. 
\end{proof}


\end{document}